\documentclass[11pt]{article}
\usepackage{amsmath, amssymb, amsthm,complexity,comment}
\usepackage{hyperref} %how to use
\usepackage{cleveref}
\usepackage{siunitx}
\usepackage{tikz}
\usepackage{thmtools}
\DeclareMathOperator{\surp}{surp}
\DeclareMathOperator{\minsurp}{minsurp}
\DeclareMathOperator{\val}{val}
\DeclareMathOperator{\shad}{shad}
\DeclareMathOperator{\mindeg}{mindeg}
\DeclareMathOperator{\LPVC}{LPVC}
\DeclareMathOperator{\maxdeg}{maxdeg}
\DeclareMathOperator{\girth}{girth}
\DeclareMathOperator{\bonus}{bonus}
\DeclareMathOperator{\dist}{dist}
\DeclareMathOperator{\codeg}{codeg}
\DeclareMathOperator{\lpopt}{\lambda}

\usepackage{float}
\usepackage{enumitem}
\usepackage{subfigure}
\usepackage{algorithmic}
\usepackage{graphicx}
\usepackage{xspace}
\usepackage[ruled,vlined,linesnumbered]{algorithm2e}
\usepackage{multicol}

\usepackage[margin=1in]{geometry}
\usepackage{authblk}
\usepackage{listings}

\crefname{observation}{Observation}{Observations}

\hypersetup{
    colorlinks,%
    citecolor=black,%
    filecolor=black,%
    linkcolor=black,%
    urlcolor=black
}
\title{A faster algorithm for Vertex Cover parameterized by solution size\footnote{This is an extended version of a paper appearing in the 41st international Symposium of Theoretical Aspects of Computer Science (STACS) 2024}}  

\author[1]{David G. Harris}
\affil[1]{Department of Computer Science, \authorcr University of Maryland, College Park. \authorcr davidgharris29@gmail.com}
\author[2]{N. S. Narayanaswamy}
\affil[2]{Department of Computer Science and Engineering, \authorcr Indian Institute of Technology Madras, India. \authorcr swamy@cse.iitm.ac.in}
\date{}
\theoremstyle{plain}
\newtheorem{theorem}{Theorem}
\newtheorem{lemma}[theorem]{Lemma}

\newtheorem{proposition}[theorem]{Proposition}

\newtheorem{observation}[theorem]{Observation}

\usepackage{graphicx}
\usepackage{multirow}
\usepackage{comment}
\usepackage{arydshln}

\newcommand{\aFour}{0.59394}         \newcommand{\bFour}{0.039361}
\newcommand{\aFiveOne}{0.496708}  \newcommand{\bFiveOne}{0.064301}
\newcommand{\aFiveTwo}{0.437086}  \newcommand{\bFiveTwo}{0.080751}
\newcommand{\aFiveThree}{0.379406} \newcommand{\bFiveThree}{0.097471}
\newcommand{\aSixOne}{0.254135}     \newcommand{\bSixOne}{0.137360}
\newcommand{\aSixTwo}{0.202348}    \newcommand{\bSixTwo}{0.154382}
\newcommand{\aSixThree}{0.166944} \newcommand{\bSixThree}{0.166214}
     
\newcommand{\dFour}{1.21131} 
\newcommand{\dFive}{1.24394}
\newcommand{\dSix}{1.25214}
\newcommand{\dSeven}{1.25284}

\usepackage{xspace}
\usepackage{boxedminipage}

\newcommand{\defbox}[2]{
\vspace{1mm}
\noindent\fbox{
\begin{minipage}{0.96\textwidth}
{\bf #1}
#2
\end{minipage}
}
\vspace{1mm}
}

\begin{document}
\maketitle
\begin{abstract}
\noindent
We describe a new algorithm for vertex cover  with runtime $O^*(\dSeven^k)$, where $k$ is the size of the desired solution and $O^*$ hides polynomial factors in the input size. This improves over the previous runtime of $O^*(1.2738^k)$ due to Chen, Kanj, \& Xia (2010) standing for more than a decade. The key to our algorithm is to use a measure which simultaneously tracks $k$ as well as the optimal value $\lpopt$ of the vertex cover LP relaxation.   This allows us to make use of prior algorithms for Maximum Independent Set in bounded-degree graphs and Above-Guarantee Vertex Cover.

The main step in the algorithm is to branch on high-degree vertices, while ensuring that both $k$ and $\mu = k - \lpopt$ are decreased at each step. There can be local obstructions in the graph that prevent $\mu$ from decreasing in this process; we develop a number of novel branching steps to handle these situations.
\end{abstract}

\section{Introduction}
For an undirected graph $G=(V,E)$, a  subset $C \subseteq V$ is called a {\em vertex cover} if every edge has at least one endpoint in $C$. It is closely related to an \emph{independent set}, since if $C$ is an inclusion-wise minimal vertex cover then $V \setminus C$ is an inclusion-wise maximal independent set, and vice-versa. Finding the size of the smallest vertex cover is a classic NP-complete problem \cite{garey}.
In particular, $G$ has a vertex cover of size at most $k$ if and only if it has an independent set of size at least $n - k$.   

There is natural LP formulation for vertex cover which we denote by $\LPVC(G)$:
\begin{eqnarray*}
\text{minimize } &\sum_{v \in V} \theta(v) \\
\text{subject to } & \theta(u) + \theta(v) \geq 1 & \text{for all edges $(u,v) \in E$} \\
& \theta(v) \in [0,1] & \text{for all vertices $v \in V$}
\end{eqnarray*}
The optimal solution to $\LPVC(G)$, denoted $\lambda(G)$ or just $\lambda$ if $G$ is clear from context, is a lower bound on the size of a minimum vertex cover of $G$. We also define $$
\mu(G) = k - \lambda(G),
$$ i.e. the gap between solution size and LP lower bound. This linear program has  remarkable algorithmic properties \cite{nt-vc74, nt-vc75, fpt-lp}. For instance, an optimum basic solution is half-integral and can be found efficiently by a network flow computation.  

Fixed-parameter tractable (FPT) and exact algorithms  explore a landscape of parameters to understand the complexity of different problems \cite{nied,fg,downey,DBLP:books/sp/CyganFKLMPPS15,exact-algo}.   \textsc{Vertex Cover} was one of the first studied FPT problems with respect to the parameter $k$, the optimal solution size. Building on a long line of research, this culminated in an algorithm with $O^*(1.2738^k)$ runtime \cite{vc-best}.\footnote{Throughout, we write $O^*( T )$ as shorthand for $T \cdot \poly(n)$.} This record has been standing for more than a decade. Assuming the Exponential Time Hypothesis, no algorithm with runtime $O^*(2^{o(k)})$ is possible \cite{DBLP:books/sp/CyganFKLMPPS15}.

Many different kinds of algorithms have been developed for the Vertex Cover problem.  We cannot summarize them fully here;  we describe some of the most relevant works for our own paper.  

An important variant is \textsc{Above-Guarantee Vertex Cover (AGVC)}, where the parameter is the difference between $k$ and various lower bounds on vertex cover size  \cite{maharam,maharamsik, fpt-lp,vc2lpmm, kellerhals2022vertex}. Of these, the parameter $\mu$ plays a particularly important role in this paper. We quote the following main result:
\begin{theorem}[\cite{fpt-lp}]
\label{agvc-thm}
 Vertex cover can be solved in time  $O^*(2.3146^{\mu(G)})$.
\end{theorem}
\noindent 

Another natural choice is to measure runtime in terms of the graph size $n$. In this setting, the problem is more commonly referred to as Maximum Independent Set (MaxIS). Xiao and Nagamochi have developed a number of algorithms in this vein. These algorithms, and in particular their performance on bounded-degree graphs, will also play a crucial role in our analysis. We refer to the algorithm targeted for graphs of maximum degree $\Delta$ by the \emph{MaxIS-$\Delta$} algorithm.
 \begin{theorem}
 \label{mis3thm}
  MaxIS-3 can be solved with runtime $O^*(1.083506^n)$ by \cite{xiao2013confining}.\footnote{This runtime is not claimed directly in \cite{xiao2013confining}, see \cite{xiao2017refined} instead.}

 MaxIS-4 can be solved with runtime $O^*(1.137595^{n})$ by \cite{xiao2017refined}.

 MaxIS-5 can be solved with runtime  $O^*(1.17366^{n})$ by \cite{xiao2016exact}.\footnote{This runtime is not claimed directly in \cite{xiao2016exact}, see \cite{XiaoN17} instead.}

 MaxIS-6 can be solved with runtime  $O^*(1.18922^{n})$ by \cite{XiaoN17}.
 
 MaxIS-7 can be solved with runtime  $O^*(1.19698^{n})$ by \cite{XiaoN17}.

MaxIS in graphs of arbitrary degree can be solved with runtime $O^*(1.19951^{n})$ by \cite{XiaoN17}.
  \end{theorem}

In addition,  \cite{tsur2018above} describes algorithms in graphs of maximum degree $3$ and $4$ with runtime respectively $O^*(1.1558^k)$ and $O^*(1.2403^k)$ respectively. 

 \subsection{Outline of our results}
 The main result of this paper is an improved algorithm for vertex cover parameterized by $k$. 
 \begin{theorem}
 \label{main-sum-thm}
 There is an algorithm for \textsc{VertexCover} with runtime $O^*(\dSeven^k)$.    Moreover, depending on the maximum vertex degree of the graph $G$,  better bounds can be shown:
 
 If $\maxdeg(G)  \leq 3$, we get runtime $O^*( 1.14416^k )$.
 
 If $\maxdeg(G) \leq 4$, we get runtime $O^*( \dFour^k )$.
 
 If $\maxdeg(G) \leq 5$, we get runtime $O^*( \dFive^k )$.

 If $\maxdeg(G) \leq 6$, we get runtime $O^*( \dSix^k )$.

All of these algorithms use polynomial space and are deterministic.
\end{theorem}

The FPT algorithms for vertex cover, including our new algorithm, are built out of recursive branching while tracking some ``measure'' of the graph.  Our main new idea is to use a measure which is a piecewise-linear function of $k$ and $\mu$.  To illustrate, consider branching on whether some degree-$r$ vertex $u$ is in the cover: we recurse on the subproblem $G_1 = G - u$ with a solution of size $k_1 = k-1$ and on the subproblem $G_0 = G - u - N(u)$ with the solution size $k_0 = k - r$.  We must also show that $\mu$ is significantly reduced in the two subproblems.

Suppose that, wishfully speaking, $\theta = \vec{\tfrac{1}{2}}$ remains the optimal solution to $\LPVC(G_0)$ and $\LPVC(G_1)$.  This is what should happen in a ``generic'' graph. Then $\lpopt(G_0) = \frac{|V(G_0)|}{2} = \frac{n - r - 1}{2}$ and $\lpopt(G_1) = \frac{|V(G_1)|}{2} = \frac{n-1}{2}$. This implies that $\mu$ is indeed significantly decreased:
 $$
 \mu(G_0) = \mu(G) - (r-1)/2, \qquad \mu(G_1) = \mu(G) - 1/2
 $$
 
But suppose, on the other hand, that $\vec{\tfrac{1}{2}}$ is not the optimal solution to $\LPVC(G_0)$ or $\LPVC(G_1)$. For a concrete example, suppose the neighborhood of $x$ in $G$ is a subset of that of $u$, so $G_0$ has an isolated vertex $x$ and an optimal solution $\theta^*$ to $\LPVC(G_0)$ would set $\theta^*(x) = 0$. In this situation, we develop an alternate branching rule for $G$: rather than branching on $u$ itself, we branch on the two subproblems where (i) both $u,x$ are in the cover and (ii) $x$ is not in the cover. (See Figure~\ref{fig4}.)

\begin{figure}[H]
\vspace{0.8in}
\begin{center}
\hspace{0.5in}
\includegraphics[trim = 0.5cm 23.5cm 9cm 8cm,scale=0.5,angle = 0]{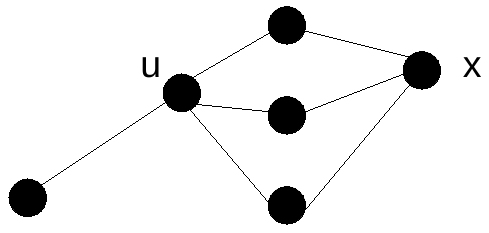}
\caption{\label{fig4} In the subgraph $G_0$, the vertex $x$ becomes isolated.}
\end{center}
\end{figure}
\vspace{-0.2in}
We emphasize that this is just one example of our branching analysis. We need to handle more general situations where $\vec{\tfrac{1}{2}}$ is not the optimal solution to $\LPVC(G_0)$ or $\LPVC(G_1)$, which can require more complex and fundamentally new branching rules.

In Section~\ref{sec:simplebranching}, we use these ideas for a relatively simple algorithm with runtime $O^*(1.2575^k)$, along with algorithms for bounded-degree graphs. This is already much better than \cite{vc-best}. To improve further, and achieve \Cref{main-sum-thm}, we need an additional idea: instead of branching on an arbitrary high-degree vertex, we branch on a \emph{targeted} vertex to create additional simplifications in the graph. This is much more complex. For example, when we are branching on degree-4 vertices, it is favorable to choose one that has a degree-$3$ neighbor. This will  create a degree-$2$ vertex in subproblem $G - u$, leading to further simplifications of the graph. See Figure~\ref{fig88}. 

\begin{figure}[H]
\vspace{0.95in}
\begin{center}
\hspace{0.5in}
\includegraphics[trim = 0.5cm 22.0cm 9cm 5cm,scale=0.5,angle = 0]{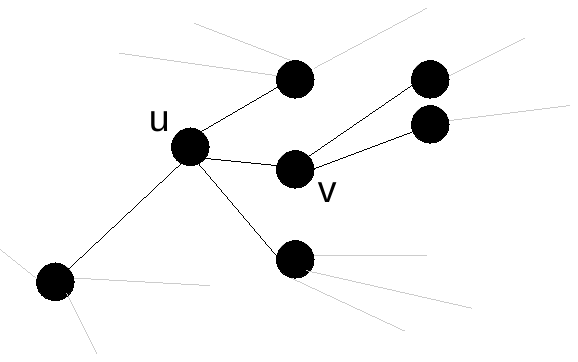}
\caption{\label{fig88} In the subgraph $G - u$, the vertex $v$ has degree two.}
\end{center}
\end{figure}
\vspace{-0.3in}

\subsection{Notation} 
We consider throughout a simple, undirected, unweighted graph $G = (V,E)$, and we write $n = |V|$. 
We write an ordered pair $\langle G, k \rangle$ for a graph $G$ where we need to decide if there is a vertex cover of size at most $k$, and we say it is \emph{feasible} if such a vertex cover exists.  We say $C$ is a \emph{good cover} if it is a vertex cover of size at most $k$.

 For a subset $S$ of $V$, the subgraphs of $G$ induced by $S$ and $V \setminus S$ are denoted by $G[S]$ and $G-S$, respectively.  We write $u \sim v$ if $(u,v) \in E$ and $u \nsim v$ otherwise.  For vertex sets $X, Y$, we write $X \sim Y$ if there exists $x \in X, y \in Y$ with $x \sim y$. 

For a vertex $u$, the neighborhood $N_G(u)$ is the set $\{v \in V(G): u \sim v\}$ and the closed neighborhood $N_G[u]$ is the set $N_G(u) \cup \{u\}$. Extending this notation, for a vertex set $S \subseteq V(G)$, we write $N_G(S) = \bigl( \bigcup_{v \in S} N_G(v) \bigr) \setminus S$ and $N_G[S]=N_G(S) \cup S = \bigcup_{v \in S} N_G[v]$.  For readability, we sometimes write $N(x,y)$ or $N[x,y]$ as shorthand for $N( \{x,y \})$ or $N[ \{x,y \} ]$.

The degree of vertex $v$, denoted by $\deg_G(v)$, is the size of $N_G(v)$. We call a vertex of degree $i$ an \emph{$i$-vertex}; for a vertex $v$, we refer to a neighbor $u$ which has degree $j$ as a  \emph{$j$-neighbor of $v$}.  An isolated vertex is a 0-vertex.  We say a vertex is \emph{subquartic} if it has degree in $\{1,2,3,4 \}$ and \emph{subcubic} if it has degree in $\{1,2,3 \}$, and \emph{subquadratic} if it has degree in $\{1,2 \}$. 

We write $\maxdeg(G)$ and $\mindeg(G)$ for the minimum and maximum vertex degrees in $G$, respectively.  The graph is \emph{$r$-regular} if $\maxdeg(G) = \mindeg(G) = r$. We write $V_i$ for the set of degree-$i$ vertices and $n_i = |V_i|$.
 
We say that a pair of vertices $u, v$ \emph{share neighbor $y$} if $y \in N(u) \cap N(v)$. We denote by $\codeg(u, v) = |N(u) \cap N(v)|$ the number of vertices shared by $u, v$.

An \emph{independent set} (or, for brevity, an \emph{indset}) denotes a set $X \subseteq V$ where no two vertices in $X$ are adjacent, i.e. $X \not \sim X$.

An $\ell$-cycle denotes a set of $\ell$ vertices $x_1, \dots, x_{\ell}$ with a path $x_1, x_2, \dots, x_{\ell}, x_1$. 
\section{Preliminaries}
\label{sec:prelim}
Our algorithm will heavily use the LP and its properties, which are closely related to properties of indsets.  We review some basic facts here. Much of this material is standard, see e.g. \cite{fpt-lp}; we defer some proofs to Appendix~\ref{sec:prelimproof} for completeness.

For an indset $I$ in $G$, we define the \emph{surplus} by $$
\surp_G(I) = |N_G(I)|-|I|.
$$

 We define  $\minsurp(G)$ to be the minimum value of $\surp_G(I)$ over \emph{non-empty} indsets $I$ in $G$, and a \emph{min-set} to be any non-empty indset $I$ with $\surp_G(I) = \minsurp(G)$.  Since this plays a central role in our analysis, we define the \emph{shadow} of a vertex set $X$, denoted $\shad_G(X)$, to be $$
\shad_G(X) = \minsurp(G - X).
$$ 

We write $\shad_G(x_1, \dots, x_{\ell}, X_1, \dots, X_r)$ as shorthand for $\shad_G( \{x_1, \dots, x_{\ell} \} \cup X_1 \cup \dots \cup X_r ) = \minsurp( G - \{x_1, \dots, x_{\ell} \} - (X_1 \cup \dots \cup X_r) )$.   Since this comes up in a number of places, we define $\minsurp^0(G) = \min\{0, \minsurp(G) \}$.  If $G$ is clear from context, we write just $\surp(I), N(I)$, etc.  

To explain the notation, consider the following diagram:

\begin{figure}[H]
\vspace{1.1in}
\begin{center}
\hspace{0.5in}
\includegraphics[trim = 1cm 21.5cm 9cm 8cm,scale=0.5,angle = 0]{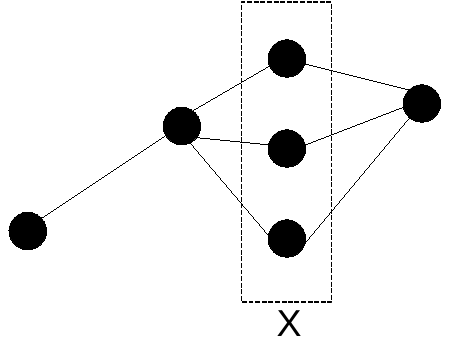}
\vspace{-0.2in}
\caption{In the subgraph $G - X$, there is an isolated vertex, hidden in the ``shadow'' of $X$. Here $\shad(X) = -1$.}
\end{center}
\end{figure}
\vspace{-0.2in}

\begin{proposition}
\label{lambda-surp}
There holds $\lambda(G) = \tfrac{1}{2} (|V| + \minsurp^0(G)).$ 
\end{proposition}

\begin{proposition}
\label{lambda-surpxx}
Given any indset $I$ and vertex set $X \supseteq I$, one can efficiently find an indset which attains the minimum value of $\surp_G(J)$, over all indsets $J$ such that $I \subseteq J \subseteq X$.
\end{proposition}

Using \Cref{lambda-surpxx}, we can efficiently find min-sets of various types. For example, we can determine $\minsurp(G)$ by
applying \Cref{lambda-surpxx} with $X = V, I = \{ v\}$ for all choices of vertex $v$; the minimum over all these values gives a (non-empty) min-set $I$ of $G$. Similarly, we can find inclusion-wise minimal or maximal min-sets by repeatedly applying \Cref{lambda-surpxx} to try to include or exclude each vertex from a given min-set. For the rest of the paper, we will not discuss the computational aspects of finding min-sets that are needed in our algorithm.

\medskip

We define a \emph{critical-set} to be a non-empty indset $I$ with $\surp_{G}(J) \geq \surp_{G}(I)$ for all non-empty subsets $J \subseteq I$.  Clearly, any min-set or any singleton set is a critical-set. 

\begin{lemma}
\label{s-prop1}
For a critical-set $I$, there is a good cover $C$ with either $I \subseteq C$ or $I \cap C = \emptyset$.  Moreover:

If $\surp(I) \leq 0$, there is a good cover $C$ with $I \cap C = \emptyset$. 

 If $\surp(I) = 1$, there is a good cover $C$ with either $N[I] \cap C = I$ or $N[I] \cap C = N(I)$.
\end{lemma}

\subsection{Preprocessing rules and reductions in $k$} 
Our algorithm uses  three major preprocessing rules.  The first two are defined in terms of critical sets of small surplus. 
 
\begin{center}
\defbox{Preprocessing Rule 1 (P1):}{Given a critical-set $I$ with $\surp_{G}(I) \leq 0$, form graph $G'$ with $k' = k - |N(I)|$ by deleting $N[I]$. }
\end{center}

\vspace{-0.18in}

\begin{center}
\defbox
{Preprocessing Rule 2 (P2):}
{Given a critical-set $I$ with $\surp_{G}(I) = 1$, form graph $G'$ with $k' = k - |I|$  by deleting $N[I]$  and adding a new vertex $y$ adjacent to $N(N(I))$.}
\end{center}

The final rule is based on a graph structure called a \emph{funnel}. This structure consists of a vertex $u$ with a neighbor $x$ such that $G[ N(u) \setminus \{ x \} ]$ forms a clique. We call $x$ the \emph{out-neighbor} of $u$. For example, if a degree-3 vertex $u$ has a triangle with neighbors $t_1, t_2$, and one other vertex $x$, then $u$ is a funnel with out-neighbor $x$; since this comes up so often, we refer to this as a \emph{3-triangle}. The following result is well-known:
\begin{observation}
For a funnel $u$ with out-neighbor $x$, there exists a good cover $C$ satisfying either (i) $u \notin C$ or (ii) $x \notin C$ and $|C \setminus N(u)| = 1$.
\end{observation}
This leads to the following simplification rule:
 
\begin{center}
\defbox
{Preprocessing Rule 3 (P3):}{Given a funnel $u$ with out-neighbor $x$, form graph $G'$ with $k' = k-1 - \codeg(u,x)$ by removing vertices the vertices in $N[u] \cap N[x]$ and adding edges between all vertices in $N(u) \setminus N[x]$ to all vertices in $N(x) \setminus N[u]$.}
\end{center}

 See Figure~\ref{fig2} for an illustration of rule (P3) applied to a funnel.
\begin{figure}[H]
\vspace{0.75in}
\begin{center}
\hspace{0.4in}
\includegraphics[trim = 0.5cm 21.5cm 9cm 5cm,scale=0.43,angle = 0]{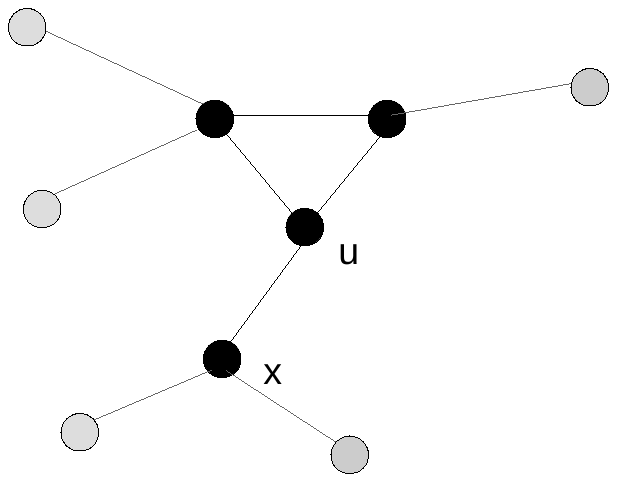}
\hspace{0.3in}
\includegraphics[trim = 0.5cm 21.5cm 9cm 5cm,scale=0.42,angle = 0]{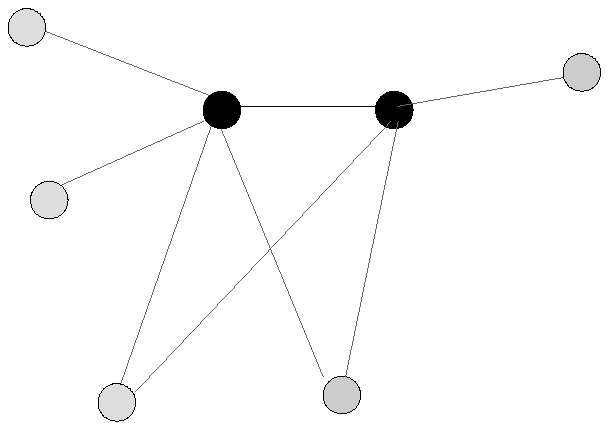}
\vspace{0.4 in}
\caption{
\label{fig2}A funnel (here, a 3-triangle) $u$  before (left) and after (right) applying P3}
\end{center}
\end{figure}

\vspace{-0.2in}

These rules can be applied to any eligible sets in any order. Note that if a vertex $v$ has degree zero or one, then applying (P1) to $\{v \}$ removes $v$ and its neighbor (if any). If a vertex $v$ has degree two, then applying (P2) to $ \{v \}$ removes $v$ and contracts its neighbors $x,y$ into a new vertex $v'$. When no further preprocessing rules can be applied, we say $G$ is \emph{simplified}; the graph then has a number of nice properties, for instance, $\minsurp(G) \geq 2, \mindeg(G) \geq 3$,  and the solution $\vec {\tfrac{1}{2}}$ is the unique optimal solution to $\LPVC(G)$.

\begin{proposition}
\label{p1p2prop}
After applying any preprocessing rule, we have $\mu(G') \leq \mu(G)$, and $\langle G, k \rangle$ is feasible if and only if $\langle G', k' \rangle$ is feasible.
\end{proposition}

\begin{observation}
\label{ppthm22}
Suppose $\minsurp(G) \geq 2$. Then any vertex set $X$ satisfies $\shad(X) \geq 2 - |X|$, and any indset $I$ satisfies $\shad(N[I]) \geq 2 -  \surp_{G}(I).$

 In particular, a vertex $u$ satisfies $\shad(u) \geq 1$ and $\shad(N[u]) \geq 3 - \deg(u)$.
\end{observation}
 
We define $S(G)$ to be the largest decrease in $k$ obtainable from applying rules (P1) -- (P3) in some efficiently-computable sequence.\footnote{It is not clear how to calculate the \emph{largest} decrease in $k$ via preprocessing rules, since applying some rules may prevent other rules. We assume that we have fixed some polynomial-time computable sequence of potential preprocessing rules to reduce $k$ as much as possible. At various points in our algorithm, we will describe simplifications available for intermediate graphs. We always assume that our preprocessing rules can find such simplifications.} For vertices $x_1, \dots, x_{\ell}$ and vertex sets $X_1, \dots, X_{r}$, we write $S(x_1, \dots, x_{\ell}, X_1, \dots, X_r)$ as shorthand for $S( G - \{x_1, \dots, x_{\ell} \} - X_1 - \dots -  X_r )$.

We record a few observations on simplifications of various structures.
\begin{proposition}
\label{simp-obs0}
Let $x,y$ be non-adjacent subquadratic vertices. Then $S(G) \geq 2$ if any of the three conditions hold: (i) $\codeg(x,y) = 0$ or (ii) $\deg(x) + \deg(y) \geq 3$ or (iii) $\minsurp(G) \geq 0$.
\end{proposition}

\begin{proposition}
\label{simp-obs02}
If $G$ has 2-vertices $x_1, \dots, x_{\ell}$ with pairwise distance at least $3$, then $S(G) \geq \ell$.
\end{proposition}

Since a simplified graph $G$ has $\minsurp(G) \geq 2$, we have the following immediate observation from \Cref{simp-obs0}:

\begin{observation}
\label{subquartic-simp}
If $G$ is simplified and $G - \{u,v \}$ has two non-adjacent subquadratic vertices, then $S(u,v) \geq 2$. 
\end{observation}

We now consider a more involved simplification rules.
\begin{lemma}
\label{simp-obs1}
If $\minsurp(G) \geq 0$ and we have an indset $I$ with $\surp_G(I) \leq 1$, then $S(G) \geq |I|$.
\end{lemma}
\begin{proof}
Let $r = |I|$. If $\surp_G(I) = 0$, or $I$ is a critical-set, we can apply (P1) or (P2) to $I$. So suppose that $\surp_G(I) = 1$ and $\surp_G(K) = 0$ for a non-empty subset $K \subsetneq I$. Choose $K$ to be inclusion-wise-maximal with this property, and let $s = |K|$. We apply (P1) to $K$, reducing $k$ by $s$ and obtaining a graph $G' = G - N[K]$. Now consider the indset $J = I \setminus K$ in $G'$, where $|J| = r - s$. We have $|N_{G'}(J)| = |N_G(I) \setminus N_G(K)| = r+1 - s$ and so $\surp_{G'}(J) = 1$. 

We claim that $J$ is a critical-set in $G'$. For, if some non-empty subset $K' \subseteq J$ has $|N_{G'}(K')| \leq |K'|$, then indset $I' = K \cup K' \subseteq I$ would have $|N_{G}(I')| \leq |N_G(K)| + |N_{G'}(K')| \leq |K| + |K'| = |I'|$. Hence $\surp_G(I') \leq 0$, contradicting maximality of $K$. So, we can apply (P2) to $J$ in $G'$, getting a further drop of $r-s$. Overall, we get a net drop of $s + (r-s) = r$ as desired.
\end{proof}

\begin{lemma}
\label{p5rule1}
We define a \emph{kite} to be a 3-vertex $u$ with neighbors $x,y,z$ where $x \sim y$ and $y \sim z$. 

If $\minsurp(G) \geq 1$, then applying (P3) to a kite $u$ yields a graph $G'$ with $k' \leq k - 2$ and $\mu(G') \leq \mu(G) - 1/2$. (See Figure~\ref{fig1}.)

\begin{figure}[H]
\vspace{0.75in}
\begin{center}
\hspace{0.5in}
\includegraphics[trim = 0.5cm 21.5cm 9cm 5cm,scale=0.45,angle = 0]{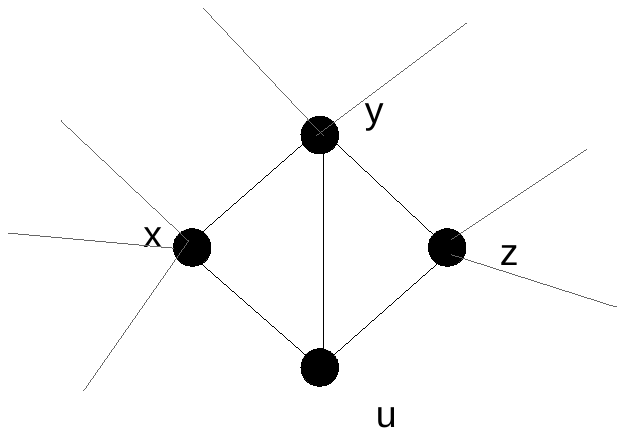}
\hspace{0.5in}
\includegraphics[trim = 0.5cm 21.5cm 9cm 5cm,scale=0.45,angle = 0]{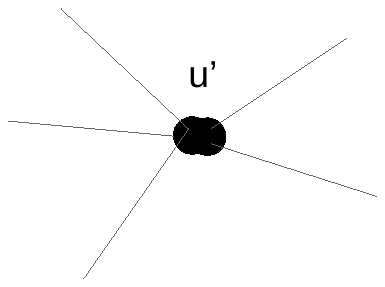}
\vspace{0.2 in}
\caption{\label{fig1} A kite $u,x,y,z$ before (left) and after (right) applying P3}
\end{center}
\end{figure}
\vspace{-0.4in}
\end{lemma}
\begin{proof}
We view the application of (P3) to $u$ as a two-part process. First, we remove the shared neighbor $y \in N(u) \cap N(x)$; then, we merge $z,x$ into a new vertex $u'$.  The first step gives graph $G'' = G - y$ with $k'' = k - 1, n'' = n - 1$. So, by \Cref{lambda-surp}, we have $\lambda(G'') = \lambda(G) - 1/2$ and so $\mu(G'') \leq \mu(G) - 1/2$. By \Cref{p1p2prop}, the second step gives $\mu(G') \leq \mu(G'')$ and $k' = k'' - 1$.
\end{proof}

\begin{lemma}
\label{lem:combinelem}
Suppose that $\mathcal G$ is a class of graphs closed under vertex deletion, and for which there are vertex cover algorithms with runtimes $O^*(e^{a \mu + b k})$ and $O(e^{c n})$ for $a,b,c \geq 0$. Then we can solve vertex cover in $\mathcal G$ with runtime $O^*(e^{d k})$ where $d = \frac{2 c (a+b)}{a+2 c}$.
\end{lemma}
\begin{proof}
First, exhaustively apply (P1) to $G$; the resulting graph $G'$ has $n' \leq n, k' \leq k$, and $\minsurp(G') \geq 1$. In particular,  $\mu(G') = k' - \lambda(G') = k' - n'/2$. Note that $G' \in \mathcal G$ since (P1) just deletes vertices.  Now dovetail the two algorithms for $G'$, running both simultaneously and returning the output of whichever terminates first. This has runtime $O^*( \min\{ e^{c n'} , e^{a (k' - n'/2) + b k'} \} )$. For fixed $k'$, this is maximized at $n' = \frac{2k'(a+b)}{a +2 c}$, at which point it takes on value $e^{d k'} \leq e^{d k}$.
\end{proof}

To illustrate \Cref{lem:combinelem}, consider graphs of maximum degree 3. We can combine the AGVC algorithm with runtime $O^*(2.3146^{\mu})$ (corresponding to $ a= \log(2.3146), b = 0$), and the MaxIS-3 algorithm with runtime $O^*(1.083506^n)$ (corresponding to $c = \log(1.083506)$), to get an algorithm with runtime $O^* (1.14416^k)$; this already gives us one of the results in \Cref{main-sum-thm}.

\subsection{Branching framework}
\label{sec:branchingrules}

Our algorithm follows the measure-and-conquer approach. Given an input graph $G$, it runs some preprocessing steps and then generates subproblems $G_1, \dots, G_t$  such that $G$ is feasible if and only if at least one $G_i$ is feasible. It then runs recursively on each $G_i$. Such an algorithm has runtime $O^*(e^{\phi(G)})$ as long as the preprocessing steps have runtime $O^*(e^{\phi(G)})$ and it satisfies\footnote{Here and throughout $e = 2.718...$ is the base of the natural logarithm.}
\begin{equation}
\label{eq:rec1}
 \sum_{i=1}^t e^{\phi(G_i) - \phi(G)} \leq 1 \tag{Branching Inequality}
 \end{equation}
We refer to $\phi$ as the \emph{measure} of the algorithm. For the most part, we will use measures of the form
\begin{equation}
\label{phi-eqn}
\phi(G) = a \mu + b k \qquad \qquad \text{for $a,b \geq 0$};
\end{equation}

We say a subproblem $G'$ has \emph{drop} $(\Delta \mu, \Delta k)$ if $k' \leq k - \Delta k$ and $\mu' \leq \mu - \Delta \mu$. We say a branching rule has \emph{branch sequence} (or \emph{branch-seq} for short) $B = [  (\Delta \mu_1, \Delta k_1), \dots, (\Delta \mu_t, \Delta k_t) ]$ if generates subproblems $G'_1, \dots, G'_t$ with the given drops. Given values of $a,b$, we define the \emph{value} of $B$ to be
\begin{equation}
\label{eq:rec2}
\val_{a,b}(B ) = \sum_{i=1}^t e^{-a \Delta \mu_i - b \Delta k_i}.
\end{equation}

We say branch-seq $B'$ \emph{dominates} $B$ if $\val_{a,b}(B') \leq \val_{a,b}(B)$ for all $a,b \geq 0$, and $G$ has a branch-seq $B$ \emph{available} if we can efficiently find a branching rule with a branch-seq dominating $B$. 

Most of our branching rules come from  guessing that a certain vertex set $X$ is in the cover, and then removing a set of vertices $Y$ which are isolated in $G - X$. The resulting subproblem $\langle G', k' \rangle$ with $G' = G - (X \cup Y), k' = k - |X|$ is called \emph{principal} with $\Delta k = |X|$ and  \emph{excess}  $s = |Y|$.   The graph $G'$ may have other simplifications available. We say $\langle G', k' \rangle$ has drop $(\Delta \mu',  |X|)$ \emph{directly}, and the final subproblem $G'' = \langle G'', k'' \rangle$ has drop $(\Delta \mu'', \Delta k'')$ \emph{after simplification}.   Note that if we apply (P1) to $G'$, then the subproblem $G'' = \langle G - (X \cup Y \cup N[Z]), k - |X| - |N(Z)| \rangle$ is again principal, with excess $|Y| + |Z|$; however, if we apply (P2) or (P3) to $G'$, then $G''$ is \emph{not} principal as it requires other operations  such as adding edges.

\begin{proposition}
\label{ppthm}
If $\minsurp(G) \geq 0$, then a principal subproblem $\langle G - Z, k - \Delta k \rangle$ with excess $s$  has $\Delta \mu=  \tfrac{1}{2} (\Delta k - s + \minsurp^0(G-X))$.
\end{proposition} 
\begin{proof}
\Cref{lambda-surp} gives $\lambda(G) = n/2$. Let $G' = G - Z$, where $|Z| = \Delta k + s$. Then $ \mu(G') = k' - \lpopt(G') = k' -  \tfrac{1}{2} (n' +  \minsurp^0(G') )$, where $n' = n - |Z| = 2 \lambda(G) - (s + \Delta k)$, which simplifies to $ \mu(G') = \mu(G) - \tfrac{1}{2} (\Delta k - s + \minsurp^0(G'))$.
 \end{proof}

 Consider a critical-set $I$ of $G$. By \Cref{s-prop1}, we know that there is either a good cover containing $I$, or disjoint to $I$. In the latter case, all the neighbors of $I$ must go into the cover.  There is a natural branching rule to guess which of these possibilities holds:
 
\begin{center}
\defbox{(B-Crit):}{For a critical-set $I$ branch on subproblems $\langle G - I, k - |I| \rangle$ and $\langle G - N[I], k - |N(I)| \rangle$. }
\end{center}

\begin{proposition}
\label{ppthmaaa}
Suppose that $\minsurp(G) \geq 2$. Let $I$ be a critical-set with $\surp_G(I) = \ell, |I| = r$. Then (B-Crit) generates principal subproblems  $G_1 = \langle G - I, k - r \rangle$ and $G_0 = \langle G - N[I], k - r - \ell \rangle$ with excess 0 and $r$ respectively. Furthermore:
\begin{itemize}
\vspace{-0.05in}
\item If $r = 1$, then the subproblem $G_1$ has $\Delta \mu = \tfrac{1}{2}$.
\vspace{-0.07in}
\item If $r > 1$, then the subproblem $G_1$ has $\Delta \mu = \tfrac{1}{2} (r + \minsurp^0(G - I)) \geq 1$.
\vspace{-0.07in}
\item The subproblem $G_0$ has $\Delta \mu = \tfrac{1}{2}( \ell + \minsurp^0(G - N[I])) \geq 1$.
\end{itemize}
\end{proposition}
\begin{proof}
The bound $\Delta \mu = \tfrac{1}{2} (r + \minsurp^0(G - I))$ for $G_1$ follows from \Cref{ppthm}. Also, by \Cref{ppthm22}, we have $\minsurp(G-I) \geq \minsurp(G) - |I| \geq 2 - r$. Thus, if $r = 1$, we have $\minsurp^0(G - I) = 0$. Similarly, by \Cref{ppthm}, we have  $\Delta \mu = \tfrac{1}{2}( \ell + \minsurp^0(G - N[I]))$ and  by \Cref{ppthm22}, we have $\minsurp(G - N[I]) \geq 2 - \surp_G(I) = 2 - \ell$.
 \end{proof}

\section{Analysis of vertex branching}

\label{sec:blockers}
Our bread-and-butter branching rule  is to choose a vertex $u$ and branch on whether it is in the cover.  We refer to this as \emph{splitting on $u$}. It is a special case of (B-Crit) with a singleton set $I = \{ u \}$, generating principal subproblems $\langle G - u, k-1 \rangle$ and $\langle G - N[u], k - \deg(u) \rangle$. By \Cref{ppthm}, if we start with a simplified graph, then the subproblems have drops $(0.5,1)$ and $(1, \deg(u))$ respectively. 

Let us examine the subproblem $\langle G - N[u], k - \deg(u) \rangle$ more carefully: if $\shad(N[u]) \geq 0$, then the subproblem has drop $(\tfrac{\deg(u)-1}{2},\deg(u))$. This is the ``generic'' and desired situation.  Otherwise, when $\shad(N[u]) \leq 0$, we say that $u$ is \emph{blocked}. If $x$ is a vertex in a min-set of $G - N[u]$, we say that $x$ is a \emph{blocker} of $u$. This motivates the following powerful branching rule:

\begin{center}
\defbox{(B-Block):}{Given a vertex $x$ which is a blocker of each of the vertices $u_1, \dots, u_{\ell}$,   branch on subproblems $G_1 = \langle G - \{u_1, \dots, u_{\ell}, x \}, k - \ell - 1 \rangle$ and $G_0 = \langle G - N[x], k - \deg(x) \rangle$. }
\end{center}

\begin{proposition}
(B-Block) is a valid branching rule.
\end{proposition}
\begin{proof}
The subproblem $G_0$ is feasible if and only if $G$ has a good cover omitting $x$, and the subproblem $G_1$ is feasible if and only if $G$ has a good cover which includes all vertices $u_1, \dots, u_{\ell}, x$. We claim that, if $\langle G, k \rangle $ is feasible, at least one of these cases holds. For, suppose a good cover of $G$ omits vertex $u_i$. Then $G - N[u_i]$ has a cover of size at most $k - \deg(u_i)$. By \Cref{s-prop1}, then $G - N[u_i]$ has such a cover $C'$ which omits any min-set $I$ of $G - N[u_i]$, and in particular $x \notin C'$. Then $G$ has a good cover $C = C' \cup N(u_i)$ omitting $x$. 
\end{proof}

In the vast majority of cases, we use (B-Block) with $\ell = 1$, that is, $x$ is a blocker of a single vertex $u$. For intuition, keep in mind the picture where $G - N[u]$ has $s$ isolated vertices $v_1, \dots, v_s$, so that $\{ v_1, \dots, v_s \}$ is an indset in $G - N[u]$ with zero neighbors and surplus $-s$, in particular, $\shad(N[u]) \leq -s$. See Figure~\ref{fig5}. 

\begin{figure}[H]
\vspace{0.9in}
\begin{center}
\hspace{0.5in}
\includegraphics[trim = 0.5cm 22.0cm 9cm 5cm,scale=0.45,angle = 0]{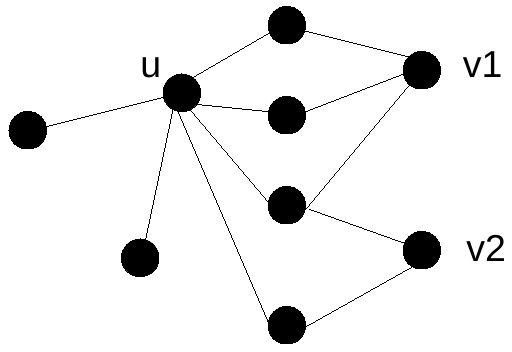}
\caption{\label{fig5} Here, vertices $v_1, v_2$ are both blockers of $u$; in particular, $\shad(N[u]) \leq -2$.}
\end{center}
\end{figure}
\vspace{-0.2in}

Also, note that if $\shad(N[u]) = 0$ exactly, then we could still split on $u$ if desired and get the ideal drop in $\mu$. However, we can also apply rule (B-Block), and this is usually more profitable. Vertices with $\shad(N[u]) = 0$ share many properties with vertices with $\shad(N[u]) < 0$ strictly. This is the reason for the somewhat unintuitive definition of \emph{blocked vertex.}

We will develop a series of branching rules to handle blocked vertices and other related cases. The most significant consequence of these rules is to give near-ideal branch sequences for high-degree vertices. Specifically, we have the following main result which is the heart of our algorithm:

\begin{restatable}{theorem}{thmA}
\label{branch5-3}
Suppose $G$ is simplified, and let $r = \maxdeg(G)$. Depending on $r$, the following branch-seqs are available:

\noindent If $r \geq 4$:  $[ (1,3), (1,5) ]$ or $[ (0.5,1), (1.5,4) ]$.
 
\noindent If $r \geq 5$:  $[ (1,3), (1,5) ]$ or $[ (0.5,1), (2,5) ]$.
 
\noindent If $r \geq 6$:   $[ (1,3), (1,5) ]$,   $[ (0.5,2), (2,5) ]$, or $[ (0.5,1),(2.5,r) ]$.
\end{restatable}

The proof of \Cref{branch5-3} has many cases which we will build up slowly.  \textbf{For \Cref{sec:blockers} only, we always assume that the starting graph $\boldsymbol{G}$ is simplified.}

\subsection{Branching rules for indsets with surplus two}
\label{sec:surptwo}
We first develop branching rules for non-singleton indsets with surplus two. As a starting point,  note that since $G$ is assumed to be simplified, we have $\minsurp(G) \geq 2$ hence such a set $I$ would be a critical-set. We can apply (B-Crit) to obtain subproblems $\langle G - I, k - |I| \rangle$ and $\langle G - N[I], k - r - 2 \rangle$. These are principal subproblems, and they have drops $(1,r), (1,r+2)$ respectively.

This is already a powerful branching rule, but it can be improved when $|I| = 2$ or $|I| = 3$ by finding additional simplifications in the subproblems. We start with a few preliminary observations.

\begin{observation}
\label{3ar-propg}
Suppose $G$ has an indset $I$ with $\surp_G(I) = 2$, and let $z \in N(I)$. 

Then $S(z) \geq |I|$, and for any vertex $x \notin I$, there holds $S(z,x) \geq |I|$.
\end{observation}
\begin{proof}
We have $\surp_{G - z}(I) = \surp_{G}(I) - 1 = 1$, and $\shad(z) \geq \minsurp(G) - 1 \geq 1$. So (P2) applied to $I$ gives $S(z) \geq |I|$.  Likewise,  $\shad(z,x) \geq \minsurp(G) - 2 \geq 0$ and $\surp_{G - \{z,x \}}(I) \leq \surp_G(I) - 1 \leq 1$. So \Cref{simp-obs1} gives $S(z,x) \geq |I|$.
\end{proof}

\begin{proposition}
\label{branch4-0}
Suppose  $\shad(N[u]) \leq 4 - \deg(u)$ and $x$ is a blocker of $u$.  If we apply (B-Block) to vertices $u,x$, then subproblem $G' = \langle G - N[x], k - \deg(x) \rangle$ has drop $(1,4)$. Moreover, if $\deg(x) = 3$ and $x$ is contained in a non-singleton min-set of $G - N[u]$, then $G'$ has drop $(1,5)$.
\end{proposition}
\begin{proof}
If $\deg(x) = 4$, then $G'$ has drop $(1,4)$ directly. So, suppose $\deg(x) = 3$; by \Cref{ppthm}, $G'$ has drop $(1,3)$ directly. Let $I \ni x$ be a min-set of $G - N[u]$. The indset $J = I \cup \{u \} \setminus \{x \}$ has $\surp_{G - N[x]}(J) \leq \surp_{G - N[u]}(I) + (\deg(u) - \deg(x)) \leq 1$. On the other hand,  $\shad(N[x]) \geq \minsurp(G) - \deg(x) + 1 \geq 0$. So \Cref{simp-obs1} gives $S(N[x]) \geq |J| = |I|$.
\end{proof}

	\begin{lemma}
\label{3ar-prop}
Suppose $G$ has an indset $I$ with $\surp_G(I) = 2, |I| \geq 3$. Then $G$ has available branch-seq $[ (1,4),(1,5) ]$  or $[ (0.5,4), (2,5) ]$.
\end{lemma}
\begin{proof}

If $|I| \geq 4$, then we simply apply (B-Crit). So suppose $|I| = 3$. The vertices in $I$ are independent, so there are at least $\sum_{x \in I} \deg(x) \geq 9$ edges from $I$ to $N(I)$. Since $|N(I)| = 5$, by the pigeonhole principle there must be at least one vertex $z$ with $|N(z) \cap I| \geq 2$. Our strategy will revolve around a chosen vertex $z$; as a tie-breaking rule, we assume that in a given case analysis there were no other vertices in a previous case.
\smallskip

\noindent \textbf{Case I: $\boldsymbol{\deg(z) \leq 4}$ and $\boldsymbol{N(z) \not \subseteq I}$.} We apply (B-Crit) to $I$, generating subproblems $\langle G - I, k - 3 \rangle $ and $\langle G - N[I], k - 5 \rangle$ with drops $(1,3)$ and $(1,5)$ directly. The former subproblem  has a subquadratic vertex $z$ so it has drop $(1,4)$ after simplification.

\smallskip

\noindent \textbf{Case II: $\boldsymbol{\deg(z) \geq 5}$ and $\boldsymbol{\shad(N[z]) \geq 5 - \deg(z)}$.} We split on $z$; by \Cref{3ar-propg},  subproblem $\langle G - z, k - 1 \rangle$ has drop $(0.5,4)$ after simplification, and subproblem $\langle G - N[z], k - \deg(z) \rangle$ has drop $(2,5)$ directly.  

\smallskip

\noindent \textbf{Case III: $\boldsymbol{\deg(z) \geq 5}$ and $\boldsymbol{z}$ has a blocker $\boldsymbol{t  \notin I}$.} We apply (B-Block) to $z,t$; by  \Cref{branch4-0}, subproblem $\langle G - N[t], k - \deg(t) \rangle$ has drop $(1,4)$ and by \Cref{3ar-propg}, subproblem $\langle G - \{z,t \}, k - 2 \rangle$ has drop $(1,5)$ after simplification.

\smallskip

\noindent \textbf{Case IV: $\boldsymbol{\deg(z) \geq 5}$.} Let $I = \{x_1, x_2, y \}$ where $z \sim x_1, z \sim x_2$. Since we are not in Case II or Case III, it must be that $y$ is the sole blocker of $z$: that is, we have $N(y) \subseteq N(z)$. 

 Each vertex $x_i$ has $|N(x_i, y)| = \deg(x_i) + \deg(y) - \codeg(x_i,y) \leq |N(I)|$, i.e. $\codeg(x_i, y) \geq \deg(x_i) - 2 \geq 1$. So we must have $\deg(x_1), \deg(x_2) \geq 4$, as otherwise some $x_i$ would have a 3-triangle with $u, z$, contradicting that $G$ is simplified.  As a result, there are at least $11$ edges going from $I$ to $N(I)$; since $|N(I)| = 5$ the pigeonhole principle implies there must be some vertex $z'$ with $|N(z') \cap I| = 3$. Since $z' \in N(y) \subseteq N(z)$, we must have $z' \sim z$ as well. So $N(z') \not \subseteq I$.  Necessarily, any blocker of this vertex $z'$ would be outside $I$. Depending on the degree of $z'$, it would fall into one of the earlier cases. This contradicts our tie-breaking rule for $z$. So Case IV is impossible.

\begin{figure}[H]
\vspace{1.1in}
\begin{center}
\hspace{-1in}
\includegraphics[trim = 0.5cm 21.5cm 9cm 5cm,scale=0.55,angle = 0]{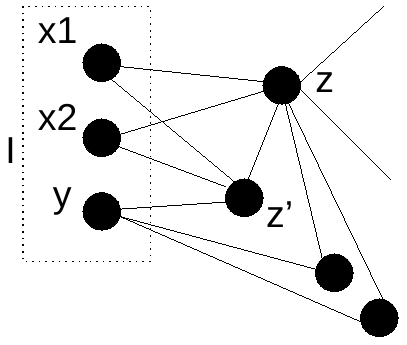}
\vspace{-0.2in}
\caption{\label{fig9a} Illustration of Case IV}
\end{center}
\end{figure}
\vspace{-0.15in}

\smallskip

\noindent \textbf{Case V: $\boldsymbol{N(z) \subseteq I}$.} Necessarily $\deg(z) = 3$. There are at least $6$ edges from $I$ to $N(I)  \setminus \{z \}$, so by the pigeonhole principle there is another vertex $z' \neq z$ with $|N(z') \cap I| \geq 2$. If $N(z') \subseteq I$, then necessarily $z \not \sim z'$. So $\{z, z' \}$ is an indset of surplus one, contradicting that $G$ is simplified. If $N(z') \not \subseteq I$, it would fall into one of the earlier cases. In either event, Case V is also impossible.
\end{proof}

\begin{lemma}
\label{2ar-prop}
Suppose $G$ has an indset $I$ with $\surp_G(I) = 2, |I| \geq 2$. Then $G$ has available branch-seq $[ (1,4),(1,4) ]$ or $[ (0.5,3), (2,5) ]$.
\end{lemma}
\begin{proof}
If $|I| \geq 3$, then we apply \Cref{3ar-prop}. So let $I = \{x_1, x_2 \}$ and let $A = N(x_1) \cap N(x_2)$. Note that $|A| = \deg(x_1) + \deg(x_2) - |N(I)| \geq \deg(x_1) - 1 \geq 2$. The vertices in $A$ cannot form a clique, as then $x_1$ would have a funnel. So consider vertices $z_1, z_2 \in A$ with $z_1 \not \sim z_2$.

If $z_1$ and $z_2$ are both subquartic, then we apply (B-Crit) to $I$; the subproblem $\langle G - N[I], k - 4 \rangle$ has drop $(1,4)$ directly. The subproblem $\langle G - I, k - 2  \rangle$ has non-adjacent subquadratic vertices $z_1, z_2$, so by \Cref{subquartic-simp} it has drop $(1,4)$ after simplification.

So suppose $\deg(z_1) \geq 5$. If $\shad(N[z_1]) \geq 0$, we split on $z_1$; by \Cref{3ar-propg}, the subproblem $\langle G - z_1, k - 1 \rangle$ has drop $(0.5,3)$ after simplification, while the subproblem $\langle G - N[z_1], k - \deg(z_1) \rangle$ has drop $(2,5)$ directly. Otherwise, suppose $\shad(N[z_1]) \leq -1$ and $t$ is a blocker of $z_1$. Note that $t \notin I$ since $t \not \sim z_1$ whereas $z_1 \in N(x_1) \cap N(x_2)$. We then apply (B-Block) to $t,z_1$; the subproblem $\langle G - N[z_1], k - \deg(z_1) \rangle$ has drop $(1,5)$ directly, and by \Cref{3ar-propg}, the subproblem $\langle G - \{t,z \}, k - 2 \rangle$ has drop $(1,4)$ after simplification. 
\end{proof}

\subsection{Branching rules for blocked vertices}
We now turn to developing branching rules for blocked vertices of various types.

\begin{proposition}
\label{branch4-3prep}
Suppose $\shad(N[u]) \leq 5 - \deg(u)$ and $u$ has a blocker $x$ of degree at least 4. Then $G$ has available branch-seq $[(1,4), (1,4)]$ or $[ (0.5,2), (2,5) ]$.
\end{proposition}
\begin{proof}
If $\shad(N[x]) \leq 3 - \deg(x)$ and $J$ is a min-set of $G - N[x]$, then $\surp_{G}(J \cup \{x \}) \leq (3 - \deg(x)) + (\deg(x) - 1) \leq 2$ and we can apply \Cref{2ar-prop} to indset $J\cup \{x \}$. 

So suppose $\shad(N[x]) \geq 4 - \deg(x)$. In this case, we apply (B-Block) to $u,x$. Subproblem $\langle G - \{u,x \}, k - 2 \rangle$ has drop $(1,2)$ directly. If $\deg(x) \geq 5$, then subproblem $\langle G - N[x], k - \deg(x) \rangle$ has drop $(1.5,5)$ directly. If $\deg(x) = 4$, then consider a min-set $I$ of $G - N[u]$ with $x \in I$. We have $\surp_{G - N[x]}(I \cup \{u \} \setminus \{x \}) \leq \surp_{G - N[u]}(I) + \deg(u) - \deg(x) \leq 1$. So $S(N[x]) \geq 1$ and $G'$ has drop $(1.5,5)$ after simplification.   In either case, $\langle G - \{u,x \}, k -2 \rangle$ and $\langle G - N[x], k - \deg(x) \rangle$ give drops $(1.2), (1.5,5)$ respectively.  This branch-seq $[(1,2), (1.5,5)]$ dominates $[(0.5,2), (2,5)]$. 
\end{proof}

\begin{lemma}
\label{branch55}
Suppose $G$ has a vertex of degree at least 5 which has a 3-neighbor. Then $G$ has available branch-seq  $[(1,3),(1,5)]$ or $[(0.5,2), (2.5)]$.
\end{lemma}
\begin{proof}
Let $u$ be a vertex of degree at least $5$ and let $z$ be a 3-neighbor of $u$. There are a number of cases to consider.

\smallskip

\noindent \textbf{Case I: $\boldsymbol{\shad(N[u]) \geq 5 - \deg(u)}$.}  Splitting on $u$ gives subproblems $\langle G - u, k-1 \rangle$ and $\langle G - N[u], k-\deg(u) \rangle$ with drops $(0.5,1), (2,5)$ directly; the former has drop $(0.5,2)$ after simplification due to its 2-vertex $z$. 

\smallskip

 \noindent \textbf{Case II: $\boldsymbol{u}$ has a blocker $\boldsymbol{x}$ with $\boldsymbol{ \deg(x)  \geq 4}$.} We suppose that $\shad(N[u]) \leq 4 - \deg(u)$ since otherwise it would be covered in Case I. Then the result follows from \Cref{branch4-3prep}, noting that $[(1,4), (1,4)]$ dominates $[(1,3), (1,5)]$.
   
\smallskip

\noindent \textbf{Case III: $\boldsymbol{u}$ has a blocker $\boldsymbol{x}$ with $\boldsymbol{ S(u,x) \geq 2}$.} We apply (B-Block) to vertices $u,x$. Subproblem $\langle G - \{u, x \}, k -2 \rangle$ has drop $(1,4)$ after simplification, and by \Cref{branch4-0} subproblem $\langle G - N[x], k - \deg(x) \rangle$ has drop $(1,4)$.  We get branch-seq $[(1,4),(1,4)]$, which dominates $[(1,3), (1,5)]$.
 
\smallskip
 
 \noindent \textbf{Case IV: $\boldsymbol{G-N[u]}$ has a non-singleton min-set $\boldsymbol{I}$.} Let $x \in I$. Necessarily $\deg(x) = 3$, else it would be covered in Case II. We apply (B-Block) to $u,x$;  subproblem $\langle G - \{u,x \}, k -2 \rangle$ has drop $(1,3)$ after simplification due to subquadratic vertex $z$. By \Cref{branch4-0} subproblem $\langle G - N[x], k - \deg(x) \rangle$ has drop $(1,5)$ after simplification.

\bigskip 

 \noindent \textbf{Case V: No previous cases apply.} Suppose that no vertices in the graph are covered by any of the previous cases.  Let $U \neq \emptyset$ denote the set of vertices $u$ with $\deg(u) \geq 5$ and $u$ having a 3-neighbor. We claim that every vertex $u \in U$ has degree exactly 5. For, suppose $\deg(u) \geq 6$, and since we are not in Case I we have $\shad(N[u]) \leq 4 - \deg(u) \leq -2$. Then necessarily a min-set of $G-N[u]$ would have size at least two, and it would be covered in Case IV.
  
So every $u \in U$ has degree exactly 5, and is blocked by a singleton 3-vertex $x$, i.e. $N(x) \subseteq N(u)$.   Let us say in this case that $x$ is \emph{linked} to $u$; we denote by $X$ the set of all such 3-vertices linked to any vertex in $U$. We claim that each $x \in X$ has at least two neighbors of degree at least 5, which must be in $U$ due to their 3-neighbor $x$. For, suppose that $x$ has two subquartic neighbors $z_1, z_2$; necessarily $z_1 \not \sim z_2$ since $G$ is simplified and $x$ has degree 3. Then $S(u,x) \geq 2$ by \Cref{subquartic-simp}, which would have been covered in Case III.

On the other hand, we claim that each vertex $u \in U$ has at most one neighbor in $X$. For, suppose $u$ is linked to $x$ and $u$ has two 3-neighbors $x_1, x_2 \in X$, which must be non-adjacent since $G$ is simplified. Then \Cref{subquartic-simp} gives $S(u,x) \geq 2$ due to subquadratic vertices $x_1, x_2$. This would have been covered in Case III.

Thus, we see that each vertex in $X$ has at least two neighbors in $U$ and each vertex in $U$ has at most one neighbor in $X$. Hence  $|U| \geq 2 |X|$, and so by the pigeonhole principle, there must be some vertex $x \in X$ which is linked to two vertices $u_1, u_2 \in U$. We apply (B-Block) to $u_1, u_2, x$, getting subproblems $\langle G - \{u_1,u_2, x \}, k -3 \rangle$ and $\langle G - N[x], k - \deg(x) \rangle$. By \Cref{branch4-0}, the latter subproblem has drop $(1,4)$.  Furthermore, if $z_1, z_2$ are two 5-neighbors of $x$, then $z_1 \not \sim z_2$ else $x$ would have a 3-triangle. So $G - \{u_1, u_2, x \}$ has non-adjacent 2-vertices $z_1, z_2$, and subproblem $\langle G - \{u_1,u_2, x \}, k -3 \rangle$ has drop $(1,5)$ after simplification.
\end{proof}

\begin{proposition}
\label{branch4-3}
Suppose a vertex $u$ has $\shad(N[u]) \leq 4 - \deg(u) \leq 0$. Then $G$ has available branch-seq $[ (1,3),(1,5) ]$ or $[ (0.5,2), (2,5) ]$.
\end{proposition}
\begin{proof}
Let $I$ be a min-set of $G-N[u]$. There must be some vertex $x \in I$ which shares some neighbor $t$ with $u$, as otherwise we would have $\surp_G(I) = \surp_{G - N[u]}(I) \leq 0$, contradicting that $G$ is simplified.  If $\deg(x) \geq 4$, we apply \Cref{branch4-3prep}, noting that $[(1,4), (1,4)]$ dominates $[(1,3), (1,5)]$. If $\deg(x) = 3$ and $\deg(t) \geq 5$, we apply \Cref{branch55} to $t$. If $\deg(x) = 3, \deg(t) \leq 4, |I| \geq 2$, we apply (B-Block) to $u,x$; subproblem $\langle G - \{u,x \}, k -2 \rangle$ has drop $(1,3)$ after simplification (due to subquadratic vertex $t$) and by \Cref{branch4-0} subproblem $\langle G - N[x], k - \deg(x) \rangle$ has drop $(1,5)$.

So finally take $I = \{x \}$ where $\deg(x) = 3$ and $x$ shares two neighbors $t_1, t_2$ with $u$. Necessarily $t_1 \not \sim t_2$ since $G$ has no 3-triangles. If either $t_1$ or $t_2$ has degree 5 or more, we can apply \Cref{branch55}. If they are both subquartic, then we apply (B-Block) to $u,x$; by \Cref{branch4-0}, subproblem $\langle G - N[x], k - \deg(x) \rangle$ has drop $(1,4)$ and by \Cref{subquartic-simp} the subproblem $\langle G - \{u,x \}, k - 2\rangle $ has drop $(1,4)$ after simplification. Here, $[(1,4), (1,4)]$ dominates $[(1,3), (1,5)]$.
\end{proof}

\begin{lemma}
\label{branch6-1}
Suppose vertex $u$ has $\shad(N[u]) \leq 5 - \deg(u)$. 
\begin{itemize}
\vspace{-0.05in}
\item If $\deg(u) = 5$, then $G$ has available branch-seq $[ (1,3),(1,4) ]$ or $[ (0.5,2), (2,5) ]$.

\vspace{-0.07in}
\item If $\deg(u) \geq 6$, then $G$ has available branch-seq $ [(1,3), (1,5)]$  or $[(0.5,2), (2,5)]$.
\end{itemize}
\end{lemma}
\begin{proof}
Let $I$ be an inclusion-wise-minimal min-set $I$ of $G - N[u]$.  If any vertex in $I$ has degree at least 4, the result follows from \Cref{branch4-3prep} (noting that $[(1,4), (1,4)]$ dominates $[(1,3), (1,5)]$. So suppose that all vertices in $I$ have degree 3. 

There must be some vertex $x \in I$ with $\codeg(x,u) \geq 1$ as otherwise we would have $\surp_G(I) = \surp_{G - N[u]}(I) \leq 0$.  For this vertex $x$, define $Z = N(x) \cap N(u)$ and $Y = N(x) \setminus N(u)$, where $|Z| + |Y| = \deg(x) = 3$ and $|Z| \geq 1$. 

If any vertex $z \in Z$ has degree at least 5,  we apply \Cref{branch55} to $z$ with its 3-neighbor $x$. If any vertex $z \in Z$ has degree 3, we apply \Cref{branch55} to $u$ with its 3-neighbor $z$. If $\codeg(x, x') \geq 2$ for any vertex $x' \in I \setminus \{x \}$, then  we apply \Cref{2ar-prop} to the surplus-two indset $\{x, x' \}$.  So we suppose that $\codeg(x,x')\leq 1$ for all $x' \in I \setminus \{x \}$ and that all vertices in $Z$ have degree exactly four.

The vertices in $Z$ must be independent, as otherwise $G$ would have a 3-triangle with $x$. If $\codeg(z, z') \geq 3$ for any vertices $z, z' \in Z$, then $\shad(N[z]) \leq 0$ (due to vertex $z'$), and we can apply  \Cref{branch4-3}.  So we suppose the vertices in $Z$ share no neighbors besides $u$ and $x$. In this case, in the graph $G - \{u,x \}$, the vertices in $Z$ have degree two and share no common neighbors.  By \Cref{simp-obs02}, we thus have $S(u,x) \geq |Z|$.

Now suppose that $|I| = 1$, i.e. $I$ consists of a single vertex $x$ with $\deg_{G-N[u]}(x) \leq 6 - \deg(u)$. In particular, if $\deg(u) = 5$, then $|Z| \geq 2$ and if $\deg(u) \geq 6$ then $|Z| \geq 3$. In this case, when we apply (B-Block) to $u,x$, then subproblem $\langle G - N[x], k - \deg(x) \rangle$ has drop $(1,3)$ directly and subproblem $\langle G - \{u,x \}, k - 2\rangle$ has drop $(1,4)$ or $(1,5)$ after simplification, for $\deg(u) = 5$ and $\deg(u) = 6$ respectively. This shows the claimed bound.

 \medskip
 
So suppose that $|I| \geq 2$. Each vertex $y \in Y$ must have a neighbor $\sigma(y) \in I \setminus \{x \}$, as otherwise we would have $\surp_{G - N[u]}(I') \leq \surp_{G - N[u]}(I)$ for non-empty indset $I' = I \setminus \{x \}$, contradicting minimality of $I$.  The vertices $\sigma(y): y \in Y$ must be distinct, as a common vertex $x' = \sigma(y) = \sigma(y')$ would have $\codeg(x, x') \geq 2$ which we have already ruled out. So, in the graph $G - N[x]$, each vertex $\sigma(y) \in I$ loses neighbor $y$.  By \Cref{ppthm22}, we have $\shad(N[x]) \geq 3 - \deg(x) = 0$. So $G - N[x]$ has $|Y| \leq 2$ non-adjacent subquadratic vertices. \Cref{simp-obs0} implies that $S(N[x]) \geq |Y|$.

Finally, our strategy is to apply (B-Block) to $u,x$. Subproblem $\langle G - \{u,x \}, k - 2\rangle$ has  drop $(1,2 + |Z|)$ after simplification and subproblem $\langle G - N[x], k - \deg(x) \rangle$ has drop $(1,3 + |Y|)$ after simplification. Since $|Z| + |Y| = 3$ and $|Z| \geq 1$, this always gives $(1,3), (1,5)$ or $(1,4), (1,4)$.
\end{proof}

We can finally prove \Cref{branch5-3} (restated for convenience).
\thmA*
\begin{proof}
Consider an $r$-vertex $u$. If $r = 4$ and $\shad(N[u]) \geq 0$, then split on $u$ with branch-seq $[(0.5,1), (1.5,4)]$. If $r = 4$ and $\shad(N[u]) \leq -1$, then use \Cref{branch4-3}. If $r = 5$ and $\shad(N[u]) \geq 0$, then split on $u$ with branch-seq $[(0.5,1), (2,5)]$. If $r = 5$ and $\shad(N[u]) \leq -1$, then apply \Cref{branch4-3}. If $r \geq 6$ and $\shad(N[u]) \geq 6 - r$,  then split on $u$, getting branch-seq $[ (0.5,1), (2.5,r) ]$. If $r \geq 6$ and $\shad(N[u]) \leq 5 - r$, then apply \Cref{branch6-1}.
\end{proof}

\section{A simple branching algorithm}
\label{sec:simplebranching}
We now describe our first branching algorithm for vertex cover, where the measure is a piecewise-linear function of $k$ and $\mu$. It is more convenient to describe this in terms of multiple self-contained ``mini-algorithms,'' with \emph{linear} measure functions. 

As a starting point, we can describe the first mini-algorithm: 
 
 \begin{algorithm}[H]
Simplify $G$.

If $\maxdeg(G) \leq 3$, then run either the algorithm of \Cref{agvc-thm} or the MaxIS-3 algorithm, whichever is cheaper, and return.

Otherwise, apply \Cref{branch5-3} to an arbitrary vertex of degree at least $4$, and run {\tt Branch4Simple} on the two resulting subproblems.
\caption{Function {\tt Branch4Simple}$(G, k)$}
\end{algorithm}

To clarify, despite the name {\tt Branch4Simple}, the graph is allowed to have vertices of either higher or lower degrees.  What we mean is that, depending on the current values of $k$ and $\mu$, it is advantageous to branch on a vertex of degree 4 or higher. Our algorithm is \emph{looking} for such a vertex; if it is not present (i.e. the graph has maximum degree 3) then an alternate algorithm should be used instead. Specifically, here we use either the AGVC algorithm with runtime depending on $\mu$ or the MaxIS-3 algorithm with runtime depending on $n = 2 (k - \mu)$.

Henceforth we describe our branching algorithms more concisely. Whenever we apply a branching rule, we assume that  we recursively use the algorithm under consideration (here, {\tt Branch4Simple}) on the resulting subproblems and return.  Likewise, when we list some other algorithms for the problem, we assume they are dovetailed (effectively running the cheapest of them).

\begin{proposition}
\label{alg1thm}
{\tt Branch4Simple} has measure $a \mu + b k$ for $a = 0.71808, b = 0.019442$.
\end{proposition}
\begin{proof}
Simplifying the graph in Line 1 can only reduce $\phi(G) = a \mu + b k$. If \Cref{branch5-3} is used, it gives a branch-seq $B$ with drops dominated by  $[ (1,3), (1,5) ]$ or  $[ (0.5,1), (1.5,4) ]$. To satisfy Eq.~(\ref{eq:rec2}), we thus need:
\begin{equation}
\label{amubbb0}
\val_{a,b}(B) \leq \max\{ e^{-a-3b} + e^{-a-5b}, e^{-0.5 a -b} + e^{-1.5 a - 4 b} \} \leq 1
\end{equation}

Otherwise, suppose $G$ has maximum degree $3$. Since it is simplified, it has $n = 2 \lambda = 2 (k - \mu)$, so the MaxIS-3 algorithm has runtime $O^*(1.083506^{2 (k - \mu)})$. Then, it suffices to show that
\begin{equation}
\label{amubbb}
\min\{ \mu \log(2.3146), 2(k - \mu) \log (1.083506)  \} \leq a \mu + b k.
\end{equation}
This can be verified mechanically. 
\end{proof}

By \Cref{lem:combinelem}, this combined with the MaxIS-4 algorithm (with $a = 0.71808, b = 0.019442, c = \log 1.137595$) immediately gives runtime $O^*(1.2152^k)$ for degree-4 graphs. 

Although \Cref{alg1thm} is easy to check directly, the choices for the parameters $a$ and $b$ may seem mysterious. The explanation is that the most constraining part of the algorithm is dealing with the lower-degree graphs, i.e. solving the problem when the graph has maximum degree $3$. The inequality (\ref{amubbb}), which governs this case, should exhibit a ``triple point'' with respect to {\tt Branch4Simple}, the MaxIS-3 algorithm, and the algorithm of \Cref{agvc-thm}. That is, for chosen parameters $a,b$, we should have
$$
\mu \log(2.3146) = 2 (k - \mu) \log(1.083506) = a \mu + b k.
$$

This allows us to determine $b$ in terms of $a$. Then our goal is to minimize $a$ whilst respecting the branching constraints of Eq.~(\ref{amubbb0}). This same reasoning will be used for parameters in all our algorithms. 

Likewise, we define algorithms for branching on degree 5 and degree 6 vertices.

 \begin{algorithm}[H]
Simplify $G$

If $\maxdeg(G) \leq 4$, then run either the MaxIS-4 algorithm or {\tt Branch4Simple}

Otherwise, apply \Cref{branch5-3} to an arbitrary vertex of degree at least $5$.
\caption{Function {\tt Branch5Simple}$(G, k)$}
\end{algorithm}

 \begin{algorithm}[H]
Simplify $G$

If $\maxdeg(G) \leq 5$, then run either the MaxIS-5 algorithm  or {\tt Branch5Simple}

Otherwise, apply \Cref{branch5-3} to an arbitrary vertex of degree at least $6$.
\caption{Function {\tt Branch6Simple}$(G, k)$}
\end{algorithm}

\begin{proposition}
\label{alg2thm}
{\tt Branch5Simple} has measure $a \mu + b k$ for $a = 0.44849, b = 0.085297$. 

{\tt Branch6Simple} has measure $a \mu + b k$ for $a = 0.20199, b = 0.160637$.
\end{proposition}

 The final algorithm is very similar, but we only keep track of runtime in terms of $k$, not $\mu$ or $n$. 

 \begin{algorithm}[H]
If $\maxdeg(G) \leq 6$, then use \Cref{lem:combinelem} with  algorithms of  {\tt Branch6Simple} and MaxIS-6

Otherwise, split on an arbitrary vertex of degree at least $7$.
\caption{Function {\tt Branch7Simple}$(G, k)$}
\end{algorithm}

\begin{theorem}
\label{thm:improved-vc-first}
Algorithm {\tt Branch7Simple}$(G, k)$ runs in time $O^*(1.2575^k)$.
\end{theorem}
\begin{proof}
If we split on a vertex of degree at least $7$, then we generate subproblems with vertex covers of size at most $k -1 $ and $k -7$. These have cost $1.2575^{k-1}$ and $1.2575^{k-7}$ by induction. Since $1.2575^{-1} + 1.2575^{-7}  < 1$, the overall cost is $O^*(1.2575^k)$. Otherwise the two algorithms in question have measure $a \mu + b k$ and $c n$ respectively; where $a = 0.20199, b = 0.160637, c = \log(1.1893)$;  by applying \Cref{lem:combinelem}, we get a combined algorithm with cost $O^*(1.2575^k)$ as desired.
\end{proof}

\section{The advanced branching algorithm: overview}
In the next sections, we consider more elaborate graph structures and branching rules. This improves the runtime to  $O^*( \dSeven^k)$, along with algorithms for lower-degree graphs. The outline is similar to \Cref{sec:simplebranching}: we develop a series of branching algorithms targeted to graphs with different degrees, building our way up from degree-4 graphs to degree-7 graphs.   Here we highlight some of the main changes in our analysis.

\subsection{Multiway branching} 
So far, we have only analyzed two-way branching rules. The advanced algorithms will use multiway branching; for example, we may split on a given vertex $u$, and then apply further branching steps to the subproblems $G - u$ and/or $G - N[u]$. 

A common tactic, which we refer to as \emph{splitting on a sequence of vertices $x_1, \dots, x_{\ell}$}, is to first split on $x_1$; in the subproblem $\langle G - x_1, k - 1 \rangle$ then split on $x_2$; in the subproblem $\langle G - \{x_1, x_2 \}, k - 2 \rangle$ then split on $x_3$, and so on.  When $x_1 \dots, x_{\ell}$ are independent, this generates subproblems $G_1 = \langle G - N[x_1], k - \deg(x_1) \rangle, G_2 = \langle G - x_1 - N[x_2], k - 1 - \deg(x_2) \rangle, \dots, G_{\ell} =  \langle G - \{x_1, \dots, x_{\ell-1} \} - N[x_{\ell}], k - \ell + 1 - \deg(x_{\ell}) \rangle, G_0 = \langle G - \{x_1, \dots, x_{\ell} \}, k - \ell \rangle$. We emphasize that $G_1$ does \emph{not} have further branching, such as guessing the status of $x_2$. Splitting on $x_1, x_2$ is just a 3-way branch.

A related tactic is to branch in a targeted way to create large surplus-two indsets, or so that preprocessing rules (P2) or (P3) create new high-degree vertices. The branching rules in Section~\ref{sec:surptwo} can then be applied, which are highly profitable. We record a few examples of this analysis here.

\begin{proposition}
\label{branch5-5}
Suppose that $G$ has one of the following two structures: 

\noindent (i) a 2-vertex $u$ with neighbors $x, y$ or 

\noindent (ii) a funnel $u$ with out-neighbor $x$ and a vertex $y \in N(u) \setminus N(x)$. 

In either case, define $
r = \deg(x) + \deg(y) - \codeg(x,y) - 1.
$

\hspace{-0.3in} If $r \geq 5$ then $G$ has available branch-seq  $[(0,2)]$,  $[ (1,4), (1,6) ]$ or  $[ (0.5,2), (2,6) ]$. 

\hspace{-0.3in} If $r \geq 6$ then $G$ has available branch-seq $[(0,2)]$, $[ (1,4), (1,6) ]$,  $[ (0.5,3), (2,6) ]$, or $[ (0.5,2),(2.5,r + 1) ]$.
\end{proposition}
\begin{proof}
In case (i), we apply (P2) to merge $x,y$ into a new $r$-vertex  $z$. In case (ii), we apply (P3); if $\codeg(u,x) \geq 0$ this directly gives drop $(0,2)$. Otherwise, it produces an $r$-vertex $y$. If the resulting graph $G'$ can be simplified further (aside from isolated vertices), then we get net drop $(0,2)$. Otherwise, we apply \Cref{branch5-3} to $G'$.
\end{proof}
\begin{proposition}
\label{almostred}

Suppose $\mindeg(G) \geq 1$.   We say graph $G$ is \emph {$r$-vulnerable} (or \emph{$r$-vul} for short)  if either (i) $S(G) \geq 1$ or (ii) $G$ has an indset with surplus at most two and size at least $r$.

 If $G$ is $2$-vulnerable it has available branch-seq $[(0,1)]$, $[(1,4),(1,4)]$, or $[(0.5,3), (2,5)]$. 
 
 If $G$ is $3$-vulnerable it has available branch-seq $[(0,1)]$, $[(1,4),(1,5)]$, or $[(0.5,4), (2,5)]$.
 
 If $G$ is $r$-vulnerable for $r \geq 4$, it has available branch-seq $[(0,1)]$ or $[(1,r), (1,r+2)]$. 
\end{proposition}
\begin{proof}
If $G$ is not simplified, then $S(G) \geq 1$ giving branch-seq $(0,1)$. Otherwise, we can apply \Cref{2ar-prop} or \Cref{3ar-prop} or just apply (B-Crit) depending on the size of the indset.
\end{proof}

\noindent

\medskip

\noindent \textbf{Notation for branching rules.} For the remainder of this paper, when analyzing branching subproblems, we usually omit the solution size $k'$. It should be determined from context: a subproblem $G' = G - X - N[Y]$ is principal with excess $|Y|$ and has $k' = k - |X| - |N(Y)|$. (For this notation, we are careful to ensure that $Y$ is an indset and that $X \not \sim Y$, to avoid double-counting.)

We use the following important compositional property: if we generate subproblems $i = 1, \dots, t$ each with drop $(\Delta \mu_i, \Delta k_i)$,  and then apply a branching rule with branch-seq $B'_i$ to the subproblem $i$, then the resulting branching rule $B$ overall has $$
\val_{a,b}(B) = \sum_{i=1}^{t} e^{-a \Delta \mu_i - b \Delta k_i} \val_{a,b} (B_i').
$$

For Propositions~\ref{branch5-5} and \ref{almostred}, we use the following notation for the $\val(B'_i)$:
\begin{align*}
\left.
\begin{aligned}
&\psi_5 = \max \{ e^{-2b}, e^{-a-4b} + e^{-a-6b}, e^{-0.5a-2b} + e^{-2a-6b} \} \ \ \  \\
&\psi_r = \max \{ e^{-2b}, e^{-a-4b} + e^{-a-6b},  e^{-0.5a-3b} + e^{-2a-6b}, e^{-0.5a-2b} + e^{-2.5a-(r+1) b} \} \quad \text{   $r \geq 6$} \thinspace \thinspace
\end{aligned}
\right\} \text{Prop~\ref{branch5-5}} \\ 
\left.
\begin{aligned}
&\gamma_2 = \max\{ e^{-b}, e^{-a-4b} + e^{-a-4b}, e^{-0.5a-3b} + e^{-2a-5b} \}  \qquad  \quad  \  \thinspace \ \qquad \qquad \qquad \ \ \\
&\gamma_3 = \max\{ e^{-b}, e^{-a-4b} + e^{-a-5b}, e^{-0.5a-4b} + e^{-2a-5b} \} \\
&\gamma_r = \max\{ e^{-b}, e^{-a-rb} + e^{-a-(r+2) b} \} \quad \text{   $r \geq 4$} \\
\end{aligned}
\right\} \text{Prop.~\ref{almostred}}
\end{align*}

Thus, for instance, if we have a subproblem with drop $(\Delta \mu, \Delta k)$ and then we apply \Cref{branch5-5}, its overall contribution to the Branching Inequality, would be given by $e^{-a \Delta \mu - b \Delta k} \psi_r.$

\subsection{Branching on funnels and other structures} If the graph contains a funnel, the simplest tactic is to apply (P3). However this creates additional edges, which  can make the graph too complicated to analyze further. It is often better to branch on the funnel itself. There are two possible branching rules for this situation:

\medskip

\noindent \textbf{(B-Fun1)} For a funnel $u$ with out-neighbor $x$, branch on whether $u \notin C$ or $x \notin C$, generating subproblems $G - N[u]$ and $G - N[x]$.\footnote{Recall our notational convention that these subproblems have $k' = k - \deg(u)$ and $k' = k - \deg(x)$ respectively.}

\smallskip 

\noindent \textbf{(B-Fun2)} For a funnel $u$ with out-neighbor $x$ and any other neighbor $v$ of $u$,  branch on whether $v \in C$ or whether $\{v, x \} \cap C = \emptyset$; this generates subproblems $G - v$ and $G - N[v,x]$.

\medskip

Note that, when we apply (B-Fun2), the vertex $u$ still has a funnel in the subproblem $G - v$ (of size reduced by one); we can continue to apply (P3) to it, or use other funnel reductions.

Consider a degree-two vertex $u$ with neighbors $x,y$. This can be regarded as a funnel, where $x$ is an out-neighbor and the vertex $\{y \}$ forms a (trivial) clique. If we apply (B-Fun2), we get subproblems $ G - y$ and $G - N[x,y]$. In the former subproblem, $u$ now has degree one; we can apply (P1) to remove $u$ and its neighbor $x$ and get subproblem $ G - \{x,y \} = G - N[u]$. Since this comes up so frequently, we summarize this branching strategy as follows:  

\medskip

\noindent \textbf{(B-Pair)} For a 2-vertex $u$ with neighbors $x, y$ branch on subproblems $G - \{x ,y \}$ and $G - N[x,y]$.

\medskip

 We occasionally use some other specialized branching rules; most of these will be discussed in the specific situations where they are needed. There is one branching rule worth further discussion, as it comes up in a number of contexts:

\medskip

\noindent \textbf{(B-Shared)} Suppose vertices $u,v$ share neighbors $x_1, \dots, x_{\ell}$. Branch on whether $\{u, v \} \subseteq C$, or $\{x_1, \dots, x_{\ell} \} \subseteq C$; this generates subproblems $G - \{u,v \}$ and $G - \{x_1, \dots, x_{\ell} \}$.

\medskip 

Note that, in applying (B-Shared), it is allowed to have $\codeg(u,v) > \ell$ strictly.
\subsection{Tracking degree-3 vertices} In the algorithms we have encountered so far, we only use two parameters of the graph: the values $k$ and $\mu$. This will also be the case for most of our advanced algorithms. For the degree-4 algorithm, however, we will need to track a few additional pieces of information. The most important is whether the graph has any \emph{degree-3 vertices} (the precise number of them is not tracked). These would lead to more profitable branching rules; as a rule of thumb, having a degree-3 vertex is almost as good as simplifying the graph.

For the degree-4 algorithm, we extend our notations to track 3-vertices.  We write $S(G) \geq \ell^*$ if either (i) $S(G) = \ell$ and the simplified graph $\langle G', k - \ell \rangle$ has a 3-vertex, or (ii) $S(G) \geq \ell+1$. We say that a subproblem $\langle G', k' \rangle$ has drop $(\Delta \mu, \Delta k^*)$ if $\mu' \leq \mu - \Delta \mu$ and either (i) $G'$ has a subcubic vertex and $k' = k - \Delta k$ or (ii) $k' < k - \Delta k$ strictly.

Here we record a  few observations about simplification and degree-3 vertices.

\begin{proposition}
\label{reduce-prop}
Let $s \geq 1$. If $G$ has at least $2(s+1)$ subcubic vertices, and it has a critical-set $I$ with $\surp_G(I) = 1, |I| \geq s$, then  $S(G) \geq s^*$.
\end{proposition}
\begin{proof}
Apply (P2) to $I$. If $|I| > s$, this gives $S(G) \geq s+1$ directly.  Otherwise,  all subcubic vertices outside $N[I]$ remain subcubic, leaving at least $2(s+1) - |N[I]| = 1$ subcubic vertices.
\end{proof}

\begin{proposition}
\label{reduce-propxx}
If $G$ has no 4-cycles, has at least 7 subcubic vertices, and has two non-adjacent 2-vertices, then $S(G) \geq 2^*$.
\end{proposition}
\begin{proof}
Let $u,v$ be non-adjacent 2-vertices. Apply (P2) to $v$; then all the subcubic vertices in $G$ outside from $N[v]$ remain subcubic in the resulting graph $G'$. So $G'$ has at least 4 subcubic vertices. Also, since $G$ has no 4-cycle, we   have $\deg_{G'}(u) = 2$.  By \Cref{reduce-prop}, this  implies $S(G') \geq 1^*$ and hence $S(G) \geq 2^*$.
 \end{proof}

\begin{lemma}
\label{remove-lemma}
Suppose that $\maxdeg(G) \leq 4$, and let $G' = \langle G - (X \cup Y), k - |X| \rangle$ be a principal subproblem. If all components of $G$ have size at least $5 |X| + 1$, then $S(G')$ has drop $(\Delta \mu, |X|^*)$.
\end{lemma}
\begin{proof}
There must be some vertex $z$ at distance two from $X$, since otherwise the component(s) containing $X$ would have size at most $5 |X|$. So suppose $z$ has a neighbor $y \sim X$. Here $\deg_{G - X}(y) \geq 1$  (due to neighbor $z$) and also $\deg_{G-X}(y) \leq 3$ (it loses some neighbor in $X$). So $G - X$ has a subcubic vertex. Removing isolated vertices does not affect this.
\end{proof}

\begin{proposition}
\label{remove-lemma3}
Let $r \geq 2$ be an arbitrary constant. Suppose that $G$ has at least 4 subcubic vertices and $\maxdeg(G) \leq 4$ and $\minsurp(G) \geq 1$ and $G$ has an indset $I$ with $|I| \geq r$ and $\surp_G(I) \leq 2$. Then $G$ has available branch-seq $[(0,1^*)]$ or $[(1,r^*), (1,(r+2)^*)]$.  
\end{proposition}
\begin{proof}
If $G$ has any components of size at most $5(r+2)$, we solve them exactly in constant time and remove them from $G$; since $\minsurp(G) \geq 1$, this reduces $k$ by at least two and hence gives drop $[(0,2)]$. If $\surp_G(I) = 1$, then \Cref{reduce-prop} with $s = 1$ gives $S(G) \geq 1^*$. Otherwise, we apply (B-Crit); by \Cref{remove-lemma} the subproblems have drops $(1,r^*)$ and $(1,(r+2)^*)$.
\end{proof}

Along the lines of \Cref{almostred}, we say that $G$ is $r^*$-vul if either $S(G) \geq 1^*$ or the preconditions of \Cref{remove-lemma3} hold.

 \section{Branching on 4-vertices}
We now develop an improved algorithm {\tt Branch4}. At a high level, the improvement over the baseline algorithm {\tt Branch4Simple} comes from tracking the presence of 3-vertices. If a 4-vertex $u$ has a 3-neighbor $v$, then we can split on $u$ and then, in the branch $G - u$, the degree of $v$ gets reduced to two. We then use (P2) to reduce $k$ by one -- a ``bonus'' reduction. We can also ensure that this creates new 3-vertices; thus, the process is a self-sustaining chain reaction. 

We sketch an outline of the algorithm here; many additional details will be provided later. 

 \begin{algorithm}[H]

Apply misc. branching rules to ensure $G$ has nice structure, e.g. simplified and no 4-cycles.
      
If $V_3 \neq \emptyset$ and $V_3 \not \sim V_4$, then solve vertex cover exactly on $V_3$. 

Split on a 4-vertex with a 3-neighbor, chosen according to tie-breaking rules to be described later.

Branch on a triangle.

Branch on a vertex $u$ according to tie-breaking rule to be described later.
\caption{Function {\tt Branch4}$(G, k)$}
\end{algorithm}

 \begin{theorem}
 \label{branch4alt-thm}
 Algorithm {\tt Branch4} has measure $a_4 \mu + b_4 k$ for $a_4 = \aFour, b_4 = \bFour$.
 \end{theorem}
 
 Via \Cref{lem:combinelem}, in combination with the MaxIS-4 algorithm, this will give our algorithm with runtime $O^*(\dFour^k)$  for degree-4 graphs.  \Cref{branch4alt-thm} is very lengthy to prove. Our strategy is to show that {\tt Branch4} has the measure of the form
 $$
 \phi(G) = a \mu + b k - \bonus(G).
 $$

Let us define a \emph{special 4-vertex} to be a 4-vertex with at least two 3-neighbors, and we say $G$ is \emph{special} if it has a special 4-vertex. We then define:
\begin{align*}
\bonus(G) &= \begin{cases}
\beta & \text{if $\mindeg(G) = 3$ and $\maxdeg(G) = 4$ and $G$ is special} \\
\alpha & \text{if $G$ has a subcubic vertex and the previous case does not hold} \\
0 & \text{otherwise}
\end{cases} \\
& \text{for parameters $\alpha = 0.03894, \beta = 0.05478$}
\end{align*}

Note that $\bonus(G)$ is bounded by a constant and so this in turn will show \Cref{branch4alt-thm}. Although the precise value of the parameters is not critical, we often use the important inequalities
\begin{equation}
\label{r2b-eqn}
0 \leq \alpha \leq b \leq \beta \leq 2 b
\end{equation}

To show that a branching rule generating subproblems $G_1, \dots, G_t$ respects the measure $\phi$, we using a slight variant of the branching inequality:
\begin{equation}
\label{eq:rec3}
\sum_{i = 1}^t e^{\phi(G_i) - a \mu - b k} \leq e^{-\bonus(G)}
\end{equation}

As a consequence of Eq.~(\ref{r2b-eqn}) and \Cref{remove-lemma3}, we have the following observation (which motivates our notation for tracking 3-vertices as we have described earlier):
\begin{observation}
\negthickspace  If $S(G) \geq \ell^*$, then the simplified graph $G'$ has $\phi(G') \leq \phi(G) - b \Delta k - \alpha$.

\negthickspace If subproblem has drop $(\Delta \mu, \Delta k^*)$, then its contribution to Eq.~(\ref{eq:rec3}) is at most $e^{-a \Delta \mu - b \Delta k - \alpha}$.

\negthickspace  If subproblem has drop $(\Delta \mu, \Delta k)$ and is  $r^*$-vul, its contribution to Eq.~(\ref{eq:rec3}) is at most $e^{-a \Delta \mu - b \Delta k} \gamma^*_r$ where we define $\gamma^*_r = \max\{ e^{-b-\alpha}, e^{-a-r b - \alpha} + e^{-a-(r+2) b - \alpha} \}$.
\end{observation}

\subsection{Regularizing the graph and basic analysis}
We will first ensure that $G$ has a number of nice properties:
\begin{itemize}
\vspace{-0.05in}
\item $G$ is simplified.
\vspace{-0.07in}
\item $\maxdeg(G) \leq 4$
\vspace{-0.07in}
\item $\shad(N[u]) \geq 1$ for all vertices $u$. 
\vspace{-0.07in}
\item $G$ has no 4-cycle
\vspace{-0.07in}
\item $G$ has no connected-components of size smaller than some  constant $K$.
\end{itemize}

Let us explain the last clause here. Our branching algorithms will generate principal subproblems $G - X$ with some given drop $(\Delta \mu, \Delta k)$ a constant $\Delta k$.  If  $K > 5 \Delta k + 1$, then by \Cref{remove-lemma}, such  subproblem automatically has drop $(\Delta \mu, \Delta k^*)$ ``for free''.  We choose $K$ to be larger than all such constants needed for the analysis; we do not specify it explicitly.

\begin{proposition}
\label{join-prop}
If $\maxdeg(G) \leq 4$ and $\minsurp(G) \geq 2$, but $G$ does not satisfy the Line-1 conditions, then it has available branch-seq $[(0,2)]$, $[(1,3), (1,5)]$, $[(0.5,2), (2,5)]$, or $[(1,3^*), (1,3^*)]$
\end{proposition}
\begin{proof}
If $G$ has a connected-component $X$ of size $O(1)$, we solve it exactly and update $k' = k - |C_X|$ where $C_X$ is a good cover of $X$. Since $\minsurp(G) \geq 2$, we have $|C_X| \geq 2$ and this gives drop $[(0,2)]$. 

If $G$ has a funnel $u$ with out-neighbor $x$,  then we apply (B-Fun1), generating subproblems $G - N[u], G - N[x]$. Since $\minsurp(G) \geq 2$, these subproblems have drop $(1,3)$ directly; by  \Cref{remove-lemma} they indeed both have drop $(1,3^*)$. 

Otherwise, rules (P1) and (P2) cannot be applied; so for the remainder of the analysis we may suppose that $G$ is simplified.

If $G$ has a blocked 4-vertex, then we apply \Cref{branch4-3}. 

Finally, if $G$ has a 4-cycle $u_1, u_2, u_3, u_4$, we consider two cases. First, suppose there is a ``cross-edge'', say $u_1 \sim u_3$. We must have $\deg(u_1) = \deg(u_3) = 4$, else they would have a funnel. Let $v_1, v_3$ be the other neighbors of $u_1, u_3$. We split on $v_1$; subproblem $G - N[v_1]$ has drop $(1,3^*)$ by \Cref{remove-lemma}. Subproblem $G - v_1$ has drop $(0.5,1)$ directly, and has a kite involving vertex $u_1$; by \Cref{p5rule1}, after applying (P3), we get net drop $(1,3)$ and the resulting graph has the 3-vertex $u_3$.

Otherwise, suppose that $u_1 \not \sim u_3$ and $u_2 \not \sim u_4$. We apply (B-Shared) to get subproblems $G - \{u_1, u_3 \}$ and $G - \{u_2, u_4 \}$. These both have two non-adjacent subquadratic nodes ($u_2, u_4$ and $u_1, u_3$ respectively), so by \Cref{subquartic-simp} they both have drop $(1,4)$. 
 \end{proof}

\begin{proposition}
\label{ttr1a}
After Line 1, for a 4-vertex $u$, the graph $G - N[u]$ has at least 8 subcubic vertices.
\end{proposition}
\begin{proof}
Suppose $u$ has 3-neighbors $v_1, \dots, v_s$ and 4-neighbors $v_{s+1}, \dots, v_s$ . We have $\codeg(u, v_i) = 0$ for  $i \leq s$, since otherwise $v_i$ would have a 3-triangle. Likewise, we have $\codeg(u, v_i) \leq 1$ for $i > s$, as  otherwise $G$ would have a 4-cycle.  Thus, the neighborhood of $u$ has at most $\lfloor (4 - s)/2 \rfloor$ edges, and hence  at least $3 s + 4 (4 - s) - 2 \lfloor (4 - s)/2 \rfloor - 4 \geq 8 $ edges from $N(u)$ to vertices $t$ at distance two from $u$. Since $G$ has no 4-cycles,  each such vertex $t$ has an edge to precisely one vertex in $N(u)$.
\end{proof}

We now analyze the easier steps of the algorithm line-by-line. Observe that, starting at Line 4,  the graph $G$ is 4-regular.  Lines 3 and 5, which are more involved, will be considered later. 

\medskip

\noindent \textbf{Line 1.} We may remove any isolated vertices, as this can only increase $\bonus(G)$.

First suppose that $\bonus(G) = 0$.  If $G$ is not simplified, then rules (P1)--(P3) give drop $[(0,1)]$. If $\maxdeg(G) \geq 5$, then we apply \Cref{branch5-3}; otherwise, we apply \Cref{join-prop}. In order to satisfy Eq.~(\ref{eq:rec3}), we need:
\begin{align}
\max\{ e^{-b}, e^{-0.5a-b} + e^{-2a-5b}, 
 e^{-a-3b-\alpha} + e^{-a-3b-\alpha}, e^{-a-3b} + e^{-a-5b} \} &\leq 1 \tag{4-1} 
 \end{align}
 
Next suppose $\bonus(G) = \beta$. Since $\mindeg(G) \geq 3$,  applying (P1) or (P2) reduces $k$ by at least two. Otherwise, if $\minsurp(G) \geq 2$, then we apply \Cref{join-prop}. Overall, we need:
\begin{align}
\max\{ e^{-2b},   e^{-0.5a-2b} + e^{-2a-5b},  e^{-a -3 b - \alpha} + e^{-a-3b-\alpha} , e^{-a - 3 b} + e^{-a-5b} \} &\leq e^{-\beta} \tag{4-2}
\end{align}

Finally suppose $\bonus(G) = \alpha$. If $G$ is not simplified, then rules (P1)--(P3) give drop $[(0,1)]$. If $\maxdeg(G) \leq 4$, then we can apply \Cref{join-prop}. Finally, suppose $G$ is simplified and has a vertex $u$ of degree at least $5$. If $\shad(N[u]) \leq 4-\deg(u)$, then we apply \Cref{branch6-1}. If $\shad(N[u]) \geq 5-\deg(u)$, then we split on $u$;  subproblem $G - N[u]$ has drop $(2,5)$, while subproblem $G - u$ retains the subcubic vertex in the original graph and so has drop $(0.5,1^*)$.  Over all cases, we need:
\begin{align}
\max\{ e^{-b}, e^{-0.5a-b-\alpha} + e^{-2a-5b}, 
  e^{-a-3b-\alpha} + e^{-a-3b-\alpha}, e^{-a-3b} + e^{-a-5b} \} & \leq e^{-\alpha} \tag{4-3}
 \end{align}

\medskip

\noindent \textbf{Line 2.} We run either the MaxIS-3 algorithm or \Cref{agvc-thm} on $V_3$.  Let $k_3$ be the size of the minimum vertex cover for $V_3$ and let $\mu_3 = k_3 - n_3/2$. Removing $V_3$  reduces $\phi(G)$ by at least $b k_3 - \beta \geq 2 b -\beta \geq 0$. Running the algorithms on $V_3$ has cost $O^*(2.3146^{\mu_3})$ or $O^*(1.085306^{n_3})$ respectively.  Since $n_3 \leq n = 2(k-\mu)$ and $\mu_3 \leq \mu$, it suffices to show:
\begin{align}
\min \{ \mu \log(2.3146), 2(k - \mu) \log(1.085306) \} \leq a \mu + b k \tag{4-4}
\end{align}

\medskip

\noindent \textbf{Line 4.} Let $N(u) = \{v_1, \dots, v_4 \}$ where $v_1 \sim v_2$. Since $G$ has no 4-cycles,  $\{ v_1, v_2 \} \not \sim \{ v_3, v_4 \}$.  We split on $v_3$, getting subproblems $G - N[v_3], G - v_3$ with drops $(1.5,4^*)$ and $(0.5,1^*)$. In the latter subproblem, we apply \Cref{branch5-5} to the funnel $u$ with out-neighbor $v_4$ and vertex $v_1 \in N(u) \setminus N(v_4)$, where $r = \deg_{G - v_3}(v_1) + \deg_{G - v_3} (v_4) - \codeg_{G - v_3}(v_1, v_4) - 1 \geq  4 + 3 - 1 - 1= 5$. We need:
\begin{align}
e^{-1.5 a - 4 b - \alpha} + e^{-0.5a-b} \psi_5   \leq 1 \tag{4-5}
\end{align}
 
\subsection{Line 3: Branching on a 4-vertex with a 3-neighbor}
\label{line43sec}
Suppose Line 3 selects a 4-vertex $u$, with 3 neighbors $v_1, \dots, v_s$ and 4-neighbors $v_{s+1}, \dots, v_4$. We let $v = v_1$ be a chosen 3-vertex, and let $x, y$ be the other two neighbors of $v_1$,  sorted so $\deg(x) \geq \deg(y)$.   Since $G$ has no 4-cycles, and the 3-vertex $v$ has no triangles, all vertices $u, v_{i}, x_i, y_i$ are distinct, and $\{x, y \} \not \sim \{u, v_2, v_3, v_4 \}$. See Figure~\ref{fig9a}.

\begin{figure}[H]
\vspace{0.95in}
\begin{center}
\hspace{0.5in}
\includegraphics[trim = 0.5cm 21.5cm 9cm 5cm,scale=0.55,angle = 0]{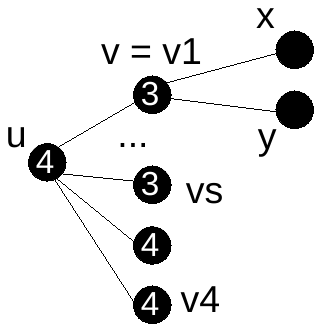}
\vspace{-0.2in}
\caption{\label{fig9a} The vertex $u$ and its associated neighbor vertices. The degrees of the nodes are shown.}
\end{center}
\end{figure}
\vspace{-0.15in}

 Note that there may be multiple ways of ordering vertices in this way; e.g., we may permute $v_1, \dots, v_s$,  and if $\deg(x) = \deg(y)$ we can swap $x/y$.  Let us fix some ordering. As a starting point, our strategy is to split on $u$, producing subproblems $G_0 = G - N[u]$ and $G - u$. There are two main further branching strategies we can use on  $G - u$:
\begin{itemize}
\item Apply (P2), contracting its 2-vertices $v_1, \dots, v_s$ into vertices $z_1, \dots, z_s$; we denote the resulting graph by $G_1$. Overall we generate two subproblems $G_0, G_1$.
\item  Apply (B-Pair) to the 2-vertex $v$, generating further subproblems $G_2 = G - N[v], G_3 = G - u - N[x, y]$. Overall we generate three subproblems $G_0, G_2, G_3$. 
\end{itemize}

Throughout Section~\ref{line43sec} we write \begin{align*}
r &= \deg_{G_1}(z_1)  = \deg(x) + \deg(y) - 2 \\
s &=\text{\# of 3-neighbors of } u
\end{align*}

\begin{proposition}
\label{ttr8e}
If $s \geq 2$, then $S(G_2) \geq 1^*$.  If $r = 4$, then $S(G_0) \geq 2^*$.  If $r = 5$, then $S(G_0) \geq 1^*$.
\end{proposition}
\begin{proof}
For the first result, if $s \leq 3$, then $G - N[v_1]$ has a 3-vertex $v_4$ and a 2-vertex $v_2 \not \sim v_4$; applying (P2) to $v_2$ reduces $k$ by one, and still leaves a subcubic vertex $v_4$.  If $s = 4$, then $G - N[v_1]$ has  non-adjacent 2-vertices $v_2, v_3$, so $S(G_2) \geq 2$ by \Cref{simp-obs0}.

For the next two results, recall from \Cref{ttr1a} that $G_0$ has at least 8 subcubic vertices. If $r = 4$, then $\deg_{G_0}(x) = \deg_{G_0}(y) = 2$ and $x \not \sim y$, so by \Cref{reduce-propxx} we have $S(G_0) \geq 2^*$. Similarly, if $r = 5$, then $\deg_{G_0}(y) = 2$ so by \Cref{reduce-prop} we have $S(G_0) \geq 1^*$. 
\end{proof}

\begin{proposition}
\label{ttr1e}
If $r \geq 5$, then $G_3$ is a principal subproblem with excess two and $\Delta k = r+2$ and $\minsurp(G_3) \geq 4-r$. Furthermore, the following bounds all hold:
\begin{itemize}
\vspace{-0.05in}
\item If $r = 5$ and $\minsurp(G_3) = -1$, then $S(N[x]) \geq 2^*$.
\vspace{-0.07in}
\item If $r = 5$ and $\minsurp(G_3) = -1$ and $ s= 1$, then $S(N[x]) \geq 3^*$.
\vspace{-0.07in}
\item If $r = 6$ and $\minsurp(G_3) = -2$, then $S(N[x]) \geq 3^*$.
\vspace{-0.07in}
\item If $r = 6$ and $\minsurp(G_3) = -1$, then $G - N[x]$ is $3^*$-vul.
\end{itemize}
\end{proposition}
\begin{proof}
Here $\deg(x) = 4$ and $\deg(y) = r-2$.  Since $G$ has no 4-cycles, we have $|N[x, y]| = r+1$ and so $G_3$ is principal with excess two and $\Delta k = r+2$. If $s = 1$, then the sole 3-neighbor of $u$ gets removed from $G_3$, and so $\mindeg(G_3) \geq 1$. Also note that, by \Cref{ttr1a}, the graph $G - N[x]$ has at least 8 subcubic vertices and  $\shad( N[x]) \geq 1$. 
 
Let $I$ be a min-set in $G_3$ and $\ell = \surp_{G_3}(I)$. If $\ell \leq -2$ or $s = 1$ then necessarily $|I| \geq 2$. Then $\surp_{G -N[x]}(I \cup \{ y \}) \leq \ell + (\deg_G(y) - 1) = \ell + r - 3$.  We have $\shad(N[x]) \geq 1$ so $\ell \geq 4 - r$.  Furthermore, if $\ell = 4 - r$ exactly, then $\surp_{G - N[x]}( I \cup \{y \}) = 1$, and then \Cref{reduce-prop} applied to  indset $I \cup \{y \}$ gives $S( N[x]) \geq 2^*$ for $|I| = 1$ and $S( N[x]) \geq 3^*$ for $|I| \geq 2$.

Finally, suppose $r = 6$ and $\ell = -1$. If $|I| \geq 2$, then $\surp_{G - N[x]}(I \cup \{y \}) \leq 2$, which by \Cref{remove-lemma3}  implies that $G - N[x]$ is $3^*$-vul. Otherwise, suppose $|I| = 1$, i.e. $G_3$ has an isolated vertex $t$; since $G$ has no 4-cycles, this is only possible if $\deg_G(t) = 3$ and $t \sim u$ and $\codeg(t, x) = \codeg(t, y) = 1$. Then $\deg_{G - N[x]}(t) = 2$ and by \Cref{reduce-prop} we have $S(N[x]) \geq 1^*$.
\end{proof}

 There are a number of cases, which implicitly define our tie-breaking rule for $u$:   when considering a given case, we assume that no other 4-vertex in the graph and no other ordering of its neighbors would be covered by any of the previous cases.

\medskip

\noindent \textbf{Case I: $\boldsymbol{s \geq 2}$ and $\boldsymbol{S(G_0) \geq 2^*}$}. We generate subproblems $G_0, G_1$. If $s \in \{2,3 \}$, then $G_1$ has drop $(0.5,3)$ and has a 3-vertex $v_4$; if $s = 4$, then $G_1$ has drop $(0.5,5)$ directly. In any case, $G_1$ has drop $(0.5,3^*)$.  We then simply need:
\begin{equation}
 e^{-1.5 a - 6 b - \alpha} + e^{-0.5 a - 3 b - \alpha}  \leq e^{-\beta} \tag{4-6}
\end{equation}

\noindent \textbf{Case II:  $\boldsymbol{G_0}$ is $\boldsymbol{3^*}$-vul and $\boldsymbol{S(G_2) \geq 1^*}$ and $\boldsymbol{r \geq 5}$ and $\boldsymbol{ \minsurp(G_3)  \geq 5-r}$}. We generate subproblems $G_0, G_2, G_3$, where $G_2$ has drop $(1,4^*)$ after simplification, and $G_3$ has drop $(2.5,7^*)$. We need:
\begin{equation}
 e^{-1.5 a - 4 b} \gamma_3^* + e^{-a-4b-\alpha} + e^{-2.5 a -7b - \alpha} \leq e^{-\beta} \tag{4-7}
\end{equation}

\noindent \textbf{Case III:  $\boldsymbol{s \geq 2}$.} By \Cref{ttr8e}, we have $S(G_2)   \geq 1^*$; also, we must have $r \geq 5$, since if $r = 4$ then $S(G_0) \geq 2^*$ and it would covered in Case I.   In particular, $\deg(x) = 4$.

We claim that $G - N[x]$ is not $3^*$-vul. For, suppose it were,  and consider $u' = x$ and its 3-neighbor $v' = v$ and its two neighbors $x' = u, y' = y$. We have $r' = \deg(x') + \deg(y') - 2 = r$. If $\shad( u' , N[x', y'])  \leq 4 - r'$, then \Cref{ttr1e} applied to $u'$ would give $S( N[x']) = S(G_0) \geq 2^*$ so $u$ would be covered in Case I. On the other hand, if $\shad( u' , N[x', y']) = \minsurp(G_3') \geq 5 - r'$, then vertex $u'$, with $v', x', y'$, would be covered by Case II.

Now, if $\minsurp(G_3) \leq -1$, then by \Cref{ttr1e}, the graph $G - N[x]$ would be $3^*$-vul, which we have already ruled out. So $\minsurp(G_3) \geq 0$. Consequently, if $r = 5$ then observe that $S(G_0) \geq 1^*$ and so we would be in Case II.  So $r = 6$, and our branching rule is to generate subproblems $G_0, G_2, G_3$, where $G_2$ has drop $(3,8^*)$ and we need:
\begin{equation}
 e^{-1.5 a - 4 b - \alpha} + e^{-a-4b-\alpha} + e^{-3 a - 8b - \alpha} \leq e^{-\beta} \tag{4-8}
\end{equation}

\medskip

\noindent \textbf{Note:} henceforth, we know the graph has no special 4-vertex (it would be covered in Case III or earlier). In particular, we have $s=1$ and $\bonus(G) = \alpha$.

\medskip

\noindent \textbf{Case IV: $\boldsymbol{S(G_0) \geq 3^*}$}. We generate subproblems $G_0, G_1$. Since $s = 1$, the graph $G_1$ has a subcubic vertex $v_2$ and so it has drop $(0.5,2^*)$. We need:
\begin{equation}
 e^{-1.5 a - 7 b - \alpha} + e^{-0.5 a - 2 b - \alpha } \leq e^{-\alpha } \tag{4-9}
\end{equation}

\noindent \textbf{Case V: $\boldsymbol{r \geq 5}$ and $\boldsymbol{G_0}$ is $\boldsymbol{3^*}$-vul.} We generate subproblems $G_0, G_2, G_3$, where $G_2$ has drop $(1,3^*)$. We have $\minsurp(G_3) \geq 5-r$, as otherwise  \Cref{ttr1e} would give $S( N[x]) \geq 3^*$ (bearing in mind that $s = 1$); then $u' = x$ would be covered in Case IV. So subproblem $G_3$ has drop $(2.5,7^*)$. Overall, we need:
\begin{equation}
 e^{-1.5 a - 4 b} \gamma_3^* + e^{-a-3b-\alpha} + e^{-2.5 a - 7b - \alpha} \leq e^{-\alpha} \tag{4-10}
\end{equation}

\noindent \textbf{Case VI: $\boldsymbol{r \geq 5}$.} We must have specifically $r = 6$, as if $r = 5$ then by \Cref{ttr1e} we would have $S(G_0) \geq 1^*$ putting us into Case V.

Again we generate subproblems $G_0, G_2, G_3$, where $G_2$ has drop $(1,3^*)$ as before. If $G - N[x]$ is $3^*$-vul, then vertex $u' = x$ would be covered in Case V earlier. So assume $G - N[x]$ is not $3^*$-vul; by \Cref{ttr1e}, this implies $\minsurp(G_3) \geq 0$ and so $G_3$ has drop $(3,8^*)$. We need:
\begin{equation}
 e^{-1.5 a - 4 b - \alpha} + e^{-a-3b-\alpha} + e^{-3a-8b - \alpha} \leq e^{-\alpha} \tag{4-11}
\end{equation}

\noindent \textbf{Case VII: Everything else.} 
  Here, $G$ cannot have any 3-vertex with two 4-neighbors, as any such 4-neighbor $u'$ would have $r' \geq 5$ and be covered in Case VI.   So we assume that all 3-vertices in the graph have at least two 3-neighbors. In particular $r = 4$, and  $x \sim x'$ and $y \sim y'$ for 3-vertices $x',y'$ distinct from $v$. We must have $x' \neq y'$ since $G$ has no 4-cycles. The merged vertex $z_1$ in $G_1$ has degree 4 with 3-neighbors $x',y'$, that is, it is a special vertex (see Figure~\ref{fig9}).  Since $\mindeg(G_1) \geq 3$ and $\maxdeg(G_1) \leq 4$, we get $\bonus(G_1) = \beta$. Also, $S(G_0) \geq 2^*$ since $r = 4$. We need:
\begin{equation}
e^{-1.5 a - 6 b - \alpha} + e^{-0.5 a - 2 b -\beta}  \leq e^{-\alpha} \tag{4-12}
\end{equation}

\begin{figure}[H]
\vspace{1.2in}
\begin{center}
\hspace{0.5in}
\includegraphics[trim = 0.5cm 21.5cm 9cm 5cm,scale=0.5,angle = 0]{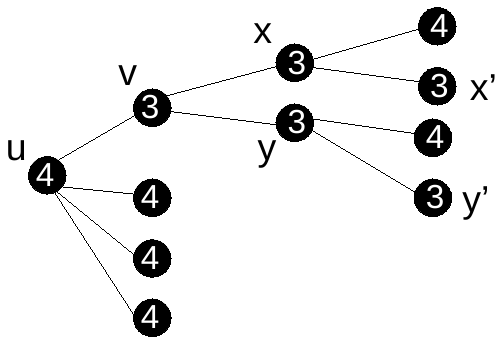}
\vspace{-0.07in}
\caption{\label{fig9} We show the degrees of each vertex. In $G - u$, the 2-vertex $v$ has neighbors $x, y$, which get merged into $z_1$. The new vertex has 4 neighbors in total and at least two 3-neighbors.}
\end{center}
\end{figure}
\vspace{-0.35in}

\subsection{Line 5: Branching on a vertex in a 4-regular graph}
\label{line45sec}
We select a vertex $u$ and an ordering of its neighbors $v_1, v_2, v_3, v_4$; our general strategy is to split on $v_1, v_2$ so as to create a 2-vertex in the resulting graph $G - \{v_1, v_2 \}$. To analyze this process, we use the following notations in Section~\ref{line45sec}:
\begin{align*}
&H = G - \{v_1, v_2 \} - N[v_3, v_4], \qquad  A = \{ v_1, v_2 \} \\
&I = \text{some chosen min-set of $H$}, \qquad r = |I|
\end{align*}

The graph $H$ and its properties are critical to this analysis. We begin with a few observations.

\begin{proposition}
\label{4spec-prop}
Suppose $G$ has a 6-cycle $x_1, \dots, x_6$. Then $S(x_1, N[x_4]) \geq 2^*$; furthermore, either $\codeg(x_1, x_4) = 1$ or $x_4$ is a special 4-vertex in $G - N[x_1]$.
\end{proposition}
\begin{proof}
Since $\girth(G) \geq 5$,  the graph $G' = G - x_1 - N[x_4]$ has two non-adjacent 2-vertices $x_2, x_6$. Since $G$ is 4-regular,  $G'$ has at least 11 subcubic (the 12 vertices at distance two to $x_4$, minus $x_1$). So  \Cref{reduce-propxx} gives $S(G') \geq 2^*$. 
Now suppose $\codeg(x_1, x_4) = 0$. Then the graph $G - N[x_1]$ retains vertex $x_4$ and its neighbors, while $x_3, x_5$ lose their neighbors $x_2, x_6$ respectively. Since $\girth(G) \geq 5$, we have $x_3 \neq x_5$ and the degrees of $x_3, x_5$ are not reduced by two. So $x_4$ has two  3-neighbors $x_3, x_5$ in $G - N[x_1]$.
\end{proof}

\begin{proposition}
\label{ttr1dc}
We have $\mindeg(H) \geq 1$. Moreover, for any vertex $x \in H$, there holds:
\begin{itemize}
\vspace{-0.05in}
\item If $\deg_H(x) = 1$, then $x \sim A$ and $S(u,N[x]) \geq 2^*$.
\vspace{-0.07in}
\item If $\deg_H(x) = 2$, then either $x \sim A$ or $x$ is a special 4-vertex in $G  - N[u]$.
\end{itemize}
\end{proposition}
\begin{proof}
Because $\girth(G) \geq 5$, we have $\codeg(x,v_3) \leq 1, \codeg(x, v_4) \leq 1$ and $|N(x) \cap A| \leq 1$. If $\deg_H(x) = 1$, then necessarily $x$ has a neighbor $v_i \in A$ and a neighbor $r_3  \sim v_3$ and a neighbor $r_4 \sim v_4$ and thus $u, v_3, r_3, x, r_4, v_4$ is a 6-cycle; by \Cref{4spec-prop} we have $S(u, N[x]) \geq 2^*$. Likewise, if $\deg_H(x) = 2$ and $x \not \sim A$, then $x$ has neighbors $r_3 \sim v_3, r_4 \sim v_4$, and then $G$ has a 6-cycle $u, v_3, r_3, x, r_4, v_4$. Furthermore, $x \not \sim A$ by hypothesis and $x \not \sim \{v_3, v_4 \}$ since $x \in H$. Hence $\codeg(u, x) = 0$ and so by \Cref{4spec-prop} the vertex $x$ is special in $G - N[u]$.
\end{proof}

\begin{proposition}
\label{ttr1db}
If $\minsurp(H)  < 0$ and $\deg_H(x) \leq 2$ for all $x \in I$, then $G  - N[u]$ is special.
\end{proposition}
\begin{proof}
We assume that $x \sim A$ for all $x \in I$, as otherwise the result follows from \Cref{ttr1dc}.  Since $\mindeg(H) \geq 1$, and $|N(I)| < |I|$ strictly, there must be some vertex $y \in H$ with two neighbors $x_1, x_2 \in I$. The neighbors of $x_1, x_2$  in $A$ must be distinct else $G$ would have a 4-cycle. So suppose without loss of generality $x_1 \sim v_1, x_2 \sim v_2$ and $G$ has a 6-cycle $u, v_1, x_1, y, x_2, v_2$. Note that $y \not \sim \{v_3,v_4 \}$ since $y \in H$ and $y \not \sim \{v_1, v_2 \}$ since then $G$ would have a 3-cycle $v_1, x_1, y$ or $v_2, x_2, y$. So $\codeg(u,y) = 0$ and hence $G - N[u]$ is special by \Cref{4spec-prop}.
\end{proof}

We now describe the branching strategy, in particular, how to order $v_1, \dots v_4$.  For our tie-breaking rule, we assume that if we are in a given case that no other ordering of neighbors $v_1', \dots, v_4'$ would fall into a previous case.
\medskip
  
  \noindent \textbf{Case I: $\boldsymbol{\shad(N[v_3]) \leq 1}$.} Let $J$ be a min-set of $G - N[v_3]$; since $\mindeg(G - N[v_3]) \geq 3$, necessarily $|J| \geq 2$, and since we are after Line 1 we have $\surp_{G - N[v_3]}(J) = 1$. Here, we split on $v_3$. By \Cref{ttr1a}, $G - N[v_3]$ has at least 8 subcubic vertices, so  $S(N[v_3]) \geq 2^*$ by \Cref{reduce-prop}. Overall, we need:
\begin{align}
e^{-0.5 a - b - \alpha} + e^{-1.5 a - 6 b - \alpha} \leq 1 \tag{4-13}
\end{align}

\noindent \textbf{Case II: $\boldsymbol{\minsurp(H) \leq -2}$.} Since $\mindeg(H) \geq 1$, we must have $ r \geq 3$. Let $J = I \cup \{v_4 \}$. Here, $\surp_{G - N[v_3]}(J) \leq 2$, and $|J| = |I|+1 \geq 4$. Since we are not in Case I, we have $\shad(N[v_3]) \geq 2$ and so $J$ is a min-set in $G - N[v_3]$.

Here we split on $v_3$ and in the subproblem $G - N[v_3]$, we apply (B-Crit) to $J$.  Overall, this generates subproblems $G - v_3, G - N[v_3] - J, G - N[v_3,J]$. Since we are not in Case II, we have $\shad(N[v_3], J) \geq \shad(N[v_3]) - |J| \geq 2 - |J|$ and likewise $\shad(N[v_3,J]) \geq \shad(N[v_3]) - |N[J]| \geq 2 - |N[J]|$. So the latter two subproblems have drops $(2.5,8^*), (2.5,10^*)$ respectively. Here we need:
\begin{align}
e^{-0.5 a - b - \alpha} + e^{-2.5 a - 8 b - \alpha} + e^{-2.5 a - 10 b - \alpha}  \leq 1 \tag{4-14}
\end{align}

  \medskip
\noindent \textbf{Note:} In all subsequent cases, our strategy is to split on $v_1, v_2$, and then apply (B-Pair) to vertex $u$ in subproblem $G - \{v_1, v_2 \}$. This gives four subproblems $G_1 = G - N[u], G_2 = G - N[v_1], G_3 = G - v_1  - N[v_2], H$ where $G_1, G_2, G_3$ have drop $(1.5,4^*), (1.5,4^*), (2,5^*)$ directly. 

  \medskip

\noindent \textbf{Case III: $\boldsymbol{\minsurp(H) = 0}$.} Here $H$ has drop $(3.5,9^*)$, so we need:
\begin{align}
  e^{-1.5a - 4b - \alpha} + e^{-1.5 a - 4 b - \alpha} + e^{-2a - 5 b - \alpha} + e^{-3.5 a - 9 b - \alpha} \leq 1 \tag{4-15} 
  \end{align}
  
 \medskip
  
\noindent \textbf{Note:} For the remainder of the analysis, we assume $\minsurp(H) = -1$. We have $r \geq 2$ since $\mindeg(H) \geq 1$. For $r \leq O(1)$, the graph $H$ has drop $(3, (8+r)^*)$ after applying (P1). 

\medskip

\noindent \textbf{Case IV: $\boldsymbol{S(G_3) \geq 2^*}$.} There are a few different subcases depending on $r$. If $r \geq 4$, then $H$ has drop $(3,12^*)$ and we need:
\begin{align}
e^{-1.5a - 4b - \alpha} + e^{-1.5 a - 4 b - \alpha}  + e^{-2a - 7 b - \alpha} + e^{-3 a  - 12 b - \alpha} \leq 1 \tag{4-16}
\end{align}

Next suppose that $r = 3$ and $H$ has drop $(3,11^*)$. For $x \in I$, we have $\deg_H(x) \leq |N(I)| = r - 1 = 2$, so $G - N[u]$ is special by \Cref{ttr1db}. We then need:
\begin{align}
e^{-1.5a - 4b -\beta}  + e^{-1.5 a - 4 b - \alpha}  + e^{-2a - 7 b - \alpha} + e^{-3 a  - 11 b - \alpha} \leq 1 \tag{4-17}
\end{align}

Finally suppose $r = 2$ and $H$ has drop $(3,10^*)$. Again, $G - N[u]$ is special by \Cref{ttr1db}. Let $I = \{x_1, x_2 \}, N_H(I) = \{y \}$. Since $\deg_H(x_1) = \deg_H(x_2) = 1$,  \Cref{ttr1dc}  implies that $x_1 \sim A$ and $x_2 \sim A$;  since they already have common neighbor $y$, we assume without loss of generality that $x_1 \sim v_1$ and $x_2 \sim v_2$. Now consider the 6-cycle $u, v_1, x_1, y, x_2, v_2$.  Let us enumerate $N(x_2) = \{v_2, y, r_3, r_4 \}$ for vertices $r_3 \sim v_3, r_4 \sim v_4$.  We must have $\codeg(v_1, x_2) = 0$; for, if $v_1 \sim y$, then $G$ would have a 3-cycle $v_1, x_1, y$, and if $v_1 \sim v_2$, then $G$ would have a 3-cycle $u, v_1, v_2$, and if $v_1 \sim r_i$ for $i \in \{3,4 \}$, then $G$ would have a 4-cycle $r_i, v_i, u, v_1$.

So by \Cref{4spec-prop}, the graph $G_2 = G - N[v_1]$ is also special. Accordingly, in this case we need:
\begin{align}
e^{-1.5a - 4b -\beta}  + e^{-1.5 a - 4 b - \beta}  + e^{-2a - 7 b - \alpha} + e^{-3 a  - 10 b - \alpha} \leq 1 \tag{4-18}
\end{align}

\noindent \textbf{Case V: $\boldsymbol{G_3}$ is $\boldsymbol{3^*}$-vul.}  If $\deg_H(x) = 1$ for any $x \in I$, then by \Cref{ttr1dc}, $u, x$ share a neighbor $v \in A$ and $S(u, N[x]) \geq 2^*$. But then vertex $u' = v$ and its two neighbors $v_1' = u, v_2' = x$ would be covered in Case IV, or an earlier case. 

Accordingly, we must have $\deg_H(x) \geq 2$ for all $x \in I$. Now consider counting the two-element subsets of $N_H(I)$. Each vertex $x \in I$ contributes $\binom{\deg_H(x)}{2}$ such subsets and since $G$ has no four-cycles these are all distinct. So $\sum_{x \in I} \binom{\deg_H(x)}{2} \leq \binom{|N_H(I)|}{2} = \binom{r-1}{2}$. Since $\deg_H(x) \geq 2$ for all $x \in I$, this implies that $r \geq 5$ and, furthermore, if $r = 5$ then $\deg_H(x) = 2$ for all $x \in I$; in the latter case, by \Cref{ttr1db} the graph $G - N[u]$ is special. So for $r = 5$ or $r \geq 6$, we need respectively:
\begin{align}
e^{-1.5a - 4b -\beta}  + e^{-1.5 a - 4 b - \alpha} + e^{-2a - 5 b} \gamma_3^* + e^{-3 a  - 13 b - \alpha} \leq 1 \tag{4-19} \\
e^{-1.5a - 4b -\alpha}  + e^{-1.5 a - 4 b - \alpha}  + e^{-2a - 5 b} \gamma_3^* + e^{-3 a  - 14 b - \alpha} \leq 1 \tag{4-20} 
\end{align}

\noindent \textbf{Case VI: Everything else.} Since we are not in Case I,  we have $\shad( v_1, N[v_3]) \geq \shad( N[v_3]) - 1 \geq 1$. On the other hand, $\surp_{G - v_1 - N[v_3]} (I \cup \{v_4\}) \leq \minsurp(H) + 1 +  (\deg_{G - N[v_3]}(v_4) - 1) = 2$.  So $G - v_1 - N[v_3]$ is $3^*$-vul via indset $I \cup \{v_4 \}$; it would be covered by Case V. Case VI is impossible.

\section{Branching on 5-vertices}

Our strategy for the degree-5 algorithm depends on the relative density of 5-vertices.  When 5-vertices are very common, we can find closely-clustered groups of them for more advanced branching rules; the denser the 5-vertices, the more favorable the branching. We will actually develop three separate degree-5 algorithms with different ``intensity'' levels. 

There is a major difference between the degree-4 and degree-5 algorithms. In the former, branching on a \emph{single} vertex can create new 2-vertices. This cannot be maintained in the degree-5 case: at least two 5-vertices are \emph{jointly} needed for the desired simplifications. The process is \emph{not} self-sustaining.  On the positive side, it means that we do not need to keep track of 3-vertices; as far as our analysis goes, they will be created anew at each branching step.

We will show for the following main result in this section:
\begin{theorem}
\label{branch5-alt}
 {\tt Branch5}$(G, k)$ has measure $a_5 \mu + b_5 k$ for $a_5 = \aFiveThree, b_5 = \bFiveThree$.
 \end{theorem}
  
By \Cref{lem:combinelem}, \Cref{branch5-alt} combined with the MaxIS-5 algorithm immediately gives the algorithm with runtime $O^*(\dFive^k)$  for degree-5 graphs.  

The first algorithm is  very simple and is targeted to the regime where 5-vertices are rare:

 \begin{algorithm}[H]
Simplify $G$

If $\maxdeg(G) \leq 4$, then run either the MaxIS-4 algorithm or {\tt Branch4}

Apply \Cref{branch5-3} to an arbitrary vertex of degree at least $5$.
\caption{Function {\tt Branch5-1}$(G, k)$}
\end{algorithm}

\begin{proposition}
{\tt Branch5-1} has measure $a_{51} \mu + b_{51} k$ for $a_{51} = \aFiveOne, b_{51} = \bFiveOne$.
\end{proposition}
\begin{proof}
 Given the potential branch-seqs from \Cref{branch5-3}, we need the following constraints:
 \begin{align}
& \max\{  e^{-a - 3 b} + e^{-a - 5 b}, e^{-0.5 a - b} + e^{-2 a - 5 b} \} \leq 1 \tag{5-1} \\
&\min\{ a_4 \mu + b_4 k, 2 (k - \mu) \log(1.137595) \} \leq a \mu + b k \tag{5-2} 
\end{align}
The latter constraint can be routinely checked by LP algorithms.
\end{proof}

The next algorithms {\tt Branch5-2} and {\tt Branch5-3} are targeted to medium and high density of 5-vertices  respectively. They share  many branching rules, so we describe them in a single listing.

\begin{algorithm}[H]

Apply miscellaneous two-way branching rules to ensure $G$ has nice structure.

Branch on non-adjacent 5-vertices $u,v$ with $S(u,v) \geq 1$ and $S(u, N[v]) \geq 1$.

Branch on vertices $u,v$ where $\deg(u) = 5, \deg_{G - u}(v) \geq 4$, and $G - \{u,v \}$ has a kite.

Branch on a 4-vertex with two non-adjacent 5-neighbors

Branch on a 5-vertex $u$ with neighbors $v_1, v_2, v_3, v_4, v_5$ where $v_1 \sim v_2$ and $\deg(v_1) = \deg(v_3) = \deg(v_4) = \deg(v_5) = 5$.

Branch on a $5$-vertex with five  $5$-neighbors.

\If{$L=2$}{
Branch on a 5-vertex $u$ where  $G - N[u]$ is 3-vul.

Branch on non-adjacent 5-vertices $u, v$ with $S(u,v) \geq 1$.

Branch on a $4$-vertex which has a triangle involving a 5-vertex.

Branch on a 5-vertex with fewer than three  4-neighbors.

}
Run the cheaper of  {\tt Branch5-$(L-1)$}$(G, k)$ or the MaxIS-5 algorithm
\caption{Function {\tt Branch5-$L$}$(G, k)$ for $L = 2,3$}
\end{algorithm}

We will show the following results for these algorithms:
\begin{theorem}
 {\tt Branch5-2}$(G, k)$ has measure $a_{52} \mu + b_{52} k$ for $a_{52} = \aFiveTwo, b_{52} = \bFiveTwo$.
 
 {\tt Branch5-3}$(G, k)$ has measure $a_{53} \mu + b_{53} k$ for $a_{53} = \aFiveThree, b_{53} = \bFiveThree$.
\end{theorem}

Thus, algorithm {\tt Branch5-3} will achieve the claim of \Cref{branch5-alt}.

\subsection{Regularizing the graph and structural properties}
In Line 1, we will apply a series of branching rules to ensure the following conditions:
\begin{itemize}
\vspace{-0.05in}
\item $G$ is simplified.
\vspace{-0.07in}
\item $\maxdeg(G) \leq 5$.
\vspace{-0.07in}
\item Every 5-vertex $u$ has $\shad(N[u]) \geq 1$ and $S(u) = 0$.
\vspace{-0.07in}
\item Every 4-vertex $u$ has $\shad(N[u]) \geq 1$.
\vspace{-0.07in}
 \item There are no vertices $u, v$ with $\codeg(u, v) \geq 2$ and $S(u,v) + S( N(u) \cap N(v) ) + \codeg(u, v) \geq 5$.
 \end{itemize}
 
Note that, when these conditions hold, then no 5-vertex $u$ can have subcubic neighbors (else we would have $S(u) \geq 1$).  Unlike in the degree-4 case, we cannot ensure that $G$ has no four-cycles. However, we can show that the four-cycles are heavily restricted.

\begin{proposition}
\label{u1propa} 
After Line 1, any pair of non-adjacent vertices $x,y$ have $\codeg(x, y) \leq 2$.
\end{proposition}
\begin{proof}
Suppose $\deg(x) \geq \deg(y)$ and $\codeg(x,y) \geq 3$. If $\deg(x) = 3$, then $\{ x,y \}$ would be an indset with surplus one, contradicting that $G$ is simplified. If $\deg(y) = 4$, then $\shad(N[y]) \leq 0$ (due to vertex $x$), and we would branch at Line 1. So we suppose $\deg(x) = \deg(y) = 5$. Again, if $\codeg(x, y) \geq 4$ then $\shad(N[y]) \leq 0$ due to vertex $x$, so we suppose that $\codeg(x,y) = 3$ exactly.  Then $S(N(x) \cap N(y)) \geq 2$ due to non-adjacent 2-vertices $x,y$. So  $S(x,y) + S(N(x) \cap N(y)) + \codeg(x,y) \geq 5$ and we would have branched at Line 1.
\end{proof}

\begin{proposition}
\label{u1prop} 
After Line 1, any 4-cycle $x_1, x_2, x_3, x_4$ satisfies the following properties:
\begin{itemize}
\vspace{-0.05in}
\item All the vertices $x_i$ have degree at least four.
\vspace{-0.07in}
\item If $x_1 \sim x_3$ then $\deg(x_1) = \deg(x_3) = 5$.
\vspace{-0.07in}
\item At least two of the vertices $x_i$ have degree five.
\end{itemize}
\end{proposition}
\begin{proof}
For the first result, suppose that $\deg(x_1) = 3$. If $x_1 \sim x_3$ or $x_2 \sim x_4$, then $x_1$ would have a 3-triangle. So suppose $x_1 \not \sim x_3$ and $x_2 \not \sim x_4$. By \Cref{u1propa}, this means that $x_2, x_4$ are the only shared neighbors of $x_1, x_3$.  If $\deg(x_3) = 4$, then $\shad(N[x_3]) \leq 0$ (due to 1-vertex $x_1$), contradicting that $G$ is simplified. So, suppose that $\deg(x_3) = 5$. Since $x_2, x_4$ are neighbors of 5-vertex $x_3$ and 3-vertex $x_1$, they must have degree 4 (as 5-vertices cannot be neighbors of 3-vertices). At this point, observe that $S(x_1, x_3) \geq 2$ (due to non-adjacent 2-vertices $x_2, x_4$) and $S( N(x_1) \cap N(x_3) ) = S(x_2, x_4) \geq 1$ (due to 1-vertex $x_1$); we would branch at Line 1.

For the second result, suppose without loss of generality that $\deg(x_1) = 4$ and $x_1 \sim x_3$ and $\deg(x_2) \geq \deg(x_4)$; we have already shown that $\deg(x_4) \geq 4$. We have $x_2 \not \sim x_4$, as otherwise $x_1$ has a funnel (due to clique $x_2, x_3, x_4$). If $\deg(x_2) = 5$, then $S(x_2) \geq 1$ (due to 3-triangle $x_1, x_3, x_4$). Otherwise, if $\deg(x_2) = \deg(x_4) = 4$, then $S( N(x_1) \cap N(x_3) ) \geq 1$ (due to subquadratic vertex $x_2$) and $S(x_1,x_3) \geq 2$ (due to non-adjacent 2-vertices $x_2, x_4$);  we would branch at Line 1.

For the final result, suppose that $\deg(x_1) = \deg(x_2) = \deg(x_3) = 4$. We have already shown then that $x_1 \not \sim x_3$, and \Cref{u1propa} gives $\codeg(x_1, x_3) \leq 2$. So $G - \{x_2, x_4 \}$ has two non-adjacent subquadratic vertices $x_1, x_3$. Thus $S(N(x_1) \cap N(x_3)) = S(x_2, x_4) \geq 2$ by \Cref{subquartic-simp}, and $S(x_1, x_3) \geq 1$ due to subquadratic vertex $x_2$.  We would branch at Line 1.
\end{proof}

\begin{proposition}
\label{4nprop05}
After Line 1, the neighbors of any 4-vertex $u$ can be enumerated as $v_1, \dots, v_4$ to satisfy one of the following three conditions:
\begin{itemize}
\vspace{-0.05in}
\item[(I)] $v_1, \dots, v_4$ are all subquartic.
\vspace{-0.07in}
\item[(II)] $\deg(v_1) = 5$ and $v_3, v_4$ are subquartic, and $G[N(u)]$ has exactly one edge $(v_1, v_2)$.
\vspace{-0.07in}
\item[(III)] $\deg(v_1) = 5$, and all the vertices $v_1, \dots, v_4$ are independent.
\end{itemize}
Moreover, after Line 4,  in case (III) we may further suppose that $v_2, v_3, v_4$ are subquartic.
\end{proposition}
\begin{proof}
Suppose $u$ has a 5-neighbor $v_1$. By \Cref{u1prop}, we must have $\codeg(u, v_1) \leq 1$. Also, if there are any edges among the other neighbors of $u$, then $S(v_1) \geq 1$ due to the resulting 3-triangle on $u$. In Case II, if say $\deg(v_3) = 5$, then we would have $S(v_3) \geq 1$ (due to the 3-triangle on $u, v_1, v_2$), and we would have applied the Line-1 branching rules. In Case III, if say $\deg(v_2) = 5$, then $u$ would have two non-adjacent 5-neighbors and we would branch at Line 4.
\end{proof}

\begin{proposition}
\label{4nprop0}
After Line 4, if a 5-vertex $u$ has at least four 5-neighbors, then the induced neighborhood $G[N(u)]$ is a matching.
\end{proposition}
\begin{proof}
Suppose $x \sim u$ and $u,x$ share neighbors $y_1, y_2$; let $v_1, v_2$ be the two remaining neighbors of $u$, where $\deg(v_1) = 5$. Note that $G - \{v_1, v_2 \}$ contains a kite $u, x, y_1, y_2$. Since we did not branch on $v_1, v_2$ at Line 3, we must have $\deg(v_2) = 4, v_1 \sim v_2$. Since $u$ has at least four 5-neighbors, this implies that  $\deg(y_1) = \deg(y_2) = \deg(x) = 5$.  

The 4-vertex $v_2$ must be in Case II of \Cref{4nprop05}. So $\{y_1, y_2, x \} \not \sim v_2$. Then $y_1 \not \sim v_1$, as otherwise $G - \{y_2, v_2 \}$ would have a kite, where $\deg_{G-y_2}(v_2) =  4$, and we would branch at Line 3. So  $S(v_1, y_1) \geq 1$ (due to the 3-triangle $u, x, y_2$) and $S(v_1, N[y_1]) \geq 1$ (due to 2-vertex $v_2$), and we would branch on non-adjacent 5-vertices $v_1, y_1$ at Line 2.
\end{proof}

\begin{lemma}
\label{n45lem2}
After Line 6, the graph has $n_4 \geq n_5$.
\end{lemma}
\begin{proof}
Let $A_{51}, A_{5+}$ be the set of 5-vertices with exactly one 4-neighbor or more than one 4-neighbor respectively, and let $B$ be the set of 4-vertices with exactly one 5-neighbor. Note that $A_{51}, A_{5+}$ partition $V_5$ at this point; for, if a 5-vertex $u$ has five independent 5-neighbors, we would branch on it at Line 6; if it has a triangle in its neighborhood, we would branch on $u$ at Line 5. 

We claim that, for any vertex in $A_{51}$ its sole 4-neighbor must be in $B$, and hence $|B| \geq |A_{51}|$. For, suppose that $u \in A_{51}$ has a neighbor $x \in B$ and $x$ has a second 5-neighbor $v$; if $u \not \sim v$ we would branch on $x$ at Line 4 and if $u \sim v$ we would branch on $u$ at Line 5. 

There are at least  $|A_{51}|+ 2 |A_{5+}|$ edges between $V_4$ and $V_5$. Conversely, each vertex $v \in B$ has exactly one 5-neighbor, and by \Cref{4nprop05} any 4-vertex $v$ has at most two 5-neighbors. So $|B| + 2(|V_4| - |B|) \leq |A_{51}| + 2 |A_{5+}|$.  Since $|B| \geq |A_{51}|$, we conclude that $|V_4| = |B| + (|V_4| - |B|) \geq  |A_{51}| + |A_{5+}| = |V_5|.$
\end{proof}

\begin{observation}
\label{line7structprop}
After Line 8 of algorithm {\tt Branch5-2}, any 5-vertex $x$ has $\shad(N[x]) \geq 2$, and has no 3-vertex within distance two. Any vertex $y \not \sim x$ has $\codeg(x,y) \leq \deg(y) - 3$.
\end{observation}
\begin{proof}
If $\shad(N[x]) \leq 1$, then after Line 1 we would have $S(N[x]) \geq 1$ and we would branch on $x$ at Line 8.  So  $\deg_{G - N[x]}(y) \geq 3$, i.e. $\codeg(x,y) \leq \deg(y) - 3$. We have already seen that $x$ has no 3-neighbor; this also shows it cannot have any 3-vertex at distance two. 
\end{proof}
\begin{lemma}
\label{codeg45lem}
After Line 9 of algorithm {\tt Branch5-2}, any non-adjacent vertices $x,y$ with $\deg(x) = 5, \deg(y) = 4, \codeg(x,y) = 1$ satisfy $\shad(N[x,y]) \geq 1$. 
\end{lemma}
\begin{proof}
Let $G' = G - N[x,y]$. We first claim that $\mindeg(G') \geq 2$. For, consider a vertex $t \in G'$. If $\deg_G(t) = 3$, then by \Cref{line7structprop} we have $\codeg(x,t) = 0$ and by \Cref{u1prop} we have $\codeg(y,t) \leq 1$. So $\deg_{G'}(t) \geq 2$. Likewise, if $\deg_G(t) = 5$, then $\codeg(y,t) \leq 1$ and $\codeg(x, t) \leq 2$. So again $\deg_{G'}(t) \geq 2$.

Next suppose $\deg_G(t) = 4$; by \Cref{u1propa}  and \Cref{line7structprop} we have $\codeg(y,t) \leq 2$ and $\codeg(x,t) \leq 1$.  To get $\deg_{G'}(t) = 1$, vertex $t$ must share two neighbors $z_1, z_2$ with $y$ and share a distinct neighbor $z_3$ with $x$. By \Cref{u1prop} we must have $\deg(z_1) = \deg(z_2) = 5$. If $z_1 \not \sim z_2$, then $S(z_1, z_2) \geq 1$ (due to subquadratic vertex $y$), and we would branch on $z_1, z_2$ at Line 9. If $z_1 \sim z_2$, then in the graph $G - N[x]$, vertex $t$ loses its neighbor $z_3$, and so it has a 3-triangle $z_1, z_2$; in particular, $S(N[x]) \geq 1$ and we would branch on $x$ at Line 8.

Thus, in all cases, $\deg_{G'}(t) \geq 2$. Now, suppose $G'$ has an indset $I$ with $\surp_{G'}(I) \leq 0$. Since $\mindeg(G') \geq 2$, necessarily $|I| \geq 2$. Then $\surp_{G - N[x]}(I \cup \{y \}) \leq 0 + \deg_{G - N[x]}(y) - 1 = 2$. So $G - N[x]$ would be 3-vul and we would branch on $x$ at Line 8.
\end{proof}

\begin{observation}
\label{line7structprop2}
After Line 10 of algorithm {\tt Branch5-2}, any 4-vertex $v$ has at most one 5-neighbor and is not part of any 4-cycle.
\end{observation}
\begin{proof} 
 If $v$ has two 5-neighbors we would branch at Line 4 or Line 10, depending if they were adjacent or not. Now consider a four-cycle $v, x_1, x_2, x_3$. By \Cref{line7structprop}, we have $\deg(x_2) \leq 4$. Thus, by \Cref{u1prop}, we have $\deg(x_1) = \deg(x_3) = 5$; that is, $v$ has two 5-neighbors, which is what we have showed is impossible.
\end{proof}

\subsection{Basic branching rules}
We analyze it line by line; Lines 6 and 11, which are more involved, will be handled later.

\medskip

\noindent \textbf{Line 1:} 
We first simplify $G$.  If a vertex has degree at least 6, we  apply \Cref{branch5-3}.  If a 5-vertex $u$ has $\shad(N[u]) \leq 0$ we apply \Cref{branch6-1}. Otherwise if a 5-vertex $u$ has $S(u) \geq 1$, then we split on $u$ giving branch-seq $[(0.5,2),(2,5)]$. If a 4-vertex $u$ has $\shad(N[u]) \leq 0$, then we apply \Cref{branch4-3}. If vertices $u, v$ have $S(u, v) + S(N(u) \cap N(v))) + \codeg(u, v) \geq 5$, we apply (B-Shared).

Over all cases, we get branch-seq  $[ (1,3), (1,4) ]$, $[ (1,2), (1,5) ]$, $[ (0.5,2), (2,5) ]$, or $[ (0.5,1), (2.5,6) ]$. Accordingly, we need:
\begin{align}
&\max \{ e^{-a-3b} + e^{-a-4b},  e^{-a-2b} + e^{-a-5b},  e^{-0.5 a - 2b} + e^{-2a - 5b},  e^{-0.5 a - b} + e^{-2.5 a - 6 b} \} \leq 1  \tag{5-3}
\end{align}

\noindent \textbf{Line 2:} We split on $u, v$, obtaining subproblems $G_1 = G - u, G_2 = G - u - N[v], G_3 =  G - \{u, v \}$. They have drop $(2,5), (2.5,6), (1,2)$ directly. After simplification,  $G_2$ has drop $(2.5,7)$ and $G_3$ has drop $(1,3)$. So we need:
\begin{align}
e^{-2a - 5 b} + e^{-2.5a - 7 b} + e^{-a-3b} \leq 1 \tag{5-4}
\end{align}

\noindent \textbf{Line 3:} We split on $u, v$  and apply (P3) to $G - \{u, v \}$. Since we are after Line 1, subproblems $G - N[u]$ and $G - u - N[v]$ both have drops $(2,5)$. Also, $\shad(u,v) \geq \shad(u) - 1 \geq 1$ so by \Cref{p5rule1} the subproblem $G - \{u, v\}$ has drop $(1.5,4)$ after simplification.  We need:
\begin{align}
e^{-2a - 5 b} + e^{-2a - 5 b} + e^{-1.5a-4b} \leq 1 \tag{5-5}
\end{align}

\noindent \textbf{Line 4:} Consider a 4-vertex $u$ with non-adjacent 5-neighbors $y_1, y_2$. This must be Case III of \Cref{4nprop05}, where all neighbors of $u$ are independent; let $x_1, x_2$ be the other two neighbors of $u$, sorted so that $\deg(x_1) \leq \deg(x_2)$.  We split on $y_1, y_2$, generating subproblems $G_1 = G - N[y_1], G_2 = G - y_1 - N[y_2], G_3 = G - \{y_1, y_2 \}$.  Since we are after Line 1,  the subproblems have drops $(2,5), (2.5,6), (1,2)$ directly.

 If $\deg(x_1)  = 3$, then $S(y_1, y_2) \geq 1$ (due to 2-vertex $u$) and $S(y_1, N[y_2]) \geq 1$ (due to 2-vertex $x_1$); in this case we would branch on $y_1, y_2$ at Line 2. So  suppose that $\deg(x_1) \geq 4$. We can apply \Cref{branch5-5} to 2-vertex $u$ in $G_3$ with $r = \deg(x_1) + \deg(x_2) - \codeg(x_1, x_2) - 1$.   If $\deg(x_1) = \deg(x_2) = 4$, then \Cref{u1prop} gives $\codeg(x_1, x_2) = 1$ (they share 4-neighbor $u$). If $\deg(x_2) = 5$, then \Cref{u1propa} gives $\codeg(x_1, x_2) \leq 2$. In either case $r \geq 6$ and so we need:
\begin{align}
e^{-2a - 5 b} + e^{-2.5a - 6 b} + e^{-a-2b} \psi_6 \leq 1 \tag{5-6} \label{t513}
\end{align}

\noindent \textbf{Line 5:} Vertices $v_3, v_4, v_5$ cannot form a clique (else $S(v_1) \geq 1$ due to the funnel $u$), so assume without loss of generality $v_3 \not \sim v_4$. Then we split on $v_3, v_4$, generating subproblems $G_1 = G - N[v_3], G_2 = G - v_3 - N[v_4], G_3 = G - \{v_3, v_4 \}$, which have drops $(2,5), (2.5,6), (1,2)$ directly.

In $G_3$, we also apply \Cref{branch5-5} to vertices $v_1, v_5$ for the 3-triangle  $u, v_1, v_2$. By \Cref{4nprop0}, at this point $\codeg(v_1, u) \leq 1$ and $\codeg(v_5, u) \leq 1$. So $v_1 \not \sim v_5$ and $\deg_{G_3}(v_1) = 5, \deg_{G_3}(v_5) \geq 4$ and, by \Cref{u1propa}, we have $\codeg(v_1, v_5) \leq 2$. So $\deg_{G_3}(v_1) + \deg_{G_3}(v_5) - \codeg_{G_3}(v_1, v_5) - 1 \geq 6$.  Overall we need
\begin{align*}
e^{-2a - 5 b} + e^{-2.5a - 6 b} + e^{-a-2b} \psi_6 \leq 1
\end{align*}
which is just a repetition of (\ref{t513}).

\smallskip

\noindent \textbf{Line 8:} Since we are after Line 1, subproblem $G - N[u]$ has drop $(2,5)$ directly, and subproblem $G - u$ has drop $(0.5,1)$. We apply \Cref{almostred} to $G - N[u]$, so we need:
 \begin{align}
 & e^{-2 a - 5 b} \gamma_3  + e^{-0.5a-b} \leq 1 \qquad \text{(when $L = 2$)} \tag{5-7} 
 \end{align}

\noindent \textbf{Line 9:} We get subproblems $G - N[u], G - u - N[v],  G - \{u, v \}$. Since we are after Line 1,  these have drops $(2,5), (2.5,6), (1,2)$ directly, and the latter has drop $(1,3)$ after simplification. We need:
 \begin{align}
 & e^{-2 a - 5 b}  + e^{-2.5a - 6 b} + e^{-a-3b} \leq 1 \qquad \text{(when $L = 2$)} \tag{5-8} 
 \end{align}
 
 \noindent \textbf{Line 10:} Let us branch on a 4-vertex $u$ with neighbors  $x_1, x_2, y_1, y_2$ where $\deg(x_1) = 5$ and $x_1 \sim x_2$. This is Case II of  \Cref{4nprop05}:  $u$ has no other edges in its neighborhood and $y_1, y_2$ are subquartic. Since $x_2, y_1, y_2$ are within distance two of 5-vertex $x_1$ we have $\deg(y_1) = \deg(y_2) = 4$ and so by \Cref{line7structprop} we have $\codeg(x_1, y_2) = \codeg(x_1, y_1) = 1$.

Here, we split on $y_1$, and then apply (B-Fun2) to $x_1$ in subproblem $G - y_1$. This generates subproblems $G_1 = G - N[y_1], G_2 = G - y_1 - N[y_2, x_1]$, and $G_3 = G - \{y_1, x_1\}$. Subproblem $G_1$ has drop $(1.5,4)$. By \Cref{codeg45lem}, $\minsurp(G_2) \geq \shad(N[y_2, x_1]) - 1 \geq 0$, so  $G_2$ has drop $(3.5,9)$. Subproblem $G_3$ has drop $(1,2)$ directly, and has a 2-vertex $u$ so it has drop $(1,3)$ after simplification. We need:
\begin{align}
&e^{-1.5 a - 4b}  +e^{-3.5a-9b} + e^{-a-3b} \leq 1   \qquad \text{(when $L = 2$)} \tag{5-9}
\end{align}

\noindent \textbf{Line 12:}  We use a more refined property of the MaxIS-5 algorithm of  \cite{xiao2016exact}; specifically, their algorithm has runtime $O^*( 1.17366^{w_3 n_3 + w_4 n_4 + n_5})$ for constants $w_3 = 0.5093, w_4 = 0.8243$.  For {\tt Branch5-3}, \Cref{n45lem2} implies that $n_4 \geq n_5$, and so we need:
\begin{align}
 \min\{ a_{52} \mu + b_{52} k, 2 (k - \mu) \cdot (w_4 \cdot 1/2 + 1/2) \log 1.17366  \} \leq a \mu + b k \tag{5-10}
\end{align}
which can be checked routinely by LP algorithms.

For {\tt Branch5-2}, each 5-vertex has at least three 4-neighbors (else we would branch at Line 11), and each 4-vertex has at most one 5-neighbor. So $n_4 \geq 3 n_5$, and we need
\begin{align}
 \min\{ a_{51} \mu + b_{51} k, 2 (k - \mu) \cdot (w_4 \cdot 3/4 + 1/4) \log 1.17366  \} \leq a \mu + b k \tag{5-11}
\end{align}
which can also be checked mechanically.

\subsection{Line 6: A 5-vertex with five independent 5-neighbors} 
\label{line7sec} 
Our basic plan is to choose an ordering of the neighbors of $u$ as $v_1, \dots, v_5$, and then split on $v_1, v_2, v_3$ to create a 2-vertex $u$. For Section~\ref{line7sec}, we use the following notations:
\begin{eqnarray*}
&A = \{v_1,v_2, v_3 \}, \qquad B = \{v_4, v_5 \}, \qquad d_{45} = \codeg(v_4, v_5) \\
&H = G - \{v_1, v_2, v_3 \} - N[v_4, v_5] = G - A - N[B] \\
&I = \text{a chosen min-set of $H$}, \qquad s = |N(I)|, \qquad \ell = \minsurp(H) = \surp_H(I)
\end{eqnarray*}

The properties of the graph $H$ are central to this analysis; note that $d_{45} \in \{1,2 \}$ by \Cref{u1propa}.  We begin with a few observations.

\begin{proposition}
\label{hh-obs2}
If there is any vertex $x \in H$ with $\deg_H(x) \leq 5 - \deg_G(x)$, then $u$ has neighbors $v_i, v_j$ with $S(v_i, N[v_j]) \geq 1$.
\end{proposition}
\begin{proof}
 If $\deg_G(x) = 3$, then $x \not \sim A$; to get $\deg_H(x) \leq 2$, we must have $\codeg(x, v_j) \geq 1$ for $v_j \in B$, so  $S( v_1,  N[v_j]) \geq 1$.  If $\deg_G(x) = 4$, then $|N(x) \cap A| \leq 1$ (else we would branch on $x$ at Line 4); so to get $\deg_H(x) \leq 1$, we have either $\codeg(x, v_j) \geq 2$ for $v_j \in B$, or $x \sim A$. In the former case, we have  $S(v_1, N[v_j]) \geq 1$ (due to subquadratic vertex $x$). In the latter case, if $\codeg(x, v_4) = \codeg(x, v_5) = 1, x \sim v_i \in A$, then $S( v_i,  N[v_4]) \geq 1$  (due to subquadratic vertex $x$).  
 
 Finally, suppose $\deg_G(x) = 5$ and $x$ is isolated in $H$. By \Cref{u1propa} we have $\codeg(x, u) \leq 2$ and $\codeg(x, v_4) \leq 2$ and $\codeg(x, v_5) \leq 2$. If $\codeg(x, v_4) = \codeg(x, v_5) = 2$, then $x$ has a neighbor $v_i \in A$, so  $S(v_i,  N[v_4]) \geq 1$ (due to subquadratic vertex $x$).  Otherwise, if $\codeg(x, v_j) = 2$ for $v_j \in B$ and  $x \sim v_i \in A$, then $S( v_i, N[v_j]) \geq 1$ (due to subquadratic vertex $x$)
\end{proof}

\begin{proposition}
\label{hh-obs22x}
If $s = 1$ and $\ell = -1$, then one of the following holds:
\begin{itemize} 
\vspace{-0.05in}
\item There are neighbors $v_i, v_j$ where $G - v_i - N[v_j]$ is 3-vul.
\vspace{-0.07in}
\item $d_{45} = 1$ and there are neighbors $v_i, v_{i'}, v_j$ with $S(v_i, v_{i'}, N[v_j]) \geq 1$.
\vspace{-0.07in} 
\item  $d_{45} = 2$ and the graph $G - N[u]$ has a 5-vertex with a 3-neighbor.
\end{itemize}
\end{proposition}
\begin{proof}
We assume that $I = \{x_1, x_2 \}$ for 5-vertices $x_1, x_2$ where $\deg_H(x_1) = \deg_H(x_2) = 1$, as otherwise the result follows from \Cref{hh-obs2}. Let $y$ be the single shared neighbor of $x_1, x_2$ in $H$. Since $y$ has non-adjacent 5-neighbors $x_1, x_2$, we have $\deg_G(y) = 5$.  By \Cref{u1propa}, each $x \in I$ has at most two neighbors in $A$ and shares at most two neighbors with each $v_j \in B$.

If $I \not \sim A$, then $\surp_{G - N[v_4, v_5]}(I) = -1$, and hence $\surp_{G - v_1 - N[v_4]}(I \cup \{v_5 \}) \leq 2$, and $G - v_1 - N[v_4]$ is 3-vul via in the indset $J = I \cup \{v_5 \}$.

If some $x \in I$ has exactly one neighbor $v_{i} \in A$, then $\codeg(x, v_j) \geq 2$ for some $v_j \in B$ and $S( v_i,  N[v_j]) \geq 1$ (due to subquadratic vertex $x$).  Similarly, if there is some  $x \in I$ with two neighbors $v_i, v_{i'}$ in $A$, then $\codeg(x, v_j) \geq 1$ for some $v_j \in B$, and  so $S( v_i, v_{i'} ,  N[v_j]) \geq 1$. 

The only remaining case if $d_{45} = 2$ and every $x \in I$ has at least two neighbors in $A$. If $|N(I) \cap A| < 3$, say without loss of generality $v_1 \not \sim I$, then  $\surp_{G - \{ v_2, v_3 \} - N[v_4, v_5]} (I) = -1$ and hence $\surp_{G - v_2 - N[v_4]}(I \cup \{v_5 \}) \leq 2$; so $G - v_2 - N[v_4]$ is 3-vul (via indset $I \cup \{v_5 \})$.  So we may suppose without loss of generality that the edges between $A$ and $I$ are precisely $(v_1, x_1), (v_2, x_2), (v_3, x_1), (v_3, x_2)$. 

Now $v_3$ has three 5-neighbors $u, x_1, x_2$ and $y, v_3$ share neighbors $x_1, x_2$. So by \Cref{4nprop0} we have $y \not \sim v_3$ (else $v_3$ would have four 5-neighbors and its neighborhood would be a non-matching). If $y \sim v_i$ for $v_i \in \{v_1,  v_2 \}$, then $S(v_i, N[v_3]) \geq 1$ due to subquadratic vertex $y$.  Otherwise, if $y \not \sim A$, then in the graph $G - N[u]$, vertices $x_1, x_2$ remain with degree 3 (losing neighbors $v_1, v_2, v_3$), while $y$ remains with degree 5. 
\end{proof}

\medskip

We now analyze the branching rule and subproblems. When describing each case, we assume that no other vertex ordering could go into an earlier case.

\smallskip

\noindent \textbf{Case I:  $\boldsymbol{\shad(v_1, v_2, N[v_3]) \leq -1.}$} Let $G' = G - \{v_1, v_2 \} - N[v_3]$ and let $J$ be a min-set of $G'$. We claim that $\mindeg(G') \geq 1$.  For, consider a vertex $x \in G'$. By \Cref{u1propa}, we have $\codeg(x, v_3) \leq 2$. If $\deg(x) = 5$, this implies $\deg_{G'}(x) \geq 1$. If $\deg(x) = 4$,  then $|N(x) \cap \{v_1, v_2 \}| \leq 1$ (else we would branch on $x$ at Line 4), so again $\deg_{G'}(x) \geq 1$. Finally, if $\deg(x) = 3$, then $x \not \sim \{v_1, v_2 \}$ and $\deg_{G'}(x) \geq 1$.  

So necessarily $|J| \geq 2$. Then $\surp_{G - N[v_3]}(J) \leq 1$, which implies that $S(N[v_3]) \geq |J| \geq 2$. Our strategy is to split on $v_3$, generating subproblems $G - v_3, G - N[v_3]$ with drops $(0.5,1), (2,7)$ after simplification. We need:
\begin{align}
& e^{-0.5a-b} +  e^{-2a-7b} \leq 1   \tag{5-12}
\end{align}

\noindent \textbf{Case II:  $\boldsymbol{S(v_1, v_2, N[v_3]) \geq 3.}$} We split on $v_1, v_2, v_3$, generating subproblems $G - N[v_1], G - v_1 - N[v_2], G - \{v_1, v_2 \} - N[v_3], G - \{v_1, v_2, v_3 \}$. Since we are not in Case I, the first three have drops $(2,5), (2.5,6), (3,7)$ directly. Note that $\shad(v_1, v_2, v_3) \geq \shad(v_1) - 2 \geq 0$. So subproblem $G - \{v_1, v_2, v_3 \}$ has drop $(1.5,3)$ directly, and we can apply \Cref{branch5-5} to its 2-vertex $u$ where $\deg(v_4) + \deg(v_5) - d_{45} - 1 \geq 7$. So we need:
\begin{align}
& e^{-2a-5b} + e^{-2.5a - 6b} + e^{-3a-10b} + e^{-1.5a-3b} \psi_7 \leq 1   \tag{5-13}
\end{align}

\noindent \textbf{Note:} After Case II, we suppose that $\ell \geq -3 + d_{45}$. For, suppose $\surp_H(I) \leq -4 + d_{45} \leq -2$, where necessarily $|I| \geq 2$. Then $\surp_{G - \{v_1, v_2 \} - N[v_4]}(I \cup \{v_5 \}) \leq 1$. So either $\shad(v_1, v_2, N[v_4]) \leq -1$ (which would put us into Case I with alternate vertex ordering $v_1, v_2, v_4$) or, by \Cref{simp-obs1} we have $S(v_1, v_2, N[v_4]) \geq |I \cup \{v_5 \}| \geq 3$, putting us into Case II.

Our strategy in all subsequent cases is to split on $v_1, v_2, v_3$, generating subproblems $G_1 = G - N[v_1], G_2 = G - v_1 - N[v_2], G_3 = G - \{v_1, v_2 \} - N[v_3]$, and $G - \{v_1, v_2, v_3 \}$; we then apply (B-Pair) to the 2-vertex $u$ in $G - \{v_1, v_2, v_3 \}$, to get subproblems  $G_0 = G - N[u]$ and $H = G - \{v_1, v_2, v_3 \} - N[v_4, v_5 ]$. Since we are not in Case I, subproblems $G_0, G_1, G_2, G_3$ have drops $(2,5), (2,5), (2.5,6)$ and $(3,7)$  respectively, and the final subproblem $H$ has drop 
\begin{align*}
&(5.5 - d_{45}/2,  \ \ \ \ \ \ \ \ 13 - d_{45}) && \text{if $\ell \geq 0$} \\
&( 5.5 - d_{45}/2 + \ell/2,   13 - d_{45}  + s  ) && \text{if $\ell < 0$ (after applying P1)}
\end{align*}

\medskip

\noindent \textbf{Case III: $\boldsymbol{G_3}$ is 2-vul, and either $\boldsymbol{G_1}$ or $\boldsymbol{G_2}$ is 3-vul.} Here, $H$ has drop $(4,11)$. We can apply \Cref{almostred} to $G_3$ and $G_1$ or $G_2$; overall, in the two cases,  we need:
\begin{align}
& e^{-2 a - 5 b} + \max\{ e^{-2 a -5   b} \gamma_3 + e^{-2.5 a - 6 b}, e^{-2a - 5b} + e^{-2.5 a - 6 b} \gamma_3 \}  + e^{-3a - 7 b} \gamma_2 + e^{-4a -11 b} \leq 1  \tag{5-14}
\end{align}

\medskip

\noindent \textbf{Case IV: Either $\boldsymbol{G_1}$ or $\boldsymbol{G_2}$ is 3-vul.} If $\ell = -3 + d_{45}$, then $\surp_{G - \{v_1, v_2 \} - N[v_4]}(I \cup \{v_5 \}) \leq 2$ and so $G - \{v_1, v_2 \} - N[v_4]$ is 2-vul; this would put us into Case III via vertex ordering $v_1' = v_1, v_2' = v_2, v_3' = v_4$.  So we suppose $\ell \geq -2 + d_{45}$ and hence $H$ has drop $(4.5,11)$. We need:
\begin{align}
& e^{-2 a - 5 b} + \max\{ e^{-2 a -5  b}  \gamma_3 + e^{-2.5 a - 6 b}, e^{-2a - 5b} + e^{-2.5 a - 6 b} \gamma_3 \} + e^{-3a - 7 b} + e^{-4.5a -11 b} \leq 1  \tag{5-15}
\end{align}

\medskip

\noindent \textbf{Note:} After Case IV,  we must have $s \geq 1$; for, if $s = 0$, then $H$ has an isolated vertex; by \Cref{hh-obs2} there are then neighbors  $v_i, v_{j}$ with $S(v_i, N[v_j]) \geq 1$, putting us into Case IV.

\medskip

\noindent \textbf{Case V: $\boldsymbol{G_3}$ is 4-vul.} Recall that $\ell \geq -3 + d_{45}$.  First, suppose $d_{45} = 2, \ell = -1, s = 1$, so $H$ has drop $(4,12)$. By \Cref{hh-obs22x}, there are either  vertices $v_i, v_j$ where $G - v_i - N[v_j]$ is 3-vul (putting us into Case IV), or $G_0$ has a 5-vertex with a 3-neighbor. In the latter case, we can either simplify $G_0$ further, or apply \Cref{branch55} to $G_0$, and we then need:
\begin{align}
& e^{-2 a - 5 b} \max\{e^{-b}, e^{-a-3b} +e^{-a-5b}, e^{-0.5a - 2 b} + e^{-2a-5b} \}  \notag \\
& \qquad \qquad +  e^{-2 a -5  b} + e^{-2.5 a - 6 b} + e^{-3 a - 7 b} \gamma_4 + e^{-4 a - 12 b} \leq 1  \tag{5-16}
\end{align}

Otherwise, if $\ell \geq -2 + d_{45}$ or $s \geq 2$ or $d_{45} = 1$, then  $H$ has drop $(4.5,11)$ or $(4,13)$; we need:
\begin{align}
& e^{-2 a - 5 b} + e^{-2 a -5  b} + e^{-2.5 a - 6 b} + e^{-3 a - 7 b} \gamma_4 + \max\{ e^{-4 a - 13 b}, e^{-4.5 a - 11 b} \} \leq 1  \tag{5-17}
\end{align}

\noindent \textbf{Case VI: $\boldsymbol{d_{45} = 1}$.} If $\ell \leq -2$, then $\surp_{G - \{v_1, v_2 \} - N[v_4]}(I \cup \{v_5 \}) \leq 2$ and $|I| = s - \ell \geq 3$, so $G - \{v_1, v_2 \} - N[v_4]$ is 4-vul; this would put us into Case V.   If $s = 1, \ell = -1$, then by \Cref{hh-obs22x}, there are either neighbors $v_i, v_{i'}, v_j$ with $S(v_i, v_{i'}, N[v_j]) \geq 1$ (putting us into Case V) or $G - v_i - N[v_j]$ is 3-vul (putting us into Case IV).  So either $\ell \geq 0$ or $\ell = -1, s \geq 2$, and thus  $H$ has drop $(5,12)$ or $(4.5,14)$ respectively. We need:
\begin{align}
& e^{-2 a - 5 b} + e^{-2 a -5  b} + e^{-2.5 a - 6 b} +  e^{-3 a - 7 b} +  \max\{  e^{-4.5a - 14 b}, e^{-5a-12 b} \}  \leq 1  \tag{5-18}
\end{align}

\noindent \textbf{Case VII: Everything else.}  Since no ordering of the vertices $v_1, \dots, v_5$ puts us into Case V, we suppose that every pair $v_i, v_j$ shares some neighbor $x_{ij} \neq u$. The vertices $x_{ij}$ are distinct, since otherwise $\codeg(x_{ij},u) \geq 3$ for some $i,j$, contradicting \Cref{u1propa}.  So the distance-two vertices from $u$ are precisely $\{x_{ij} : 1 \leq i < j \leq 5 \}$. Then indset $I = \{v_2, \dots, v_5 \}$ has precisely $6$ neighbors in $G - N[v_1]$, namely vertices  $x_{ij}$ for $2 \leq i < j \leq 5$.  So $G - N[v_1]$ is 4-vul via indset $I$; this would be covered in Case IV.  So Case VII is impossible.

\subsection{Line 11: A 5-vertex with fewer than three  4-neighbors}
\label{gvg2}
We suppose we are in {\tt Branch5-2}.  Let $U$ be the non-empty set of 5-vertices with fewer than three 4-neighbors. For any vertex $u \in U$, define the induced neighborhood $H_u = G[N(u)]$. Note that any 4-vertex is isolated in $H_u$, as otherwise we would apply Line 10 to it with its 5-neighbor $u$. The following theorem characterizes the possible neighborhood structure (see Figure~\ref{fig444}).

\renewcommand\thesubfigure{(\Roman{subfigure})}
\vspace{0.53in}
\begin{figure}[H]
\centering
\hspace{0.5in}
\subfigure[]{
\includegraphics[trim = 0.5cm 23.5cm 7cm 4cm,scale=0.5,angle = 0]{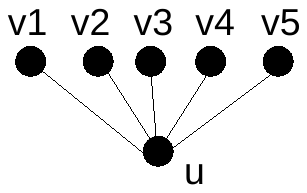}
}
\subfigure[]{
\includegraphics[trim = 0.5cm 23.5cm 7cm 4cm,scale=0.5,angle = 0]{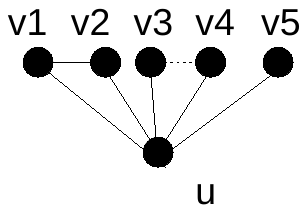}
}
\vspace{1.4in}

\subfigure[]{
\includegraphics[trim = 0.5cm 22.5cm 9cm 4cm,scale=0.5,angle = 0]{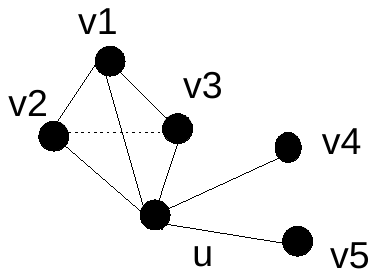}
}
\subfigure[]{
\includegraphics[trim = 0.5cm 19cm 9cm 8cm,scale=0.5,angle = 0]{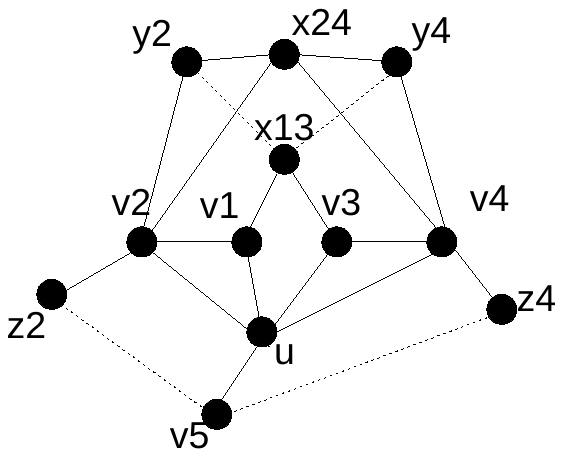}
}
\vspace{0.1in}

\caption{The four cases for the neighborhood of $u$. Dashed edges may or may not be present. \label{fig444}}
\end{figure}

\begin{theorem}
\label{vav5new}
There is a vertex $u \in U$ and an enumeration of its neighbors $v_1, \dots, v_5$ to satisfy one of the following three conditions:
\begin{enumerate}
\vspace{-0.05in}
\item[(I)]  $v_1, \dots, v_5$ are independent and $\deg(v_1) = \deg(v_2) = \deg(v_3) = 5$.
\vspace{-0.07in}
\item[(II)] $H_u$ is a matching with an edge $(v_1, v_2)$ and possibly an edge $(v_3, v_4)$, where $\deg(v_1) = \deg(v_2) = \deg(v_3) = 5, \deg(v_5) = 4, \deg_{G - v_3}(v_4) = 4$ and $\codeg(v_1,v_4) = \codeg(v_2, v_4) = \codeg(v_1, v_5) = \codeg(v_2, v_5) = \codeg(v_3, v_4) = 1$. 
\vspace{-0.07in}
\item[(III)] $v_1$ has two neighbors $v_2, v_3$, where $\deg(v_1) = 5, \deg(v_4) = \deg(v_5) = 4$ and $\{v_4, v_5 \} \not \sim v_1$.
\vspace{-0.07in}
\item[(IV)] $H_u$ is a matching with edges $(v_1, v_2)$ and $(v_3, v_4)$, where $\deg(v_1) = \deg(v_2) = \deg(v_3) = \deg(v_4) = 5$. Furthermore, there are 5-vertices $x_{13}, x_{24}, y_2, y_4$ and 4-vertices $z_2, z_4$, all distinct from each other and from  $N[u]$, such that $v_2, x_{24}, y_2$ and $v_4, x_{24}, y_4$ are both triangles in $G$ and $x_2 \sim x_{24} \sim x_4$  and $x_{13} \not \sim \{ x_{24}, v_2, v_4 \}$ and $v_2 \sim z_2, v_4 \sim z_4$.
\end{enumerate}
\end{theorem}
\begin{proof}
 If $H_u$ has no edges, it falls into the Case I (choose $v_1, v_2, v_3$ to be 5-neighbors of $u$). 

Suppose $H_u$ is not a matching. Let $v_1$ be a neighbor which shares two other neighbors $v_2, v_3$ with $u$. Let $v_4, v_5$ be the other neighbors of $u$. Note that $G - \{v_4, v_5 \}$ has a kite. Since 4-vertices are isolated in $H_u$, the only possible structure where we would not use the branching rule at Line 3 is $\deg(v_1) = 5, \deg(v_4) = \deg(v_5) = 4$; since 4-vertices are isolated in $H_u$ necessarily $\{v_1, v_2, v_3 \} \not \sim \{v_4, v_5 \}$. This is Case III.

Suppose $H_u$ is a matching with exactly one edge $(v_1, v_2)$, and an additional 5-vertex $v_3$. If $u$ has any other 5-neighbor $t$, then $S(v_3, t) \geq 1$ (due to 3-triangle $u, v_1, v_2$) and we would branch at Line 9. So $u$ must have two 4-neighbors $v_4, v_5$. Then \Cref{line7structprop2} implies $\codeg(v_1, v_4) = \codeg(v_2, v_4) = \codeg(v_1, v_5) = \codeg(v_2, v_5) = \codeg(v_3, v_4) = 1$. This is Case II.

Thus, for the remainder of the proof, we suppose that, for every vertex $u \in U$, the graph $H_u$ is a matching with exactly two edges, where the one unmatched neighbor has degree 4. 

Consider some $u \in U$, where $H_u$ has edges $(v_1, v_2), (v_3, v_4)$ and has a 4-vertex $v_5$. By \Cref{line7structprop} we automatically have $\codeg(v_1, v_5) = \codeg(v_2, v_5) = 1$.  We claim that $\codeg(v_3, v_4) = 1$. For, suppose $v_3, v_4$ share neighbor $t \neq u$; since $t$ has two 5-neighbors necessarily $\deg(t) = 5$. Then $v_3 \in U$  since it has three 5-neighbors $u, v_4, t$, but $\deg_{H_{v_3}}(v_4) \geq 2$, which contradicts our assumption on structure of $U$. By completely symmetric arguments also $\codeg(v_1, v_2) = 1$.

Now consider forming a bipartite graph $J$ with left-nodes $v_1, v_2$ and right-nodes $v_3, v_4$, with an edge $(v_i, v_j)$ if $\codeg_G(v_i, v_j) = 2$. If any node has degree zero in $J$, we may suppose without loss of generality it is $v_4$, in which case case $\codeg(v_1, v_4) = \codeg(v_2, v_4) = 1$ giving us Case II.

Otherwise, if every node in $J$ has at least one neighbor, then by Hall's theorem  $J$ has a perfect matching;  by swapping vertices as needed we assume this matching is $(v_1, v_3), (v_2, v_4)$, that is, $v_1, v_3$ share a neighbor $x_{13} \neq u$ and $v_2, v_4$ share a neighbor $x_{24} \neq u$. Here, $\deg(x_{13}) = \deg(x_{24}) = 5$ since they each have two 5-neighbors.  Clearly $u \not \sim \{x_{13}, x_{24} \}$. By \Cref{line7structprop}, we have $\codeg(x_{13}, u) \leq 2$ and $\codeg(x_{24}, u) \leq 2$, so $x_{13} \not \sim \{v_2, v_4 \}$ and $x_{24} \not \sim \{v_1, v_3 \}$. In particular, $x_{13} \neq x_{24}$.  Also,  $x_{13} \not \sim x_{24}$, as otherwise $x_{13}$ would have three independent 5-neighbors $x_{24}, v_1, v_3$, contradicting our assumption of structure of $U$.

Now  $v_2 \in U$ as it has three 5-neighbors $v_1, u, x_{24}$ and similarly $v_4 \in U$ as it has three 5-neighbors $v_3, u, x_{24}$. By our hypothesis on the structure of $U$, it must be that $v_2$ has an edge $(x_{24}, y_2)$ in its neighborhood and has a 4-neighbor $z_2$, and likewise $v_4$ has an edge $(x_{24}, y_4)$ in its neighborhood and has a 4-neighbor $z_4$.  By \Cref{line7structprop}, we have $\codeg(x_{24}, u) \leq 2$, and so $y_2, y_4 \notin N[u]$. Also $y_2 \neq y_4$, as otherwise the vertex $y = y_2 = y_4$ would have $\codeg(v_2, v_4) \geq 3$, sharing neighbors $u, x_{24}, y$. Since $x_{13} \not \sim x_{24}$, it must be that $y_2, y_4$ are distinct from $x_{13}$. Likewise, we have $z_2 \neq z_4$ since otherwise the vertex $z = z_2 = z_4$ would have two 5-neighbors $v_2, v_4$.    This gives Case IV.
\end{proof}

\medskip

We thus choose a vertex $u$ and an enumeration of its neighbors $v_1, \dots, v_5$ to satisfy \Cref{vav5new}. We then use the following branching rules:

\smallskip

\noindent \textbf{Case I.} Split on $v_1, v_2, v_3$, getting subproblems $G_1 = G - N[v_1], G_2 = G - v_1 - N[v_2], G_3 = G - \{v_1, v_2 \} - N[v_3], G_4 = G - \{v_1, v_2, v_3 \}$. By \Cref{line7structprop}, these have drops $(2,5), (2.5,6), (3,7), (1,3)$ directly. In $G_4$, we apply \Cref{branch5-5} to 2-vertex $u$, where by \Cref{line7structprop} and \Cref{line7structprop2} we have $r = \deg_{G_4} (v_4) + \deg_{G_4}(v_5) - \codeg_{G_4}(v_4, v_5) - 1 \geq 6$ irrespective of the degrees of $v_4, v_5$. We need:
\begin{align}
&e^{-2 a - 5b} +e^{-2.5a-6b} +  e^{-3a-7b} + e^{-1.5a-3b}  \psi_6 \leq 1 \tag{5-19}
\end{align}

\noindent \textbf{Case IIa: $\boldsymbol{\shad( v_3,  N[v_1,v_4]) \geq 1}$.} Split on $v_3, v_5$ and apply (B-Fun2) to vertex $v_1$ in $G - \{v_3, v_5 \}$. This generates subproblems $G_1 = G - N[v_3], G_2 = G - v_3 - N[v_5], G_3 = G - \{v_3, v_5 \} - N[v_1, v_4], G_4 = G - \{v_1, v_3, v_5 \}$.   Since we are after Line 1, and $\deg_G(v_3) = 5, \deg_{G}(v_5) = 4, v_3 \not \sim v_5$, the subproblems $G_1, G_2$ both have drop $(2,5)$ directly. Since $\codeg(v_1, v_4) = 1$, and since we are in Case IIa with $\minsurp(G_3) \geq 0$, the subproblem $G_3$ has drop $(4,10)$.  Subproblem $G_4$ has drop $(1.5,3)$ directly, and we also apply \Cref{branch5-5} to 2-vertex $u$ in $G_4$, noting that $\deg_{G_4}(v_2) + \deg_{G_4}(v_4) - 1 - \codeg_{G_4}(v_2, v_4) = 4 + 4 - 1 - 1 = 6$.

Overall, we need:
\begin{align}
&e^{-2 a - 5b} +e^{-2a-5b} + e^{-4a-10b} + e^{-1.5a-3b} \psi_6 \leq 1 \tag{5-20}
\end{align} 

\noindent \textbf{Case IIb: $\boldsymbol{\shad( v_3,  N[v_1,v_4]) \leq 0}$.}  Split on $v_3, v_4$,  then apply (B-Fun2) to vertex $v_1$ in $G - \{v_3, v_4 \}$. We get subproblems $G_1 = G - N[v_3], G_2 = G - v_3 - N[v_4], G_3 = G - \{v_3, v_4 \} - N[v_1, v_5], G_4 = G - \{v_1, v_3, v_4 \}$.   Subproblems $G_1, G_2$ both have drop $(2,5)$. 

In subproblem $G_2$, consider vertex $v_1$; it has degree 4 since $\codeg(v_1, v_4) = 1$ and $v_1 \not \sim v_3$. Case IIb gives $\shad_{G_2}(v_1) = \shad( v_3, N[v_1,v_4]) \leq 0$, that is, $v_1$ is a blocked 4-vertex in $G_2$. Also, $\minsurp(G_2) \geq \shad( N[v_4]) -1 \geq 0$. So, either $S(G_2) \geq 1$ or we can apply \Cref{branch4-3} to $G_2$.

For subproblem $G_3$, \Cref{codeg45lem} gives $\minsurp(G_3) \geq \shad(N[v_1, v_5]) - 2 \geq -1$. Also, since $\codeg(v_3, v_4) = 1$, we have $\mindeg(G_3) \geq \mindeg(G - N[v_1,v_5]) - 1 \geq 1$.   Now $G_3$ is principal with $\Delta k = 10$ and excess two; if $\minsurp(G_3) \geq 0$, it has drop $(4,10)$ and if $\minsurp(G_3) = -1$ it has drop $(3.5,11)$ after applying (P1).

Finally, subproblem $G_4$ has drop $(1.5,3)$ directly. We apply \Cref{branch5-5} to 2-vertex $u$ in $G_4$, noting that $\deg_{G_4}(v_2) + \deg_{G_4}(v_5) - 1 - \codeg(v_2, v_5) = 4 + 4 - 2 = 6$. Overall we need:
\begin{align}
&e^{-2 a - 5b} +e^{-2a-5b} \max\{e^{-b}, e^{-a-3b} +e^{-a-5b}, e^{-0.5a - 2 b} + e^{-2a-5b} \}  \notag \\
& \qquad \qquad  + \max\{ e^{-3.5a-11b},e^{-4a-10b} \} + e^{-1.5a-3b} \psi_6  \leq 1 \tag{5-21}
\end{align}

\noindent \textbf{Case III:}  We claim that there is some good cover $C$ satisfying one of the following three cases: (i) $v_1 \in C$; or (ii) $v_1 \notin C, v_4 \notin C$; or (iii) $v_1 \notin C, v_4 \in C, v_5 \notin C$.   For, if not, then consider a good cover $C$ where $v_1 \notin  C, v_4 \in C, v_5 \in C$. Then $C' = C \cup \{v_1 \} \setminus \{ u \} $ would be another cover, of the same size, which would be in case (i).

So our strategy is to branch on subproblems $G_1 = G - v_1, G_2 = G - N[v_1, v_4], G_3 = G - v_4 - N[v_1, v_5]$. Here, $G_1$ has drop $(0.5,1)$.  By \Cref{codeg45lem}, $\minsurp(G_2) = \shad(N[v_1, v_4]) \geq 1$ and $\minsurp(G_3) \geq \shad(N[v_1,v_5]) - 1 \geq 0$,  so $G_2, G_3$ have drops $(3,8), (3.5,9)$ respectively. We need:
\begin{align}
&e^{-0.5a-b} + e^{-3a-8b} + e^{-3.5a-9b}  \leq 1 \tag{5-22}
\end{align}

\noindent \textbf{Case IV:}  Split on the independent 5-vertices $u, x_{13}, x_{24}$, generating subproblems $G_1 = G - N[u], G_2 = G - u - N[x_{13}], G_3 = G - \{ u, x_{13} \} -  N[x_{24}], G_4 = G - \{x_{14}, x_{23}, u \}$. By \Cref{line7structprop}, we have  $\minsurp(G_2) \geq \shad(N[x_{13}]) - 1 \geq 1$ and $\minsurp(G_3) \geq \shad(N[x_{24}]) - 2 \geq 0$ and $\minsurp(G_4) \geq \shad(x_{14}) - 2 \geq 0$. So these have drops $(2,5), (2.5,6), (3,7), (1.5,3)$ directly.

Note that both $z_2, z_4$ remain in the graph $G_2$, since $z_2, z_4$ are 4-vertices which already have a 5-neighbors $v_2$ or $v_4$ respectively. We claim that $S(G_2) \geq 2$; there are a few different cases depending on $| N(x_{13})  \cap \{y_2, y_4 \}|$. First, suppose $x_{13} \sim y_2$ and $x_{13} \sim y_4$; then, in the graph $G_2$, vertices $v_2$ and $v_4$ have degree two (losing neighbors $v_1, u, y_2$ and $v_3, u, y_4$ respectively). So by \Cref{simp-obs0} we immediately have $S(G_2) \geq 2$. Next, suppose $|N(x_{13}) \cap \{y_2, y_4 \}| = 1$, say without loss of generality $x_{13} \sim y_2, x_{13} \not \sim y_4$. In this case, we apply (P2) to the 2-vertex $v_2$ in $G_2$; in the resulting graph $G_2'$, the vertices $z_2,  x_{24}$ get merged into a new vertex $z'$. We then have $N_{G_2'}(v_4) = \{ z_4, z', y_4 \}$ where $z' \sim y_4$, since $x_{24} \sim y_4$ in $G_2$. So we can apply (P3) to the 3-triangle on $v_4$, yielding $S(G_2') \geq 1$ and hence $S(G_2) \geq 2$. Finally suppose $x_{13} \not \sim \{y_2, y_4 \}$; then in $G_2$, the vertex $v_2$ has degree 3 with a triangle $v_1, y_2$.  We apply (P3) to $v_2$; let $G_2'$ denote the resulting graph, where we remove $z_2$ and add edges from $v_1, y_2$ to the neighbors of $z_2$. Since $v_4 \not \sim v_2$, the degree of $v_4$ does not change in $G_2'$, and in particular $v_4$ still has degree 3 with a triangle $v_3, y_4$. So $S(G_2') \geq 1$ and $S(G_2) \geq 2$.

Consider graph $G_3$; here, vertex $v_1$ loses neighbors $u, x_{13}, v_2$ and likewise $v_3$ loses neighbors $u, x_{13}, v_4$. Thus,  by \Cref{simp-obs0} we have $S(G_3) \geq 2$.

So $G_2$ has drop $(2.5,8)$ after simplification and $G_3$ has drop $(3,9)$ after simplification. We need:
\begin{align}
&e^{-2 a - 5b} +e^{-2.5a-8b} +e^{-3a-9b} + e^{-1.5a-3b}  \leq 1 \tag{5-23}
\end{align}

\section{Branching on 6-vertices}
Similar to the degree-5 setting, we will have a few algorithms targeted for different densities of degree-6 vertices in the graph. We will show the following result:
\begin{theorem}
 {\tt Branch6}$(G, k)$ has measure $a_6 \mu + b_6 k$ for $a_6 = \aSixThree, b_6 = \bSixThree$.
 \end{theorem}
 
Combined with \Cref{lem:combinelem}, this gives the $O^*(\dSix^k)$ algorithm for degree-6 graphs.  The first algorithm is very simple, and is targeted toward the regime where degree-6 vertices are sparse. 

\begin{algorithm}[H]
Simplify $G$

If $\maxdeg(G) \leq 5$, then run the cheaper of  {\tt Branch5-3} or the MaxIS-5 algorithm.

Apply \Cref{branch5-3} to a vertex of degree at least $6$

\caption{Function {\tt Branch6-1}$(G, k)$}
\end{algorithm}

\begin{proposition}
\label{branch6-alt}
 {\tt Branch6-1} has measure $a_{61} \mu + b_{61} k$ for $a_{61} = \aSixOne, b_{61} = \bSixOne$. 
 \end{proposition}
 \begin{proof}
 Given the potential branch-seqs from \Cref{branch5-3}, we have the following constraints:
 \begin{align}
& \max\{e^{-a - 3 b} + e^{-a - 5 b},   e^{-0.5 a - 2 b} + e^{-2 a - 5 b},  e^{-0.5 a - b} + e^{-2.5a - 6 b} \} \leq 1 \tag{6-1} \\
&\min\{ a_{53} \mu + b_{53} k, 2 (k - \mu) \log(1.17366) \} \leq a \mu + b k \tag{6-2} 
\end{align}
The latter constraint can be routinely checked by LP algorithms.
\end{proof}

The more advanced algorithms try to branch on multiple 6-vertices in order to get bonus simplifications. These algorithms differ only in a few small steps, so we we describe them jointly. 

\begin{algorithm}[H]

Apply miscellaneous branching rules to ensure $G$ has nice structure.

Split on a 6-vertex $u$ where $G - N[u]$ is 2-vul.

Branch on vertices $u, v$ with $\deg_G(u) = 6, \deg_{G-u}(v) \geq 5$ and $S(u, v) + S(u, N[v]) \geq 1$.

Branch on non-adjacent 6-vertices $x_1, x_2, x_3$ with $S(x_1, x_2, x_3) \geq 1$.

Branch on a 6-vertex $u$ with six 6-neighbors, or five 6-neighbors and one 5-neighbor.

\If{$L = 2$} {

Branch on  non-adjacent 6-vertices $x_1, x_2$ with $S(x _1,  N[x_2]) \geq 1$.

Branch on a 5-vertex $u$ which has two 6-neighbors  and three 5-neighbors,  \ \ \ \ \ \ \ \ \ \         UNLESS each of these 6-neighbors itself has at least four 5-neighbors.
}

Run the cheaper of  {\tt Branch6-$(L-1) (G,k)$} or the MaxIS algorithm

\caption{Function {\tt Branch6-$L$}$(G, k)$ for $L = 2,3$}
\end{algorithm}

We will show the following:

\begin{theorem}
\label{61advanced}
{\tt Branch6-2} has measure $a_{62} \mu + b_{62} k$ for $a_{62} = \aSixTwo, b_{62} = \bSixTwo$.

{\tt Branch6-3} has measure $a_{63} \mu + b_{63} k$ for $a_{63} = \aSixThree, b_{63} = \bSixThree$

\end{theorem}

\subsection{Regularizing the graph and structural properties}
In Line 1, we will ensure that the following conditions hold:
\begin{itemize}
\vspace{-0.05in}
\item $G$ is simplified.
\vspace{-0.07in}
\item $\maxdeg(G) \leq 6$.
\vspace{-0.07in}
\item Every 6-vertex $u$ has $\shad(N[u]) \geq 2$ and $S( u) = 0$.
\vspace{-0.07in}
\item Every 5-vertex $u$ has $\shad(N[u]) \geq 1$ and $S( u) = 0$.
\vspace{-0.07in}
 \item There are no vertices $u, v$ with $\codeg(u,v) \geq 2$ and $S(u,v) + S(N(u) \cap N(v)) +  \codeg(u, v) \geq 5$.
 \end{itemize}
 
\begin{observation}
 \label{6base-obs2}
 \leavevmode 
 \begin{itemize}
 \vspace{-0.05in}
 \item After Line 1, any 6-vertex $u$ and vertex $v \not \sim u$ have $\codeg(u,v) \leq \deg(v) - 3$.
 \vspace{-0.07in}
 \item After Line 1, any 5-vertex $u$ and vertex $v \not \sim u$ have $\codeg(u,v) \leq 2$.
 \vspace{-0.07in}
 \item After Line 3, any vertices $u, v$ with $\deg(u) = 6$ and $ \deg_{G - u}(v) \geq 5$ have $\shad(u,v) \geq 2$ and $\shad(u, N[v]) \geq 1$. 
 \vspace{-0.07in}
 \item After Line 4, any non-adjacent 6-vertices $x_1, x_2, x_3$ have $\shad(x_1, x_2, x_3) \geq 2$.
 \end{itemize}
 \end{observation}
 \begin{proof}
  In the first case, if $\codeg(u,v) \geq \deg(v) - 2$, then $\shad(N[u]) \leq 1$. In the second case, we may suppose that also $\deg(v) \leq 5$  (else it is covered by the first case); then the analysis is the same as \Cref{u1propa}. In the third case, if $\shad(u,v) \leq 1$ or $\shad(u, N[v]) \leq 0$, then we would have $S(u,v) \geq 1$ or $S(u, N[v]) \geq 1$ (note that $\shad(u, N[v]) \geq \shad(N[v]) - 1 \geq 0$), and we would branch at Line 3. In the final case, we have $\shad( x_1, x_2, x_3 ) \geq \shad(x_1, x_2) - 1 \geq 1$. If $\shad(x_1, x_2, x_3) \leq 1$, then we would branch at Line 4.
 \end{proof}

\begin{proposition}
\label{tthhh1}
After Line 4, any 5-vertex $u$ has at most three 6-neighbors; moreover, if $u$ has three 6-neighbors, its other two neighbors have degree 4.
\end{proposition}
\begin{proof}
Suppose $u$ has 6-neighbors $v_1, v_2, v_3$. Let $v_4, v_5$ be the other two neighbors of $u$. They must have degree at least 4, else $S( u) \geq 1$ and we would have branched at Line 1.  If $v_1, v_2, v_3$ are independent, then $S(v_1, v_2, v_3) \geq 1$ due to 2-vertex $u$, and we would apply Line 4 to $v_1, v_2, v_3$. So we may suppose without loss of generality that $v_1 \sim v_2$. 

Now suppose for contradiction that $\deg(v_4) \geq 5$. If $v_3 \sim v_4$, then $S(v_1, v_2) \geq 1$ due to 3-triangle $u, v_3, v_4$, and we would branch on $v_1, v_2$ at Line 3. Otherwise, if $v_3 \not \sim v_4$, then $S( v_3, v_4) \geq 1$ due to 3-triangle $u, v_1, v_2$, and we would branch on $v_3, v_4$ at Line 3.
\end{proof}

\begin{proposition}
\label{6obs1}
Before Line 9 in the algorithm {\tt Branch6-2}, the graph has $5 n_4 + n_5 \geq 2 n_6$. \\
Before  Line 9 in the  algorithm {\tt Branch6-3}, the graph has $5 n_4 + 2 n_5 \geq 2 n_6$. 
\end{proposition}
\begin{proof}
For $i \geq 0$, let $A_{6i}$ be the set of 6-vertices with exactly $i$ 5-neighbors and let $A_{5i}$ be the set of 5-vertices with exactly $i$ 6-neighbors. By \Cref{tthhh1}, $A_{54} = A_{55} = \emptyset$ and any vertex in $A_{53}$ has at least two 4-neighbors.  Also let $A_{52}' \subseteq A_{52}$ be the set of 5-vertices with two 6-neighbors and at least one 4-neighbor.  By counting edges between $V_5$ and $V_6$, we have:
\begin{equation}
\label{gas1}
|A_{51}| + 2 |A_{52}| + 3 |A_{53}|  = |A_{61}| + 2 |A_{62}| + 3 |A_{63}| + 4 |A_{64}| + 5 |A_{65}| + 6 |A_{66}|
\end{equation}

Due to the branching rule at Line 5, any vertex in $A_{60} \cup A_{61}$ has at least one 4-neighbor. Each 4-vertex has at most one 6-neighbor (else we would branch on the 6-neighbors at Line 3), and so:
\begin{equation}
\label{gas0}
|V_4| \geq |A_{60}| + |A_{61}|
\end{equation}

Now consider the edges between $V_4$ and $V_5 \cup V_6$. Each vertex in $A_{60} \cup A_{61}$ contributes one such edge, each vertex in $A_{52}'$ contributes at least one edge,  each vertex in $A_{53}$ contributes at least two edges, and each vertex in $V_4$ has at most four such edges. So we get:
\begin{align}
4 |V_4| &\geq |A_{60}| + |A_{61}| + 2 |A_{53}| + |A_{52}'| \label{gas3}
\end{align}

Adding inequalities (\ref{gas1}), (\ref{gas0}),  (\ref{gas3}) and collecting terms gives:
\begin{align}
\label{gas5}
5 |V_4| + |A_{51}| + 2 |A_{52}| + |A_{53}| - |A_{52}'| \geq 2 |A_{60}| + 3 |A_{61}| + 3 |A_{63}| + 4 |A_{64}| + 5 |A_{65}| + 6 |A_{66}|
\end{align}

Since $|V_5| =  |A_{51}| + |A_{52}| + |A_{53}|$ and $|V_6| = |A_{60}| + |A_{61}| + \dots + |A_{66}|$, this implies $5 n_4 + 2 n_5 \geq 2 n_6$.

In the algorithm {\tt Branch5-2}, further consider any vertex $u \in A_{52} \setminus A_{52}'$.  Both the 6-neighbors of $u$ must be in $A_{64} \cup A_{65} \cup A_{66}$ (else we would have branched on $u$ at Line 8).  By double-counting the edges between $A_{52}\setminus  A_{52}'$ and $V_6$, we get:
\begin{align}
\label{gas6}
4 |A_{64}| + 5 |A_{65}| + 6 |A_{66}| \geq 2 ( |A_{52}| - |A_{52}'| )
\end{align}

Now adding inequality (\ref{gas5}) to $\tfrac{1}{2}$ times inequality (\ref{gas6}), and collecting terms again, gives:
\begin{align}
\label{gas9}
5 |V_4| + |A_{51}| + |A_{52}| + |A_{53}| \geq 2 |A_{60}| + 3 |A_{61}| + 3 |A_{63}| + 2 |A_{64}| + 2.5 |A_{65}| + 3 |A_{66}| 
\end{align}
which in turn implies $5 n_4 + n_5 \geq 2 n_6$ as desired.
\end{proof}

\subsection{Basic analysis}

We begin by analyzing the easier steps of the algorithm; Lines 5 and 8 will be handled later.

 \medskip
 
\noindent  \textbf{Line 1.}  We first simplify $G$.  If a vertex $u$ of degree 5 or 6 has $\shad(N[u]) \leq 5 - \deg(u)$, we apply \Cref{branch6-1}. If a vertex $u$ of degree 5 or 6 has $S(u) \geq 1$, we split on $u$ with branch-seq $[(0.5,2), (2,5)]$. If a vertex has degree at least 7, we apply \Cref{branch5-3}. If there are vertices $u, v$ with $S(u, v) + S(N(u) \cap N(v)) + \codeg(u, v) \geq 5$, we apply (B-Shared).  

Over all cases, we get branch-seq $[ (1,3), (1,4) ]$, $[ (1,2), (1,5) ]$, or $[ (0.5,2), (2,5) ]$. We need:
 \begin{align}
& \max\{e^{-a - 3 b} + e^{-a - 4 b}, e^{-a-2b} + e^{-a-5b}, e^{-0.5 a - 2 b} + e^{-2 a - 5 b} \} \leq 1 \tag{6-3} \label{t61}
\end{align}

\noindent \textbf{Line 2.} After Line 1, subproblems $G - u$ and $G - N[u]$ have drop $(0.5,1)$ and $(2.5,6)$ directly, and we can apply \Cref{almostred} to the latter problem. We need:
\begin{align}
& e^{-0.5a-b} + e^{-2.5a - 6b} \gamma_2 \leq 1 \tag{6-4}
\end{align}

\noindent  \textbf{Line 3.} Split on $u$, and in subproblem $G - u$, either apply \Cref{branch6-1} or split on $v$. After Line 1, subproblems $G - N[u], G-u$ have drops $(2.5,6), (0.5,1)$ directly and are simplified. We need:
\begin{align}
e^{-2.5a-6b} + e^{-0.5a-b} \max\{ e^{-b},  e^{-a-3b} + e^{-a-4b}, e^{-0.5a-2b} + e^{-2a-5b} \} \leq 1 \tag{6-5}
\end{align}

\noindent \textbf{Line 4.}  Split on $x_1, x_2, x_3$, getting subproblems $G_1 = G - N[x_1], G_2 = G - x_1 - N[x_2], G_3 = G - \{x_1, x_2 \} - N[x_3], G_4 = G - \{x_1, x_2, x_3 \}$. By \Cref{6base-obs2} we have $\minsurp(G_1)  \geq 2, \minsurp(G_2) \geq 1, \minsurp(G_3) \geq \shad( N[x_3]) - 2 \geq 0$ and $\minsurp(G_4) \geq \shad(x_1, x_2) - 1 \geq 1$. Because $S(G_4) \geq 1$, these subproblems $G_1, G_2, G_3, G_4$ have drops $(2.5,6), (3,7), (3.5,8), (2,5)$. We need:
 \begin{align}
 & e^{-2.5 a - 6 b}  + e^{-3a - 7 b} + e^{-3.5a - 8b} + e^{-1.5a - 4 b} \leq 1 \tag{6-6} 
 \end{align}

\noindent \textbf{Line 7.} After Line 1, subproblems $G - N[x_1]$ and $G - x_1 - N[x_2]$ and $G - \{x_1, x_2 \}$ have drops $(2.5,6), (3,7), (1,2)$ directly, and $G - x_1 - N[x_2]$ has drop $(3, 8)$ after simplification. We need:
\begin{align}
& e^{-2.5a - 6b} + e^{-3a - 8b} + e^{-a-2b} \leq 1  \qquad   \text{(only for $L =2$)} \tag{6-7}
\end{align}

\noindent \textbf{Line 9.} We use a more refined property of the MaxIS-6 algorithm of \cite{XiaoN17}; specifically, it has runtime 
$O^*(1.18922^{w_3 n_3 + w_4 n_4 + w_5 n_5 + n_6})$ for constants $w_3 = 0.49969, w_4 = 0.76163, w_5 = 0.92401$.

For {\tt Branch6-2} we have $5 n_4 + n_5 \geq 2 n_6$.  It can be checked that the maximum value of 
$w_3 n_3 + w_4 n_4 + w_5 n_5 + n_6$, subject to constraints $n_3 + n_4 + n_5 + n_6 = n, n_i \geq 0, 5 n_4 + n_5 \geq 2 n_6$, occurs at $n_5 = 2 n/3, n_6 = n/3,  n_3 = n_4 = 0$. So the runtime is at most $1.18922^{n ( w_5 (2/3) + (1/3) )}$ and we need:
\begin{align}
&\min\{ a_{61} \mu + b_{61} k, 2 (k - \mu) \cdot (w_5 \cdot 2/3 + 1/3) \log(1.18922)  \} \leq a \mu + b k \tag{6-8}
\end{align}
which can also be checked mechanically.  

For {\tt Branch6-3} we have $5 n_4 + 2 n_5 \geq 2 n_6$.  Again, the maximum value of 
$w_3 n_3 + w_4 n_4 + w_5 n_5 + n_6$, subject to constraints $n_3 + n_4 + n_5 + n_6 = n, n_i \geq 0, 5 n_4 + 2 n_5 \geq 2 n_6$, occurs at $n_6 = n_5 = n/2, n_3 = n_4 = 0$. So the overall runtime is at most $1.18922^{n ( w_5  (1/2) + (1/2) )}$ and we need:
\begin{align}
&\min\{ a_{62} \mu + b_{62} k, 2 (k - \mu) \cdot (w_5/2 + 1/2) \log(1.18922)  \} \leq a \mu + b k \tag{6-9}
\end{align}
which can also be checked mechanically. This shows \Cref{61advanced}.

\subsection{Line 5: branching on a 6-vertex with high-degree neighbors}
Consider a vertex $u$ here; we denote the induced neighborhood graph by $H = G[ N(u)] $. 
\begin{proposition}
\label{vav66}
The neighbors of $u$ can be ordered as $v_1, \dots, v_6$ to satisfy one of the following conditions:
\begin{enumerate}
\vspace{-0.05in}
\item[(I)]  $v_1, \dots, v_5$ are independent and $\deg(v_1) = \deg(v_2) = \deg(v_3) = \deg(v_4) = \deg(v_5) = 6$.

\vspace{-0.07in}
\item[(II)] $H$ is a matching and $v_1 \sim v_2, v_3 \sim v_4, \deg_{G - v_6}(v_5) = 5, \deg(v_1) =  \deg(v_3) = \deg(v_4) = \deg(v_6) = 6$.
\vspace{-0.07in}
 \item[(III)] $v_4 \sim v_5, v_4 \sim v_6$ and  $\deg_{G}(v_1) = 6, \deg_{G - v_1}(v_2) \geq 5, \deg_{G - \{v_1, v_2 \}}(v_3) \geq 5$.
\end{enumerate}
\end{proposition}
\begin{proof}
First observe that $H$ cannot have any triangle $v_1, v_2, v_3$; for, in that case, we would have $S(v_4, v_5) \geq 1$ where $v_4, v_5$ are any other 6-neighbors of $u$; we would branch on $v_4, v_5$ at Line 3.

If $H$ has no edges, then let $v_1, \dots, v_5$ be any five 6-neighbors of $u$, and we get Case I.

If $H$ is a matching with exactly one edge $(v_1, v_2)$, with three other 6-neighbors $v_3, v_4, v_5$, we would have $S(v_3, v_4, v_5) \geq 1$ due to the 3-triangle $u, v_1, v_2$ and we would branch Line 4.

If $H$ is a matching with exactly two edges, suppose without loss of generality they are $(v_1, v_2), (v_3, v_4)$ where $\deg(v_4) = 6$.   If $\deg(v_5) = \deg(v_6) = 6$, then again have $S(v_4, v_5, v_6) \geq 1$ due to 3-triangle $u, v_1, v_2$, and we would branch at Line 4.  So suppose without loss of generality that $\deg(v_5) = 5$, in which case  the ordering $v_1, \dots, v_6$ satisfies Case II. 

If $H$ is a matching with exactly three edges, we can arrange them as $(v_1, v_2), (v_3, v_4), (v_5, v_6)$, where all vertices except $v_2$ have degree 6. This also satisfies Case II.

Next suppose $H$ is not a matching, and some vertex $v_4 \in H$ has neighbors $v_5, v_6 \in H$ where the other neighbors $v_1, v_2, v_3$ all have degree 6. Since $H$ has no triangle, we order these vertices so that $v_2 \not \sim v_3$, satisfying Case III.

Finally, suppose $H$ is not a matching and some $v_4 \in H$ has neighbors $v_5, v_6 \in H$ where $\deg(v_4) = \deg(v_5) = \deg(v_6) = 6$.  The induced subgraph $H' = G[\{v_1, v_2, v_3 \}]$ must be a matching, as otherwise it would  be covered by the argument in the preceding paragraph. So let $x$ be a vertex which is isolated in $H'$. If $\deg(x) = 6$, we set $v_1 = x$ and $v_2$ to be the other 5-vertex (if any) and $v_3$ to be the 6-vertex. If $\deg(x) = 5$, we set $v_1, v_2$ to be the 6-neighbors of $u$ and $v_3 = x$. 
\end{proof}

See Figure~\ref{fig446} for an illustration.
\renewcommand\thesubfigure{(\Roman{subfigure})}
\begin{figure}[H]
\centering
\hspace{-0.2in}
\subfigure[]{
\includegraphics[trim = 0.5cm 23.5cm 9cm 2cm,scale=0.5,angle = 0]{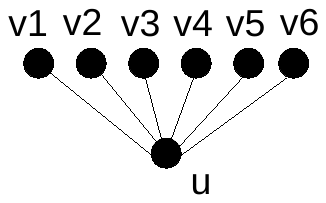}
}
\hspace{-0.3in}
\subfigure[]{
\includegraphics[trim = 0.5cm 24.7cm 11cm 2cm,scale=0.5,angle = 0]{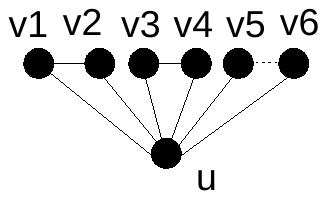}
}
\hspace{-0.3in}
\subfigure[]{
\includegraphics[trim = 0.5cm 22.5cm 9cm 2cm,scale=0.5,angle = 0]{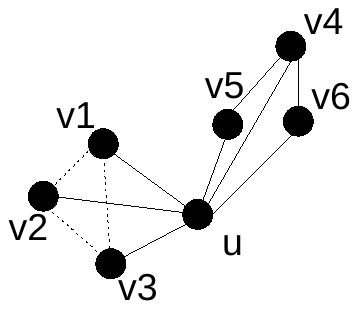}
}
\vspace{0.1in}

\caption{Possible cases for the neighborhood of $u$. Dashed edges may or may not be present. \label{fig446}}
\end{figure}

In light of \Cref{vav66}, we use the following branching rules for $u$. When analyzing a given case, we assume no earlier case was possible for any another vertex ordering.

\medskip

\noindent \textbf{Case Ia: $\boldsymbol{G - \{v_1, v_2 \} - N[v_3]}$ is 3-vul.}  Split on $v_1, \dots, v_4$, generating subproblems $G_1 = G - N[v_1], G_2 = G - v_1 - N[v_2], G_3 = G - \{v_1, v_2 \} - N[v_3], 
G_4 = G - \{v_1, v_2, v_3 \} - N[v_4], G_5 = G - \{v_1, v_2, v_3, v_4 \}$. By \Cref{6base-obs2}, we have $\minsurp(G_1) \geq 2$,  $\minsurp(G_2) \geq 1$, $\minsurp(G_3) \geq \shad( N[v_3]) -2 \geq 0$,   $\minsurp(G_4) \geq \shad( N[v_4]) - 3 \geq -1$ and $\minsurp(G_5) \geq \shad(v_1, v_2, v_3) - 1 \geq 1$. So $G_1, G_2, G_3, G_4, G_5$ have drop $(2.5,6), (3,7), (3.5,8), (3.5,9), (2,4)$ directly. We can apply \Cref{almostred} to $G_3$, and have $S(G_5) \geq 1$ due to 2-vertex $u$. We need:
\begin{align}
& e^{-2.5a - 6b} +  e^{-3a-7b} + e^{-3.5a-8b} \gamma_3 + e^{-3.5a - 9b} + e^{-2a-5b} \leq 1 \tag{6-10}
\end{align}

\noindent \textbf{Case Ib:} Split on $v_1, v_2, v_3, v_4$, and apply (B-Pair) to vertex $u$  in the subproblem $G - \{v_1, v_2, v_3, v_4 \}$. This generates  subproblems $G_1 = G - N[v_1], G_2 = G - v_1 - N[v_2], G_3 = G - \{v_1, v_2 \} - N[v_3], G_4 = G - \{v_1, v_2, v_3 \} - N[v_4], G_0 = G - N[u]$ and $G_5 = G - \{v_1, v_2, v_3, v_4 \} - N[v_5, v_6]$. 

 Since we are after Line 1, subproblems $G_0, G_1, G_2, G_3$ have drops $(2.5, 6), (2.5,6), (3,7), (3.5,8)$ directly. We must have $\minsurp(G_4) \geq 0$, as otherwise  $\shad(v_1, v_2, N[v_4]) \leq 0$, which would put us into Case Ia (with swapping the order of $v_3/v_4$). So $G_4$ has drop $(4,9)$. 
 
For subproblem $G_5$, let $d = \deg_{G - N[v_5]}(v_6)$; by \Cref{6base-obs2} we have $d \geq 3$. We claim that $\minsurp( G_5) \geq 2 - d$. For, consider an indset $I$ with $\surp_{G_5}(I) \leq 1 - d \leq -2$; necessarily $|I| \geq 2$.  Then $\surp_{G - \{v_1, v_2 \}- N [v_5]}(I \cup \{v_6 \} ) \leq (1-d) + (d-1) + 2 = 2$. So $G - \{v_1, v_2 \} - N[v_5]$ is 3-vul, putting us into Case Ia with alternate vertex ordering $v_1' = v_1, v_2' = v_2, v_3' = v_5$.    So $G_5$ is principal with excess two and $\Delta k = 10 + d \geq 13$ and $\Delta \mu \geq 5$.  

Putting all the subproblems together, we need:
\begin{align}
& e^{-2.5a - 6b} + e^{-2.5a - 6b} + e^{-3a-7b} + e^{-3.5a - 8b} + e^{-4a - 9b} + e^{-5a-13b} \leq 1 \tag{6-11}
\end{align}

\noindent \textbf{Case IIa: $\boldsymbol{G - v_6 - N[v_1]}$ is 3-vul.} Split on $v_6, v_1, v_5$, obtaining subproblems $G_1 = G - N[v_6], G_2 = G - v_6 - N[v_1], G_3 = G - \{v_1, v_6 \} - N[v_5], G_4 = G - \{v_1, v_6, v_5 \}$. By \Cref{6base-obs2}, we have $\minsurp(G_1) \geq 2, \minsurp(G_2) \geq 1, \minsurp(G_3) \geq \shad(v_6, N[v_5]) - 1 \geq 0, \minsurp(G_4) \geq \shad(v_1, v_6)  - 1 \geq 1$. So $G_1, G_2, G_3, G_4$ have drops $(2.5,6), (3,7), (3,7), (1.5,3)$ directly. We  apply \Cref{almostred} to $G_3$ and we have $S(G_4) \geq 1$ (due to 3-triangle $u, v_3, v_4$), so overall we need:
\begin{align}
& e^{-2.5a - 6b} + e^{-3a-7b} \gamma_3 + e^{-3a-7b} + e^{-1.5a-4b} \leq 1 \tag{6-12}
\end{align}

\noindent \textbf{Case IIb:} Split on $v_6, v_4, v_3$ and apply (B-Fun2) to vertex $v_1$ in subproblem $G - \{v_6, v_4, v_3 \}$. Overall, this generates subproblems $G_1 = G - N[v_6], G_2 = G - v_6 - N[v_4], G_3 = G - v_6 - N[v_3], G_4 = G - \{v_3, v_4, v_6, v_1 \}$ and $G_5 = G - \{v_3, v_4, v_6  \} - N[v_1, v_5 ]$. By \Cref{6base-obs2}, we have $\minsurp(G_1) \geq 2$ and $\minsurp(G_2) \geq 1$ and $\minsurp(G_3) \geq 1$ and $\minsurp(G_4) \geq \shad(v_4, v_6) - 2 \geq 0$. So the subproblems have drops $(2.5,6), (3,7), (3,7), (2,4)$ directly; subproblem $G_4$  has a 2-vertex $u$, so it has drop $(2,5)$ after simplification.

For subproblem $G_5$, let $d = \deg_{G - v_6 - N[v_1]}(v_5)$. We have $d \geq 3$, as otherwise $S( v_6, N[v_1]) \geq 1$ and it would be covered in Case IIa.  We must have $\minsurp(G_5) \geq 2 -d$; for, if we have an indset $I$ with $\surp_{G_5}(I)  \leq 1 - d \leq -2$, where necessarily $|I| \geq 2$, then  $\surp_{G - v_6  - N[v_1]}(I \cup \{v_5 \}) \leq 2$ and $G - v_6 - N[v_1]$ is 3-vul; it would be covered in Case IIa. 

Thus $G_5$ has $\Delta k = 9 + d \geq 12$ and $\Delta \mu \geq 4.5$. Putting all the subproblems together, we need:
\begin{align}
& e^{-2.5a - 6b} + e^{-3a-7b} + e^{-3a-7b} + e^{-2a-5b} + e^{-4.5a-12b} \leq 1 \tag{6-13}
\end{align}

\noindent \textbf{Case III.} Split on $v_1, v_2, v_3$, obtaining subproblems $G_1 = G - N[v_1],  G_2 = G - v_1 - N[v_2], G_3 = G - \{v_1, v_2 \} - N[v_3], G_4 = G- \{v_1, v_2, v_3 \}$.   By \Cref{6base-obs2}, we have $\minsurp(G_1) \geq 2$ and $\minsurp(G_2) \geq 1$ and $\minsurp(G_3) \geq \shad( v_1, N[v_3]) - 1 \geq 0$ and $\minsurp(G_4) \geq \shad( v_1,v_2) - 1 \geq 1$. So the subproblems have drops $(2.5,6), (2.5,6), (3,7), (1.5,3)$ directly; furthermore, $G_4$ has a kite $u, v_4, v_5, v_6$ and so it has net drop $(2,5)$ by \Cref{p5rule1}. We need:
\begin{align}
& e^{-2.5a - 6b} + e^{-2.5a - 6b} + e^{-3a-7b} + e^{-2a-5b} \leq 1 \tag{6-14}
\end{align}

\subsection{Line 8: Branching on a 5-vertex with two 6-neighbors} 
In the algorithm {\tt Branch6-2}, we consider a 5-vertex $u$ with 6-neighbors $y_1, y_2$ and 5-neighbors $v_1, v_2, v_3$, where either $y_1$ or $y_2$ has at most three 5-neighbors. 
See Figure~\ref{fig445}. 
\vspace{0.8in}
\begin{figure}[H]
\begin{center}
\hspace{1.5in}
\includegraphics[trim = 0.5cm 23.5cm 9cm 8cm,scale=0.5,angle = 0]{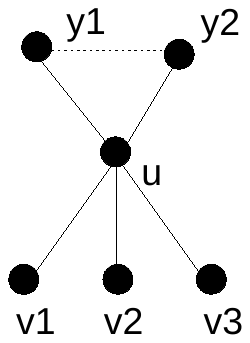}
\vspace{0.3in}
\caption{\label{fig445} The neighborhood of vertex $u$. The edge $(y_1, y_2)$ may or may not be present.}
\end{center}
\vspace{-0.17in}
\end{figure}
\begin{proposition}
\label{gfgft0}
\begin{itemize}
\item The vertices $v_1, v_2, v_3$ are independent and $\{y_1, y_2 \} \not \sim \{v_1, v_2, v_3 \}$.
\vspace{-0.07in}
\item If any pair $y_i, v_j$ share a vertex $x \neq u$, then $\deg(x) = 5$ and $x \not \sim u$.
\vspace{-0.07in}
\item There is some pair $y_i, v_j$ with $\codeg(y_i, v_j) = 1$.
\end{itemize}
\end{proposition}
\begin{proof}
\begin{itemize}
\item Vertices $v_1, v_2, v_3$ are independent, as otherwise $S(y_1, y_2) \geq 1$ (due to 3-triangle $u$) and we would branch at Line 3. There must be some pair $y_i, v_j$ with $y_i \not \sim v_j$ (else both $y_1, y_2$ would have four 5-neighbors $u, v_1, v_2, v_3$); say without loss of generality that $y_1 \not \sim v_1$. We have $y_2 \not \sim \{v_1, v_2, v_3 \}$, else $S(y_1, v_1) \geq 1$ due to the 3-triangle involving $u, y_2, v_j$, and we would apply Line 3 to $y_1, v_1$. This in turn implies that $y_1 \not \sim \{v_2, v_3 \}$, else $S(y_2, v_1) \geq 1$ due to the resulting 3-triangle and we would apply Line 3 to $y_2, v_1$.
\vspace{-0.07in}
\item Say without loss of generality that $y_1, v_1$ share a vertex $x \neq u$. Since none of the neighbors of $u$ are adjacent to $v_1$, we have $x \not \sim u$. If $\deg(x) \leq 4$, then $S(y_1, v_1) \geq 1$ (due to subquadratic vertex $x$), and we would branch on $y_1, v_1$ at Line 3. If $\deg(x) = 6$ and $x \sim y_2$, then $S(N[x]) \geq 1$ (due to vertex $u$ losing neighbors $y_1, v_1, y_2$) and we would branch at Line 2. If $\deg(x) = 6$ and $x \not \sim y_2$, then $S(y_2, N[x]) \geq 1$ (again due to vertex $u$ losing neighbors $y_1, v_1, y_2$), and we would branch on $y_2, x$ at Line 7.
\vspace{-0.07in}
\item Suppose $y_i$ shares neighbor $t_{ij} \neq u$ with $v_j$ for $j = 1,2,3$; as we have just seen, $\deg(t_{ij}) = 5$ and $t_{ij} \not \sim u$ for all $j$. Vertices $t_{i1}, t_{i2}, t_{i3}$ must be distinct; for, if $t_{ij} = t_{ij'} = t$, then $\codeg(t, u) \geq 3$ (sharing vertices $v_j, v_{j'}, y_i$), contradicting \Cref{6base-obs2}. So $y_i$ has four distinct 5-neighbors $t_{i1}, t_{i2}, t_{i3}, u$. By hypothesis, this does not hold for both $y_1$ and $y_2$. \qedhere
\end{itemize}
\end{proof}

\begin{proposition}
\label{gfgft1}
Suppose $y_1 \sim y_2$, and let $i \in \{1,2\},j \in \{1,2,3\}$. For the graph $G'_{ij} = G - \{v_1, v_2, v_3 \} - N[y_i, v_j]$, we have $\minsurp(G'_{ij}) \geq -3 + \codeg( y_i, v_j)$. Moreover, if $\codeg(y_i, v_j) = 1$ and $\minsurp(G'_{ij}) = -2$, then either $S( G'_{ij}) \geq 1$ or $S(v_{j'}, N[v_j]) \geq 1$ for some $j' \in \{1,2,3 \}$.
\end{proposition}
\begin{proof}
Suppose without loss of generality $i = j = 1$ and write $G' = G'_{ij} = G - \{v_2, v_3 \} - N[y_1, v_1]$ for brevity.  Let $I$ be a min-set of $G'$. If $\surp_{G'}(I) \leq -4 + \codeg(y_1, v_1)$, then $\surp_{G - N[y_1]}(I \cup \{v_1 \}) \leq 2$. So $G - N[y_1]$ is 2-vul (via indset $I \cup \{v_1 \}$) and we would branch on $y_2$ at Line 2.

Now suppose $\codeg(y_1, v_1) = 1$ and $\surp_{G'}(I) = -2$. If $|I| \geq 3$ then $S(G') \geq 1$ as desired, so suppose $|I| = 2$. By \Cref{6base-obs2}, we have $\codeg(v_2, v_3) \leq 2$, so there must be some vertex $x \in I$ and $v_j \in \{v_2, v_3 \}$ with $x \not \sim v_j$; say without loss of generality $x \not \sim v_3$ and $x$ is isolated in $G'$.

Let $d = \deg_G(x)$. By \Cref{6base-obs2} we have $\codeg(x, y_1) \leq d - 3$ and $\codeg(x, v_1) \leq 2$. So the only way for $x$ to be isolated in $G'$ is to have $x \sim v_2,  \codeg(x, y_1) = d-3, \codeg(x, v_1) = 2$. If $d \leq 5$, we thus have $S(v_2, N[v_1]) \geq 1$ due to subquadratic vertex $x$.  If $d = 6$, then by \Cref{gfgft0} we have $x \not \sim y_2$ (else the 6-vertex $x$ would be shared by $y_2, v_2$); in this case, $G - N[x]$ has a 3-vertex $y_1$ with a triangle $y_1, y_2, u$, and we would branch on $x$ at Line 2.
\end{proof}

\smallskip

There are now a few different cases for the analysis.

\smallskip

\noindent \textbf{Case I: $\boldsymbol{y_1 \not \sim y_2}$.} Split on $y_1, y_2, v_1$, getting subproblems $G_1 = G - N[y_1], G_2 = G - y_1 - N[y_2], G_3 = G - \{y_1, y_2 \} - N[v_1], G_4 = G - \{y_1,y_2, v_1 \}$.  By \Cref{6base-obs2}, we have $\minsurp(G_1) \geq 2$ and $\minsurp(G_2) \geq 1$ and $\minsurp(G_3) \geq \shad(y_1, v_1) - 1 \geq 1$ and $\minsurp(G_4) \geq \shad(y_1, y_2) - 1 \geq 1$. So $G_1, G_2, G_3, G_4$ have drops $(2.5,6), (3,7), (3,7), (1.5,3)$ directly. We also apply \Cref{branch5-5} to 2-vertex $u$ in  $G_4$, where $\deg(v_2) + \deg(v_3) - 1 - \codeg(v_2, v_3) \geq 7$ by \Cref{6base-obs2}.  So we need:
\begin{align}
e^{-2.5a-6b} + e^{-3a-7b} + e^{-3a-7b} + e^{1.5a-3b} \psi_7 \leq 1 \tag{6-15}
\end{align}

\noindent \textbf{Case II: $\boldsymbol{y_1 \sim y_2}$.} We choose some enumeration $v_1, v_2, v_3$, and split on $v_1, v_2$;  in the subproblem $G - \{v_1, v_2 \}$, we apply (B-Fun2) to $u, y_1$. This generates subproblems $G_1 = G - N[v_1], G_2 = G - v_1 - N[v_2], G_3 = G \{v_1, v_2 \} - N[v_3, y_1], G_4 = G - \{v_1, v_2, y_1 \}$. 

By \Cref{6base-obs2}, $G_1, G_2, G_4$ have drops $(2,5), (2.5,6), (1.5,3)$ directly; furthermore, we can apply \Cref{branch5-5} to the 2-vertex $u$ in $G_4$, where $r = \deg(v_1) + (\deg(y_2) - 1) - \codeg(y_2, v_1) - 1 \geq 5 + 5 - 2 - 1 = 7$ by \Cref{6base-obs2}. 

\Cref{gfgft1} and \Cref{6base-obs2} give  $\minsurp(G_3) \geq -3 + \codeg(y_1, v_3) $ and $\codeg(y_1, v_3) \leq 2$. In light of \Cref{gfgft0}, there are two cases to select the ordering:

\smallskip

\noindent \textbf{Case IIa: $\boldsymbol{S(v_1, N[v_2]) \geq 1.}$} In this case $G_2$ has drop $(2.5,7)$ after simplification. Since $G_3$ has drop $(4,11)$, we need:
\begin{align}
&e^{-2a-5b} +  e^{-2.5a-7b} + e^{-4a-11b} + e^{-1.5a-3b} \psi_7 \leq 1 \tag{6-16}
\end{align}

\noindent \textbf{Case IIb:  $\boldsymbol{\codeg(y_1, v_3) = 1.}$} By \Cref{gfgft1}, either $S(v_{j},N[v_3]) \geq 1$ for some $j$ (which would put us into Case IIa, with some other vertex ordering), or $S(G_3) \geq 1$, or $\minsurp(G_3) \geq -1$. In the latter two cases, $G_3$ has drop $(4,13)$ or $(4.5,12)$. So we need:
\begin{align}
&e^{-2a-5b} + e^{-2.5a-6b} + \max \{ e^{-4a-13b}, e^{-4.5a-12b} \} + e^{-1.5a-3b} \psi_7  \leq 1 \tag{6-17}
\end{align}

\section{Branching on 7-vertices and finishing up.}
We first describe the degree-7 branching algorithm.
 
\begin{algorithm}[H]
Simplify $G$

Branch on a vertex $u$ with degree larger than 7

Split on a 7-vertex $u$ with $\shad(N[u]) \geq 0$.

Branch on a 7-vertex $u$ which has a blocker $x$ with $\deg(x) \geq 5$.

Branch on a vertex $x$ which is a blocker of three 7-vertices $u_1, u_2, u_3$.

Run either  {\tt Branch6-3} or the MaxIS-7 algorithm.

\caption{Function {\tt Branch7}$(G, k)$}
\end{algorithm}

\begin{theorem}
\label{branch7-thm}
 {\tt Branch7} has measure $a_7 \mu + b_7 k$ for $a_7 = 0.01266, b_7 = 0.221723$.
 \end{theorem}
 
 To show this, let us examine the algorithm line by line.
 
 \bigskip
 
\smallskip

\noindent  \textbf{Line 2.} Here, we apply \Cref{branch5-3}; we need:
 \begin{align}
& \max\{e^{-a - 3 b} + e^{-a - 4 b}, e^{-0.5 a - 2 b} + e^{-2 a - 5 b}, e^{-0.5 a - b} + e^{-2.5 a - 8b } \} \leq 1 \tag{7-1} \label{t71}
\end{align}

\noindent \textbf{Line 3.} In this case $G - N[u]$ has drop $(3,7)$. So we need:
 \begin{align}
 e^{-0.5 a - b} +  e^{-3 a - 7 b} \leq 1 \tag{7-2}
 \end{align}
 
\smallskip

 \noindent \textbf{Line 4.} Apply (B-Block) to $u,x$, getting subproblems $G - N[x], G - \{u, x \}$ with drops $(1,5)$ and $(1,2)$ respectively. We need:
 \begin{align}
 e^{-a - 2 b} + e^{-a-5b} \leq 1 \tag{7-3}
 \end{align}
 
\smallskip

 \noindent \textbf{Line 5.}  Apply (B-Block) to $u_1, u_2, u_3, x$. The resulting subproblems $G - \{u_1, u_2, u_3, x \}$ and $G - N[x]$ have drops $(1,4), (1,3)$ directly; the branch-seq $[(1,3),(1,4)]$ is already covered in (\ref{t71}). 

\smallskip

 \noindent \textbf{Line 6.} At this point,  every 7-vertex $u$ has a blocker $x$ (else we would branch on $u$ at Line 3), where  $\deg(x) \leq 4$ (else we would branch on $u,x$ at Line 4). Due to Line 5, each vertex $x$ serves as a blocker to at most two 7-vertices $u, u'$.  Thus, at this stage, the graph satisfies
 $$
 n_3 + n_4 \geq n_7/2
 $$

We now use a more refined property of the MaxIS-7 algorithm of \cite{XiaoN17}; specifically, it has runtime $O^*(1.19698^{w_3 n_3 + w_4 n_4 + w_5 n_5 + w_6 n_6 + n_7})$ for constants $w_3 = 0.65077, w_4 = 0.78229, w_5 = 0.89060, w_6 = 0.96384$. Since $n_3 + n_4 \geq n_7/2$, it can be mechanically checked that $w_3 n_3 + w_4 n_4 + w_5 n_5 + w_6 n_6 + n_7$ is maximized at $n_6 = n, n_3 = n_4 = n_7 = 0$. Accordingly, we need to show
\begin{align}
&\min\{ a_{63} \mu + b_{63} k, 2 (k - \mu) \cdot (w_6 \log (1.19698))  \} \leq a \mu + b k \tag{7-4}
\end{align}
which can also be checked mechanically. This shows \Cref{branch7-thm}.

\bigskip

To finish up, we use a simple degree-8 branching algorithm.

 \begin{algorithm}[H]
If $\maxdeg(G) \leq 7$, then use \Cref{lem:combinelem} with the algorithms  {\tt Branch7} or MaxIS-7.

Otherwise, split on a vertex of degree at least $8$.
\caption{Function {\tt Branch8}$(G, k)$}
\end{algorithm}
\begin{theorem}
\label{thm:vc-thm}
{\tt Branch8} has runtime $O^*(\dSeven^k)$.
\end{theorem}
\begin{proof}
Splitting on a vertex of degree $r \geq 8$ generates subproblems with vertex covers of size at most $k -1 $ and $k -8$. These have cost $\dSeven^{k-1}$ and $\dSeven^{k-8}$ by induction. Since $\dSeven^{-1} + \dSeven^{-8}  < 1$, the overall cost is $O^*(\dSeven^k)$. Otherwise,  \Cref{lem:combinelem} (with parameters $a_7, b_7$ and $c = \log(1.19698)$) gives the desired algorithm with runtime $O^*(\dSeven^k)$.
\end{proof}

\appendix

\section{Proofs for basic results}

\label{sec:prelimproof}

\begin{proof} [Proof of \Cref{lambda-surp} and \Cref{lambda-surpxx}]
Consider the linear program defined by $\LPVC(G)$ and the additional constraints $\theta(v) = 0: v \in I$ and $\theta(v) \geq 1/2: v \notin X$, and let its optimum value be denoted by $\lambda_{I,X}$. We claim that any indset $J$ of minimum surplus with $I \subseteq J \subseteq X$ corresponds to an optimal basic solution to this linear program, and has $\lambda_{I, X} = \tfrac{1}{2}( n + \surp_G(J) )$. This will show \Cref{lambda-surpxx} (since the linear program can be solved in polynomial time), and it will also show \Cref{lambda-surp} (since $\lambda(G) = \lambda_{\emptyset, V} = \minsurp^0(G)$).

First, we claim that any such basic solution $\theta \in [0,1]^V$ is half-integral. For, suppose not, and let $A, B$ denote the vertices $v$ with $\theta(v) \in (0,1/2)$ and $\theta(v) \in (1/2,1)$ respectively. It can be seen that modifying $\theta(v): v \in A$ by $\pm \epsilon$ and $\theta(v): v \in B$ by $\mp \epsilon$ preserves feasibility. Thus, in a basic solution, we must have $A = B = \emptyset$.

Now consider any such basic solution $\theta \in \{0, \tfrac{1}{2},1 \}^{V}$. Let $I_0 = \theta^{-1}(0)$ and $I_1 = \theta^{-1}(1)$; note that $I_0$ is an indset (possibly empty) and $N(I_0) \subseteq I_1$. So $\sum_{v} \theta(v) = |I_1| + \tfrac{1}{2} (n - |I_0| - |I_1|)  \leq (n + |N(I_0)| - |I_0|)/2$. On the other hand, for any min-set $J$ with $I \subseteq J \subseteq X$, then setting $\theta(v) = 1$ for $v \in N(J)$ and $\theta(v) = 0$ for $v \in J$ and $\theta(v) = 1/2$ for all other vertices $v$ gives a fractional solution cover with $\sum_v \theta(v) = (n + \surp_G(J))/2$.
\end{proof}

\begin{proof}[Proof of \Cref{s-prop1}]
Let $I$ be a critical-set; we first claim there is a good cover $C$ with $I \subseteq C$ or $I \cap C = \emptyset$.
For, suppose that $I \cap C = J$ where $\emptyset \neq J \subsetneq I$. In this case, we must also have $N(I \setminus J) \subseteq C$, so overall $|C \cap N[I]| \geq |J| + |N(I \setminus J)| = |J| + |I \setminus J| + \surp(I \setminus J)$; since $I \setminus J$ is a non-empty subset of $I$, we have $\surp(I \setminus J) \geq \surp(I)$, hence this is at least $|J| + |I \setminus J| + \surp(I) = |I| + \surp(I) \geq N(I)$. Thus, replacing $C \cap N[I]$ by $N(I)$ would give a vertex cover $C'$ with $C' \cap I = \emptyset$ and $|C'| \leq |C|$. 

Furthermore, if $\surp(I) \leq 0$, then consider a good cover $C$ with $I \subseteq C$. Then replacing $I$ with $N(I)$ would give another cover $C'$ with $|C'| \leq |C| + |N(I)| - |I| \leq |C|$ and $C' \cap I = \emptyset$.

Finally, suppose $\surp(I) = 1$ and let $J = N(I)$. If $G$ has a good cover with $I \cap C = \emptyset$, then $J \subseteq C$ as desired. Otherwise, suppose $G$ has  a good cover with $I \subseteq C$. Then, $C'  = (C \setminus I) \cup J$ would be another cover of size $|C|- |I| + (|J| - |C \cap J|)$; since $|J| = |I|+1$, this is $|C| + 1 - |C \cap J|$. If $C \cap J \neq \emptyset$, this would be a good cover with $I \cap C = \emptyset$.  
 \end{proof}

 \begin{proof}[Proof of \Cref{p1p2prop}]
Let us first consider P1. If $G'$ has a good cover $C'$, then $C' \cup N(I)$ is a cover for $G$ of size $k$. Conversely, by \Cref{s-prop1}, there is a good cover $C$ of $G$ with $C \cap I = \emptyset$, and then then $C \setminus N(I)$ is a cover of $G'$ of size $k - |N(I)| = k'$. Similarly, if $G'$ has a fractional cover $\theta$, then we can extend to $G$ by setting $\theta(v) = 0$ for $v \in I$ and $\theta(v) = 1$ for $v \in N(I)$,  with total weight $\lambda(G') + (k - k')$. So $\mu(G') = k' - \lambda(G') \leq k' - (\lambda(G) - (k - k')) = k - \lambda(G) = \mu(G)$. 

Next, for rule P2, let $J = N(I)$.  Since $I$ is a critical-set with surplus one, it cannot contain any isolated vertex. Given any good cover $C$ of $G$ with $J \subseteq C$, observe that $(C \setminus J) \cup \{ y \}$ is a cover of $G'$ of size $k - |J| + 1 = k'$. Likewise, given a good cover $C$ of $G$ with $J \cap C = \emptyset$, we have $N(J) \subseteq C$; in particular, $I \subseteq C$ since $I$ has no isolated vertices, so $C \setminus I$ is a cover of $G'$ of size $k - |I| = k'$. Conversely, consider a good cover $C'$ of $G'$; if $y \in C'$, then $(C' \cup J) \setminus \{y \}$ is a cover of $G$, of size $k' + |J| - 1 = k$. If $y \notin C'$, then $C' \cup I$ is a cover of $G$, of size $k' + |I| = k$.  Similarly, if $G$ has an optimal fractional cover $\theta$, it can be extended to $G$ by setting $\theta(v) = \theta(y)$ for $v \in J$, and $\theta(v) = 1 - \theta(y)$ for $v \in I$. This has weight $\lambda(G') + |I| (1 - \theta(y)) + |J| \theta(y) - \theta(y) = \lambda(G') + |I| \geq \lambda(G)$. So $\mu(G') = k' -  \lambda(G')  \leq (k-|I|) - (\lambda(G) - |I|) = \mu(G)$.

Now for P3, let $A = N_G(u) \cap N_G(x)$ and $B_x = N_G(x) \setminus N_G[u]$ and $B_u = N_G(u) \setminus N_G[x]$.   To show the validity of the preprocessing rule, suppose $\langle G, k \rangle$ is feasible, and let $C$ be a good cover of $G$ with either $|C \cap \{u, x \}| = 1$; then $C \setminus ( \{u,x \} \cup A) $ is a cover of $G'$ of size $k'$. Conversely, suppose $\langle G', k' \rangle$ is feasible with a good cover $C'$.  So $|B_u \setminus C'| \leq 1$. if $B_u \setminus C' = \emptyset$, then $C = C' \cup A \cup \{x \}$ is a cover of $G$. If $|B_u \setminus C'| = 1$, then necessarily $B_x \subseteq C'$ and $C = C' \cup A \cup \{ u \}$ is a cover of $G$.

Finally, we claim that (P3) satisfies $\mu(G') \leq \mu(G)$; since $k' = k -1 - |A|$, it suffices to show that $\lambda(G) \leq \lambda(G') + 1 + |A|$. For, take an optimal fractional cover $\theta$ for $G'$. To extend it to a fractional cover for $G$, we set $\theta(v) = 1$ for $v \in A$ and there are a few cases to determine the values for $\theta(u)$ and $\theta(v)$. If $B_u = \emptyset$, we set $\theta(u) = 0, \theta(x) = 1$. Otherwise, if $B_u \neq \emptyset$, let $b_u = \min_{v \in B_u} \theta'(v)$ and then set $\theta(u) = 1 - b_u, \theta(x) = b_u$; since $G'$ has a biclique between $B_u$ and $B_x$, this covers any edge between $x$ and $v \in B_x$ with $\theta(x) + \theta(v) = b_u + \theta(v) \geq 1$. Overall, $\theta$ is a fractional vertex cover of $G$ of weight $\sum_{v \in G'} \theta(v) + \theta(u) + \theta(x) + |A| = \lambda(G') + 1 + \codeg(u,x)$.
 \end{proof}

\begin{proof}[Proof of \Cref{ppthm22}]
Here $\minsurp(G) \geq 2$. If $J$ is a min-set of $G - X$, then $2 \leq \surp_G(J) \leq \surp_{G-X}(J) + |X|$. So $\surp_{G-X}(J) \geq 2 - |X|$.  Similarly, if $J$ is a min-set of $G - N[I]$, then $2 \leq \surp_G(I \cup J) = \surp_{G - N[I]}(J) + \surp_{G}(I)$. 
\end{proof}

\begin{proof}[Proof of \Cref{simp-obs0}]
If $\deg(x) = \deg(y) = 1$ and $\codeg(x,y) = 0$, then $\{x,y \}$ is a min-set with $\surp( \{x,y \}) = 0$, and applying P1 reduces $k$ by two. Otherwise, if $\deg(x) \geq 2$, then we can apply (P2) to $x$, and vertex $y$ remains subquadratic, and we can follow up by applying (P1) or (P2) again to $y$. Note that if $\minsurp(G) \geq 0$, then at least one of the conditions (i) or (ii) must hold.
\end{proof}

\begin{proof}[Proof of \Cref{simp-obs02}]
We show it by induction on $\ell$. The base case $\ell = 0$ is vacuous. For the induction step, we apply (P2) to $x_{\ell}$, forming a graph $G'$ where the neighbors $y,z$ get contracted into a single new vertex $t$. By hypothesis, $y,z$ are distinct from $x_1, \dots, x_{\ell}$.  We claim that $\dist_{G'}(x_i, x_j) \geq 3$ and $\deg_{G'}(x_i) \geq 2$ for any pair $i < j < \ell$. For, clearly $x_i \not \sim x_j$ in $G'$. If $x_i, x_j$ share a neighbor in $G'$, it must be vertex $t$ as no other vertices were modified. This is only possible if $\{x_i, x_j \} \sim \{y, z \}$; but this contradicts our hypothesis that $\codeg(x_i, x_{\ell}) = \codeg(x_j, x_{\ell}) = 0$. Finally, suppose that $\deg_{G'}(x_i) \leq 1$. This is only possible if the two neighbors of $x_i$ got merged together, i.e. $x_i \sim y$ and $x_i \sim z$. Again, this contradicts that $\codeg(x_i, x_{\ell}) = 0$.

By induction hypothesis applied to $G'$, we have $S(G') \geq \ell - 1$ and hence $S(G) \geq \ell$.
\end{proof}

\bibliography{sat}
\bibliographystyle{alpha}

\end{document}